\documentclass[11pt]{article}

\title{\textbf{Deterministic Massively Parallel Connectivity}%
\thanks{Research supported in part by the Centre for Discrete Mathematics and its Applications (DIMAP), by EPSRC award EP/V01305X/1, by an EPSRC studentship, by a Weizmann-UK Making Connections Grant, and by an IBM Award.}
}

\author{
\textbf{Sam Coy}\thanks{Research supported in part by an EPSRC studentship and by the Centre for Discrete Mathematics and its Applications (DIMAP).}\\ Department of Computer Science\\ Centre for Discrete Mathematics and its Applications\\ University of Warwick\\S.Coy@warwick.ac.uk
    \and
\textbf{Artur Czumaj}\thanks{Research supported in part by the Centre for Discrete Mathematics and its Applications (DIMAP), by EPSRC award EP/V01305X/1, by a Weizmann-UK Making Connections Grant, and by an IBM Award.}\\ Department of Computer Science\\ Centre for Discrete Mathematics and its Applications\\ University of Warwick\\A.Czumaj@warwick.ac.uk
}

\author{
\textbf{Sam Coy}\thanks{Department of Computer Science and Centre for Discrete Mathematics and its Applications (DIMAP), University of Warwick. Email: S.Coy@warwick.ac.uk. Research supported in part by an EPSRC studentship and by the Centre for Discrete Mathematics and its Applications (DIMAP).}
    \and
\textbf{Artur Czumaj}\thanks{Department of Computer Science and Centre for Discrete Mathematics and its Applications (DIMAP), University of Warwick. Email: A.Czumaj@warwick.ac.uk. Research supported in part by the Centre for Discrete Mathematics and its Applications (DIMAP), by EPSRC award EP/V01305X/1, by a Weizmann-UK Making Connections Grant, and by an IBM Award.}
}

\author{
\textbf{Sam Coy} \qquad \textbf{Artur Czumaj} \\
Department of Computer Science \\
Centre for Discrete Mathematics and its Applications \\
University of Warwick
}

\date{}
\date{\today}

\usepackage[T1]{fontenc}\usepackage{yfonts}

\usepackage{amsmath,amssymb,amsfonts}
\usepackage{xspace}
\usepackage{lettrine}
\usepackage{graphicx}
\usepackage{subfigure}
\usepackage{latexsym}
\usepackage{graphicx}
\usepackage[scaled=0.92]{helvet} 
\usepackage{fancybox}
\usepackage{fullpage}
\usepackage{color,xcolor}
\usepackage{paralist}
\usepackage{mathcmd}
\usepackage{bm}

\usepackage[ruled,linesnumbered,lined,boxed,commentsnumbered]{algorithm2e}

\usepackage{ifpdf}
\newcommand{\mydriver}{hypertex}
\ifpdf
 \renewcommand{\mydriver}{pdftex}
\fi
\usepackage[breaklinks,\mydriver]{hyperref}

\usepackage{cleveref}

\usepackage{framed}
\usepackage[framemethod=tikz]{mdframed}

\newcommand{\COMMENTED}[1]{{}}
\newcommand{\junk}[1]{{}}

\newcount\shortyear\newcount\shorthour\newcount\shortminute
\shorthour=\time\divide\shorthour by 60\shortyear=\shorthour
\multiply\shortyear by 60\shortminute=\time\advance\shortminute by -\shortyear
\shortyear=\year\advance\shortyear by -1900

\def\zeit{\number\shorthour:\ifnum\shortminute<10 0\number\shortminute
\else\number\shortminute\fi}

\makeatletter

\newcommand{\Bigepsilon}{\ensuremath{{\cal E}}}
\newcommand{\EPS}{\Bigepsilon}

\newtheorem{theorempart}{\bf Theorem}[section]
\newtheorem{lemmapart}[theorempart]{\bf Lemma}
\newtheorem{observationpart}[theorempart]{\bf Observation}
\newtheorem{factpart}[theorempart]{\bf Fact}
\newtheorem{corollarypart}[theorempart]{\bf Corollary}
\newtheorem{claim}[theorempart]{\bf Claim}
\newtheorem{remark}[theorempart]{Remark}

\newenvironment{theorem}{\begin{theorempart}\sl}{\end{theorempart}}
\newenvironment{lemma}{\begin{lemmapart}\sl}{\end{lemmapart}}

\newenvironment{proof}{\noindent{\bf Proof:\ } \small}{\hfill$\Box$\par\medskip}
    \renewenvironment{proof}{\noindent{\bf Proof:\ }}{\hfill$\Box$\par\medskip}
\newenvironment{proofof}[1]{\par\medskip\noindent{\bf Proof of #1:\ } }{\hfill$\Box$\par\medskip}

\newtheorem{conjpart}{Conjecture}
\newenvironment{conjecture}{\begin{conjpart}\rm}{\end{conjpart}}

\newtheorem{defpart}{\bf Definition}[section]
\newenvironment{definition}{\begin{defpart}\rm}{\end{defpart}}

\newenvironment{ftheorem}{\begin{theorempart}\begin{mdframed}[hidealllines=true,backgroundcolor=gray!25]\sl}{\end{mdframed}\end{theorempart}}
\newenvironment{flemma}{\begin{lemmapart}\begin{mdframed}[hidealllines=true,backgroundcolor=gray!25]\sl}{\end{mdframed}\end{lemmapart}}

\newenvironment{fcorollary}{\begin{corollarypart}\begin{mdframed}[hidealllines=true,backgroundcolor=gray!25]\sl}{\end{mdframed}\end{corollarypart}}
\newenvironment{fdefinition}{\begin{defpart}\begin{mdframed}[hidealllines=true,backgroundcolor=gray!15]\rm}{\end{mdframed}\end{defpart}}

\newcommand{\showtheorem}[2]{
\begin{mdframed}[hidealllines=true,backgroundcolor=gray!10]
\medskip\noindent\textbf{Main Theorem (#1).} \sl \textcolor[rgb]{0.00,0.00,0.55}{#2}
\end{mdframed}\smallskip}

\newcommand{\congc}{\textsf{Congested Clique}\xspace}
\renewcommand{\congc}{\text{Congested Clique}\xspace}
\renewcommand{\congc}{\text{Congested Clique}\xspace}
\newcommand{\local}{\textcolor[rgb]{0.50,0.00,0.00}{\textsf{LOCAL}}\xspace}
\renewcommand{\local}{\text{LOCAL}\xspace}
\newcommand{\LOCAL}{\local}
\newcommand{\MapReduce}{\textsf{MapReduce}\xspace}
\renewcommand{\MapReduce}{\textcolor[rgb]{0.50,0.00,0.00}{\textsf{MapReduce}}\xspace}
    \renewcommand{\MapReduce}{\text{MapReduce}\xspace}
\newcommand{\MPC}{\textsf{\textcolor[rgb]{0.50,0.00,0.00}{MPC}}\xspace}
    \renewcommand{\MPC}{MPC\xspace}
\newcommand{\PRAM}{\textsf{\textcolor[rgb]{0.50,0.00,0.00}{PRAM}}\xspace}
    \renewcommand{\PRAM}{PRAM\xspace}

\newcommand{\ALGs}{\textrm{\small ALG}{\ensuremath{^*}}\xspace}

\newcommand{\machines}{\ensuremath{\mathcal{\textcolor[rgb]{0.50,0.00,0.00}{M}}}\xspace}
\newcommand{\spac}{\ensuremath{\mathcal{\textcolor[rgb]{0.50,0.00,0.00}{S}}}\xspace}
\newcommand{\gspac}{\ensuremath{\mathcal{S}_{\mathcal{G}}}\xspace}
\renewcommand{\gspac}{\ensuremath{\textcolor[rgb]{0.50,0.00,0.00}{\mathcal{S}_{Global}}}\xspace}

\newcommand{\DISTR}[2]{\ensuremath{\mathcal{D}_{#1,#2}}\xspace}
\newcommand{\DISTRnp}{\ensuremath{\DISTR{n}{p}}\xspace}
\newcommand{\Hl}{\ensuremath{\mathbb{H}_{\ell}}\xspace}
\newcommand{\HH}{\ensuremath{\mathcal{H}}\xspace}

\newcommand{\chunk}{\ensuremath{\chi}\xspace}
\newcommand{\bb}{\ensuremath{\mathbf{b}}\xspace}
\newcommand{\cost}{\ensuremath{\text{cost}}\xspace}
\newcommand{\Tcost}{\ensuremath{\text{Tcost}}\xspace}

\newcommand{\ConnectTwoHops}{\textsf{\small Connect2Hops}\xspace}
\renewcommand{\ConnectTwoHops}{\textsf{\small \textcolor[rgb]{0.50,0.00,0.25}{Connect2Hops}}\xspace}
\newcommand{\RelabelIntraLevel}{\textsf{\small RelabelIntraLevel}\xspace}
\renewcommand{\RelabelIntraLevel}{\textsf{\small \textcolor[rgb]{0.50,0.00,0.25}{RelabelIntraLevel}}\xspace}
\newcommand{\RelabelInterLevel}{\textsf{\small RelabelInterLevel}\xspace}
\renewcommand{\RelabelInterLevel}{\textsf{\small \textcolor[rgb]{0.50,0.00,0.25}{RelabelInterLevel}}\xspace}

\newcommand{\Artur}[1]{{}}
\renewcommand{\Artur}[1]{{{\footnote{$\cal A.C.$ --- \textcolor[rgb]{0.50,0.00,1.00}{#1}}}}}

\newcommand{\ArturKeep}[1]{{}} 

\newcommand{\Sam}[1]{{}}
\renewcommand{\Sam}[1]{{{\footnote{$Sam$ --- \textcolor[rgb]{0.00,0.07,1.00}{#1}}}}}

\newcommand{\revised}[1]{\textcolor[rgb]{0.00,0.59,0.00}{#1}}

\newcommand{\NAT}{\ensuremath{\mathbb{N}}}
\newcommand{\NATURAL}{\NAT}
\newcommand{\NN}{\NATURAL}

\renewcommand{\Pr}[1]{\ensuremath{\mathbf{Pr}[#1]}}
\newcommand{\PPr}[1]{\ensuremath{\mathbf{Pr}\big[#1\big]}}
\newcommand{\PPPr}[1]{\ensuremath{\mathbf{Pr}\Big[#1\Big]}}

\newcommand{\SPPPr}[2]{\ensuremath{\mathbf{Pr}_{#1}\Big[#2\Big]}}

\newcommand{\Ex}[1]{\ensuremath{\mathbb{E}[#1]}}
\newcommand{\EEx}[1]{\ensuremath{\mathbb{E}\big[#1\big]}}
\newcommand{\EEEx}[1]{\ensuremath{\mathbb{E}\Big[#1\Big]}}
\newcommand{\SEx}[2]{\ensuremath{\mathbb{E}_{#1}[#2]}}

\def\epsilon{\ensuremath{\varepsilon}}
\newcommand{\eps}{\ensuremath{\epsilon}\xspace}

\newcommand{\poly}{\operatorname{\textrm{poly}}}
\newcommand{\polylog}{\operatorname{\textrm{polylog}}}

\newcommand{\etal}{et al.\ }

\begin{document}

\begin{titlepage}
\maketitle
\thispagestyle{empty}

\begin{abstract}
We consider the problem of designing fundamental graph algorithms on the model of Massive Parallel Computation (\MPC). The input to the problem is an undirected graph $G$ with $n$ vertices and $m$ edges, and with $D$ being the maximum diameter of any connected component in $G$. We consider the \MPC with \emph{low local space}, allowing each machine to store only $\Theta(n^{\delta})$ words for an arbitrarily constant $\delta>0$, and with linear global space (which is equal to the number of machines times the local space available), that is, with optimal utilization.

In a recent breakthrough, Andoni \etal (FOCS'18) and Behnezhad \etal (FOCS'19) designed parallel randomized algorithms that in $O(\log D + \log\log n)$ rounds on an \MPC with low local space determine all connected components of an input graph, improving upon the classic bound of $O(\log n)$ derived from earlier works on \PRAM algorithms.

In this paper, we show that asymptotically identical bounds can be also achieved for deterministic algorithms: we present a deterministic \MPC low local space algorithm that in $O(\log D + \log\log n)$ rounds determines all connected components of the input graph.

\end{abstract}
\end{titlepage}




\section{Introduction}
\label{sec:introduction}

Motivated largely by recent very successful development of a number of massively parallel computation frameworks, such as \MapReduce \cite{DG08}, Hadoop \cite{White15}, Dryad \cite{IBYBF07}, or Spark \cite{ZCFSS10}, we have seen in the last decade a booming research interest in the design of parallel algorithms. While some of this research has been very applied, there has been also an increasing interest in fundamental and algorithmic research aimed to understand the computational power of such systems. In this paper we study the complexity of some fundamental graph problems on the \emph{Massively Parallel Computation (\MPC)} model, first introduced by Karloff \etal \cite{KSV10}, which is now the standard theoretical model of algorithmic study, as it provides a clean abstraction of these practical frameworks, see, e.g., \cite{ANOY14,BKS13,BKS17,GSZ11,KSV10}.

The \MPC model can be seen as a variant of the classical BSP model which balances the parallelism (like in the classical \PRAM model) with the communication costs of distributed systems. An \MPC is a parallel system with some number \machines of \emph{machines}, each machine having \spac words of its \emph{local memory}. \MPC computations are synchronous, with alternating rounds of unlimited local computation, and communication of up to \spac data per processor. In each synchronous round each machine processes its local data and performs an arbitrary local computation. At the end of each round, machines exchange messages. Each message is sent only to a single machine specified by the machine that is sending the message. All messages sent and received by each machine in each round, as well as the output have to fit into the machine's local space \spac.

It is known that a single step of \PRAM can be simulated in a constant number of rounds on \MPC \cite{GSZ11,KSV10}, implying that many \PRAM algorithms can be implemented in the \MPC model without any asymptotic loss of complexity. However, it can be also seen that the \MPC model is often more powerful as it allows for a lot of local computation (in principle, unbounded). For example, any data stored on a single machine can be processed in a single round, even if the underlying local problem is computationally hard. Using this observation, we have recently seen a stream of papers demonstrating that a number of fundamental graph problems of connectivity, matching, maximal independent set, vertex cover, coloring, etc (see, e.g., \cite{ASZ20,ABBMS17,BBDFHKU19,BDELM19,BDELMS19,BHH19,CFGUZ19,CLMMOS18,GGKMR18,GU19,LMOS20,LMSV11,Onak18}) can be solved on the \MPC model significantly faster than on \PRAM. However, the common feature of most of these results is that they were relying on randomized algorithms, and very limited amount of research has been done to study deterministic algorithms. This leads to the main theme of this work: \emph{to explore the power of the \MPC model in the setting of deterministic algorithms}. The underlying question is in what circumstances one can achieve (asymptotically) the same complexity bounds for deterministic \MPC algorithms as those known for their randomized counterparts.


In this paper we consider graph problems, where the input is a graph $G = (V,E)$ with $n$ vertices and $m$ edges. For such a problem, the input (collection $V$ of nodes and $E$ of edges) is initially arbitrarily distributed among the machines, with \emph{global space} $\gspac = \spac \cdot \machines = \Omega(n+m)$. We focus on the design of (most desirable) \emph{low space} (fully scalable) \MPC algorithms for graph problems with optimal space utilization, that is, we assume that the local space of each machine is strongly sublinear in the number of nodes, i.e., for an arbitrarily small constant $\delta > 0$, we have $\spac = O(n^{\delta})$, and the global space $\gspac$, which is the product $\spac \cdot \machines$, is \emph{linear}, that is, $\gspac = O(n+m)$.

The low-space regime is particularly challenging due to the fact that a single vertex cannot necessarily store all incident edges on a single machine, and so these are scattered over several machines. For example, already the very simple problem of distinguishing between a cycle of length $n$ and two cycles of length $\frac{n}{2}$ seems to be hard for this model, and it is conjectured to require $\Omega(\log n)$ \MPC computation rounds. Despite these challenges, it has been recently shown that this setting allows for very efficient algorithms for some fundamental graphs problems. One of the recent highlights in this area has been a sequence of works on the connectivity problem (see \cite{ASSWZ18,ASW19,BDELM19,LMW18,LTZ20}), initiated with a seminal work of Andoni \etal \cite{ASSWZ18} and culminating with a randomized \MPC algorithm that for a given graph $G$ with each connected component having the diameter at most $D$, determines all connected components in $O(\log D + \log\log n)$ \MPC rounds, with high probability due to Behnezhad \etal \cite{BDELM19}. While this algorithm matches the conjectured bound for the 1-vs-2 cycle problem mentioned above for $D = \Omega(n)$, it significantly improves the complexity for low diameter graphs (e.g., for random graphs we have $D = O(\log n)$). Furthermore, this complexity bound seems to be almost optimal, since for $D \ge \log^{1+\Omega(1)}n$, Behnezhad \etal \cite{BDELM19} show that an $o(\log D)$-rounds algorithm would imply the existence of an $o(\log n)$-rounds algorithm for the 1-vs-2 cycle problem.

Let us note that while our paper focuses on the low local space \MPC case when $\spac = O(n^{\delta})$ and $\gspac = O(n+m)$, there have been some earlier research studying the setting either with $\spac = \Omega(n)$ or when $\gspac = \omega(n+m)$. For example, Lattanzi \etal \cite{LMSV11} gave a constant-round \MPC algorithm for connectivity in the $\spac = n^{1+\Omega(1)}$ local space regime, and by exploring the relationship between \MPC with $\spac = O(n)$ and the \congc model, the \congc $O(1)$-rounds algorithm for connectivity due to Jurdzi\'{n}ski and Nowicki \cite{JN18} yields a $O(1)$-rounds connectivity \MPC algorithm with $\spac = O(n)$ and $\gspac = O(n+m)$; very recently Nowicki \cite{Nowicki21} extended this result and obtained a \emph{deterministic} $O(1)$-rounds connectivity \MPC algorithm with $\spac = O(n)$ and $\gspac = O(n+m)$. For $\spac = O(n^{\delta})$ and large \gspac (say, $\gspac = O(n^3)$) it is easy to solve the connectivity problem in $O(\log D)$ rounds using a simple squaring approach, but the (conditional) lower bound of Behnezhad \etal \cite{BDELM19} mentioned above shows that it seems unlikely to improve this bound, even if $\gspac \gg n + m$. Finally, 
a very recent related work \cite{CMT21} shows that one can \emph{slightly reduce the randomness} used in \cite{BDELM19}: $(\log n)^{O(\log D + \log\log_{m/n}n)}$ random bits suffice to obtain an \MPC algorithm with local space $\spac = O(n^{\delta})$ and $\machines = O(n+m)$ machines (and so with global space $O((n+m) \cdot n^{\delta})$, which is larger than $O(n+m)$ in \cite{BDELM19}) that determines graph connectivity in $O(\log D + \log\log_{m/n}n)$ rounds with probability $1-1/\poly((m \log n)/n)$.


\subsection{Our contribution}

In this paper we show that the recent randomized \MPC connectivity algorithms (with $\spac = O(n^\delta)$ local space and $\gspac=O(m+n)$ global space) can be derandomized without asymptotic loss of complexity. We inspects details of the algorithmic approach due to Andoni \etal \cite{ASSWZ18} and Behnezhad \etal \cite{BDELM19} and abstract from it two key algorithmic steps that use randomization. The two steps rely on some very simple random sampling, which is easy to be implemented by a randomized $O(1)$-rounds \MPC algorithm but which seems to be difficult to be implemented deterministically, especially in our setting of low space \MPC{}s with optimal global space utilization. In fact (see \Cref{sec:derandomization} for more details), because of a recently developed lower bound framework relating the (conditional) hardness of \MPC low space algorithm and the hardness of (distributed) \LOCAL algorithms due to Ghaffari \etal \cite{GKU19} (and its deterministic setting due to Czumaj \etal \cite{CDP21d}), both problems have no constant-rounds deterministic low space \MPC algorithms (assuming the 1-vs-2-cycles conjecture) which are component-stable (the notion which has been introduced in \cite{GKU19} \emph{to capture ``most if not all'' efficient \MPC algorithms}). We overcome this challenge and show efficient deterministic \MPC algorithms (which are \emph{not} component-stable) for these two problems using a combination of derandomization tools relying on small sample spaces with limited dependence with some new algorithmic techniques of implementing these tools on \MPC{}s. As the result, we obtain an algorithm (see \Cref{thm:deterministic-connectivity-logD+loglogn}) that for a given undirected graph $G$ with diameter $D$, in $O(\log D + \log\log n)$ rounds on an \MPC (with $\spac = O(n^\delta)$ local space and $\gspac=O(m+n)$ global space) deterministically identifies all connected components of $G$.
In fact, if one wants to parameterize the complexity of our deterministic algorithm in term of $n$, $m$, and $D$, then we obtain a deterministic \MPC algorithm that runs in $O(\log D + \log\log_{m/n}n)$ rounds.


{ 
\showtheorem{\Cref{thm:deterministic-connectivity-logD+loglogn}}{
Let $\delta > 0$ be an arbitrary constant. Let $G$ be an undirected graph with $n$ vertices, $m$ edges, and each connected component having diameter at most $D$. One can deterministically identify the connected components of $G$ in $O(\log D + \log\log_{m/n} n)$ rounds on an \MPC with $\spac = O(n^\delta)$ local space and $\gspac=O(m+n)$ global space.
}
}


Similarly as in \cite{BDELM19}, our algorithm does not require prior knowledge of $D$.

Furthermore, we also extends the recent conditional lower bound (conditioned on the 1-vs-2-cycles conjecture) of $\Omega(\log D)$ for connectivity in graphs of low diameter due to Behnezhad \etal \cite{BDELM19} to the entire spectrum of $D$ (see \Cref{thm:lb-in-connectivity-algorithms-rev} for precise statement). The original lower bound from \cite{BDELM19} was assuming that $D \ge \log^{1+\Omega(1)}n$ and we extend this approach to work for arbitrarily small values of $D$. In particular, this allows us to argue that even for random graph, which in a typical model have diameter $D = O(\log n)$, or even $D = o(\log n)$, no algorithm can achieve $o(\log D)$ \MPC rounds complexity (assuming the 1-vs-2-cycles conjecture).

Similarly as mentioned in \cite{ASSWZ18}, our deterministic algorithms and the techniques developed in our paper can be used to apply our framework to related graph problems (e.g., finding a rooted spanning forest, a minimum spanning forest, and a bottleneck spanning forest), albeit (similarly as in \cite{ASSWZ18}) with a small loss of the complexity, see \Cref{sec:extentions-to-spanning-forest-etc} for more details.


\subsection{Our approach}

After a thorough analysis of the earlier \MPC connectivity algorithms due to Andoni \etal \cite{ASSWZ18} and Behnezhad \etal \cite{BDELM19} we determine two basic primitives, whose $O(1)$-rounds derandomization on an \MPC would allow to deterministically solve the connectivity problem within the same complexity bounds as those in \cite{ASSWZ18,BDELM19}. The two underlying problems are a constant approximation of maximum matching in graphs of maximum degree 2 and some variant of the set cover problem with all sets of the same size (the problem is called a random leader contraction in \cite{ASSWZ18}). As we mentioned above (see also \Cref{sec:derandomization}), these problems cannot be derandomized using component-stable algorithms, but we show that they can be solved using the derandomization approach using the method of conditional probabilities in the framework of pairwise independent random variables and $k$-wise \eps-approximately independent random variables approximating \DISTRnp, the product distribution of $n$ identically distributed 0-1 random variables, each with $\Ex{X_i}=p$ (see \Cref{sec:limited-independence}). While this framework is well established, its efficient implementation on an \MPC, and especially one with low local space and optimal global space, requires some careful approach. (See \Cref{sec:conditional-probabilities-on-MPC} for a more detailed overview of the approach using the method of conditional probabilities on \MPC.)




\section{Preliminaries}
\label{sec:preliminaries}


\subsection{Graph connectivity}
\label{sec:graph-connectivity-def}

Before we proceed, let us first formally introduce the connectivity problem in the \MPC setting.

\begin{definition}
\label{def:connectivity}
Let $G = (V,E)$ be an undirected graph with a vertex set $V$ and an edge set $E$. The goal is to compute a function $cc: V \rightarrow \NN$ such that every vertex $u \in V$ knows $cc(u)$ and for any pair of vertices $u, v \in V$, $u$ and $v$ are connected in $G$ if and only if $cc(u) = cc(v)$.
\end{definition}

Our algorithms will rely on the operation of \emph{vertex/edge contractions}. We will repeatedly contract one vertex $u$ to another adjacent vertex $v$, which means deleting the edge $\{u,v\}$ and identifying vertices $u$ and $v$ by removing $u$ and connecting to $v$ all edges incident earlier to $u$. Further, since these operations will be performed in parallel, we will ensure that the contractions are independent in the sense that they all can be performed in a single round. Once we know the pairs $u, v$, these operations can be performed deterministically in constant number of rounds on an \MPC with $\spac = O(n^\delta)$ local space and $\gspac = O(n)$ global space (see also \Cref{sec:basic-primitives}). Furthermore, since after contractions the obtained graph may have parallel edges and self-loops, one can easily make it simple by removing multiple edges and self-loops within the same bounds (using the primitives presented in \Cref{sec:basic-primitives}). Let us also mention that if a graph $G^*$ is obtained by a sequence of contractions of vertices/edges in a graph $G$, then $G^*$ maintains all connected components of $G$. As the result, it is a simple exercise to see that if we solved the connectivity problem for $G^*$ then we can always retrieve all connected components in $G$ in the same time as needed to construct $G^*$. We will use this remark repeatedly and without further mentioning it.


\subsection{Basic algorithmic techniques on low space \MPC}
\label{sec:basic-primitives}

We will be using the algorithmic toolbox for \MPC with low local space and linear global space as developed in some earlier works, most notably thanks to Goodrich \etal \cite{GSZ11}. We will regularly use the following lemma describing some basic (deterministic) algorithmic results.

\begin{lemma}\textbf{\emph{\cite{GSZ11}}}
\label{lemma:constant-time-basic-primitives}
Let $\delta$ be an arbitrary positive constant. Sorting of $n$ numbers, computing prefix sums of $n$ numbers, and computing the predecessor predicate\footnote{For the predecessor operation, the input consists of $n$ objects $x_1, \dots, x_n$, each associated with a 0-1 value; the output is a sequence $y_1, \dots, y_n$ such that $y_j$ is the last object $x_{j'}$ with $j' < j$ which has associate value of 1.} can be performed deterministically in a constant number of rounds on an \MPC using $\spac = O(n^\delta)$ local space, $\gspac = O(n)$ global space, and $\poly(n)$ local computation.
\end{lemma}

In the sorting and prefix sums (given $x_1, \dots, x_n$ and an associate operator $\oplus$, output $y_1, \dots, y_n$ with $y_j = \oplus_{i=1}^j x_j$) problems the input is arbitrarily distributed among the machines, and in the output, the $i$-th machine stores the $i$-th chunk of the sorted elements or the appropriate prefix sums.

Observe that this setup allows us perform some basic computations on graphs deterministically in a constant number of \MPC rounds. This includes the tasks of computing vertex degree of every vertex, or ensuring that every vertex has all incident edges stored on the machine its allocated to (or if it's degree $\deg(u) = \Omega(n^{\delta})$ then on $O(n^{\delta}/\deg(u))$ consecutive machines), etc.

Further, note that we typically will assume that $G$ is a simple graph, or that we can make it simple. This is because using Lemma \ref{lemma:constant-time-basic-primitives}, one can make any graph simple in a constant number of rounds with local space $\spac = O(n^{\delta})$ and optimal global space $\gspac = O(n+m)$: sort all edges $\{u,v\}$ according to $u$ and $v$ and then remove all duplicates (after sorting this task can be easily done locally on each machine, or if some parallel edges do not fit a single machines, one can easily coordinate the removal between the machines). This can be further extended to the process of independent vertex/edge contractions (assuming that if we contract some vertices to a vertex $u$ then $u$ is not contracted to any other vertex).

Finally, let us notice that the framework of Lemma \ref{lemma:constant-time-basic-primitives} allows us to perform in a constant number of rounds the following task of \emph{colored summation}. Given a sequence of $n$ pairs of numbers $\langle \text{color}_i, x_i \rangle$, $1 \le i \le n$, with $C = \{\text{color}_i: 1 \le i \le n\}$, compute $S_c = \sum_{i: \text{color}_i = c} x_i$ for all $v \in C$.

\begin{lemma}
\label{lemma:colored-summation}
Let $\delta$ be an arbitrary positive constant. The problem of colored summation for $n$ pairs of numbers can be solved deterministically in a constant number of rounds on an \MPC using $\spac = O(n^\delta)$ local space, $\gspac = O(n)$ global space, and $\poly(n)$ local computation.
\end{lemma}

\begin{proof}
First sort all pairs according to their first coordinates (colors) to obtain a sequence $p_1 = \langle \text{color}_1, x_1 \rangle, \dots, p_n = \langle \text{color}_n, x_n \rangle$ with $\text{color}_i \le \text{color}_{i+1}$. Next, compute the prefix sums according to the second coordinate, obtaining sequence $y_1, \dots, y_n$ with $y_i = \sum_{j=1}^i x_j$. Next, define a 0-1 sequence $z_1, \dots, z_n$ such that $z_i = 1$ iff pair $p_i$ corresponds to the last pair with $\text{color}_i$. That is, $z_n = 1$ and for any $1 \le i \le n-1$, $z_i = 1$ iff $\text{color}_i \ne \text{color}_{i+1}$ for the two consecutive pairs $p_i = \langle \text{color}_i, x_i \rangle$ and $p_{i+1} = \langle \text{color}_{i+1}, x_{i+1} \rangle$. Then, we consider the predecessor operation (Lemma \ref{lemma:constant-time-basic-primitives}), where for the sequence $y_1, \dots, y_n$, we associate $z_i$ with $y_i$, $1 \le i \le n$. Next we compute the predecessor operation using Lemma \ref{lemma:constant-time-basic-primitives}, which returns a sequence $q_1, \dots, q_n$ such that $q_i$ is equal to the last element $y_j$ with $j < i$ and $z_j = 1$; if $i$ has no such element (since $z_j = 0$ for all $j < i$) then $q_i = 0$. Now, observe that for every $1 \le i \le n$, if $z_i = 1$ then $\sum_{j: \text{color}_j = \text{color}_i} x_i = y_i - y_{q_i}$. Therefore, in order to complete the task of colored summation we return all pairs $(\text{color}_i, y_i - y_{q_i})$ with $z_i = 1$. This completes the proof since all the operations are deterministic and can be performed on an \MPC using $\spac = O(n^\delta)$ local space, $\gspac = O(n)$ global space, and $\poly(n)$ local computation.
\end{proof}


\subsubsection{One vs. two cycles problem}
\label{sec:1-2-cycle}

One of the notorious problems for low space \MPC complexity is the problem of distinguishing whether an input graph is an $n$-vertex cycle or
consists of two $\frac{n}{2}$-vertex cycles. This fundamental \MPC problem has been studied extensively in the literature and has been playing central role in the development of conditional lower bounds for \MPC graph algorithms (see, e.g., \cite{BDELM19,LMOS20,NS19,RVW18}). This problem is trivial for \MPC{}s with local space $\Omega(n)$, but we do not know its complexity for low space \MPC{}s. This leads to the celebrated (and widely believed) 1 cycle vs 2 cycles Conjecture:

\begin{conjecture}\textbf{(1-vs-2 cycles Conjecture)}
\label{conjecture-1-vs-2-cycles}
Let $\delta<1$ be an arbitrary positive constant. No randomized \MPC algorithm with low space $\spac = O(n^{\delta})$ and a polynomial number of machines can distinguish one $n$-vertex cycle from two $\frac{n}{2}$-vertex cycles in $o(\log n)$ rounds.
\end{conjecture}

The required failure probability in \Cref{conjecture-1-vs-2-cycles} is typically $\frac{1}{\poly(n)}$, but sometimes may be made as large as 0.01. We also do not know if even an exponential number of \MPC machines would help.


\subsection{Intuitions behind randomized $\widetilde{O}(\log D)$-rounds connectivity}
\label{sec:intuitions-connectivity-algorithms}

We first present the main idea behind the approach due to Andoni \etal \cite{ASSWZ18} that gives a randomized $O(\log D \log\log_{m/n}n)$-rounds \MPC algorithm for connectivity. Then we briefly discuss how this approach can be extended to obtain a randomized $O(\log D + \log\log n)$-rounds \MPC algorithm for the connectivity due to Behnezhad \etal \cite{BDELM19}. Our main result is a derandomization of the latter algorithm.

The main idea behind a connectivity algorithm working in $\widetilde{O}(\log D)$ \MPC rounds, as introduced by Andoni \etal \cite{ASSWZ18}, is to repeatedly perform \emph{vertex contractions}, which reduces the number of vertices, and \emph{expansion}, which produces new edges without changing the internal connectivity. The critical parameter in the analysis is limited local space, and with this the global space $\gspac$, which is $\gspac = \spac \cdot \machines$. Notice that graph contraction reduces the number of vertices and possibly edges, but expansion increases the number of edges. Therefore, we will want to ensure a balance, so that the number of edges is always at most $\gspac$.

Let $b$ be some integer parameter. A central observation is that if the input graph $G = (V,E)$ has all vertices of degree at least $b$, then a classical result about dominating sets ensures that $G$ has a dominating $U$ set of size $O(\frac{n \log n}{b})$, and moreover, such a dominating set can be easily found (\emph{by a randomized algorithm}). Since $U$ is a dominating set, all vertices in $G$ are either in $U$ or are neighbors of a vertex in $U$, and therefore one can locally contract all vertices from $V \setminus U$ to $U$, reducing the number of vertices in the contracted graph to $O(\frac{n \log n}{b})$. (This step was called a \emph{random leader contraction} in \cite{ASSWZ18}.)

This observation suggests the following approach. Start with the original graph $G$ with $n$ vertices and $m$ edges, and set $b = \frac{m}{n}$.
Notice that $G$ can be represented with global space $\gspac = O(n+m)$, which we will see as $O(m)$. Then, perform ``expansion'' on $G$ by adding some new edges without changing the connectivity structure, so that each vertex has degree at least $b$ (unless its connected component has size at most $b$), while not increasing the total number of edges by too much, so that in total the size of the graph is smaller than \gspac. Then, use the property above (contracting after leader contraction using dominating set) to reduce the number of vertices from $n$ to $n_1 = O(\frac{n \log n}{b})$. In the obtained graph we perform expansion again, but this time, since we have fewer vertices, aim to ensure that each vertex has degree at least $b_1$ and the total number of added edges is $O(b_1 n_1) = O(m)$, so that it fits the global space of the \MPC. The constraint $b_1 n_1 = O(m)$ leads us to the choice for $b_1 = \frac{m}{n_1} = O(\frac{m b}{n \log n})$, or equivalently, $b_1 = O(\frac{b^2}{\log n})$.

Next, we repeat the same procedure again. Since we consider a graph with $n_1$ vertices, each of degree at least $b_1$, we can contract the graph to a new graph with at most $n_2 = O(\frac{n_1 \log n}{b_1})$ vertices. Notice that $n_2 = O(\frac{n_1 \log n}{b_1}) = O(\frac{n \log n}{b} \log n \frac{\log n}{b^2}) = O(n (\frac{\log n}{b})^3) = O(\frac{m \log^3n}{b^4})$. For such a graph we perform expansion to ensure that the new graph has $n_2$ vertices and every vertex has degree at least $b_2$ and the total number of edges is $O(b_2 n_2)$, with the constraint $b_2 n_2 = O(m)$. This constraint leads us to choose $b_2 = \frac{m}{n_2}$, which in turns gives $b_2 = \frac{m}{n_2} = O(\frac{m}{n} (\frac{b}{\log n})^3)$, or equivalently, $b_2 = O(\frac{b^4}{\log^3 n})$. If we continue this approach, then after $t$ iterations we obtain a graph with at most $n_t$ vertices, which can be later expanded to ensure that each vertex has degree at least $b_t$ with at most $b_t n_t = O(m)$ edges, where $n_t = O(\frac{n_{t-1} \log n}{b_{t-1}}) = O(\frac{m}{\log n} (\frac{\log n}{b})^{2^t})$ and $b_t = \frac{m}{n_t} = O((\frac{b}{\log n})^{2^t}\log n)$.

Notice that this process works in phases and each phase reducing the number of vertices by a factor $O(\frac{\log n}{b})$. Since we started with $b = \frac{m}{n}$, this would lead to a randomized \MPC algorithm running in $O(\log\log_{m/n}n)$ phases. This general technique is called \emph{double-exponential speed problem size reduction} in \cite{ASSWZ18}. Note that the result of $O(\log\log_{m/n}n)$ phases holds for any $b_t = (\frac{m}{n_t})^{\Omega(1)}$ (and in fact, the algorithm of \cite{ASSWZ18} uses $b_t = (\frac{m}{n_t})^{1/2}$: we presented a simpler version here for the purposes of exposition).

What is left is to design a single phase that performs expansion. That is, for every vertex $u$ the process either connects $U$ to all vertices in its connected component (if the number of such vertices is less than $b$) or it adds new edges incident to $u$ which connect $u$ to new vertices in its connected component while increasing its degree to $b$. This operation can be done in $O(\log D)$ \MPC rounds using a variant of the standard transitive closure approach. That is, one repeats $O(\log D)$ times the squaring operation of connecting any vertex $u$ to all vertices in 2 hops away from $u$. The problem with this approach is that if the number of vertices in the second neighborhood of $u$ is too large, this operation would require too much global \MPC space. But one can deal with this problem by stopping the process (for each individual vertex) when a given vertex learnt already $b$ vertices in its connected component. Thus we obtain a procedure that in $O(\log D)$ \MPC rounds, for every vertex $u$, either increases the degree of $u$ to at least $b$ or makes $u$ aware of its entire connected component.

If we combine vertex expansion with leader contraction, then we will obtain a randomized algorithm that in $O(\log\log_{m/n}n \cdot \log D)$ \MPC rounds contract each connected component to a single vertex, and with this, that computes all connected components.

There are two weaknesses of this approach: Firstly, its complexity is $\omega(\log D)$ rounds unless $m \ge n^{1 + \Omega(1)}$. Secondly, the algorithm for random leader contraction described above achieves a high success probability only when $b$ is large, and thus the total probability of succeeding is in the worst case only constant, say 0.99. Behnezhad \etal \cite{BDELM19} extended the framework due to Andoni \etal \cite{ASSWZ18} and evaded these two weaknesses to achieve $O(\log D + \log\log_{m/n}n)$ \MPC rounds with high probability.

The first part of the improvement stems from a randomized algorithm that in $O(1)$ \MPC rounds applies edge contractions to reduce the number of vertices of $G$ by a constant factor. Using this algorithm, in $O(\log\log n)$ \MPC rounds one can reduce the number of vertices from $n$ to $n/\polylog(n)$, without increasing the number of edges. As the result, the approach by Andoni \etal presented above can start with $b$ to be at least $\polylog(n)$, ensuring that the random leader contraction in \cite{ASSWZ18} \emph{will succeed with high probability}.

On its own, this will only increase the success probability, but the framework would still require $O(\log\log_{m/n}n)$ phases, each phase taking $O(\log D)$ rounds.
The second and the main contribution of Behnezhad \etal \cite{BDELM19} is that vertices can perform each step (neighborhood expansion versus leader contraction) independently of other vertices. That is, one still wants to perform $O(\log D)$ rounds for transitive closure, but this is interleaved with the random leader contraction. In each round, each vertex $u$ (which has not been contracted yet) has its own \emph{budget} of the local space it can use, with the constraint (informally) that the sum of all budgets is $O(\gspac)$. Each vertex starts with a budget of $(\frac{m}{n})^{1/2}$. Then, informally, every vertex $v$ in a constant number of \MPC rounds either increases its budget (from $b(v)$ to $b(v)^{1+c}$ for a small constant $c$), or learns its entire 2-hop neighborhood. Notice that a vertex can only increase its budget $O(\log\log n)$ times: once the budget of a vertex reaches $n$, the vertex is able to store all vertices in its allocated local space, which is always sufficient.

In fact, the analysis here requires some additional arguments since this property alone is not enough to show that the algorithm terminates in $O(\log D + \log\log_{m/n}n)$ rounds, but this issue can be dealt with either using the arguments provided in a later paper of Liu \etal \cite{LTZ20} or in a subsequent version of \cite{BDELM19a}.


\subsection{Derandomization}
\label{sec:derandomization}

Our main contribution is derandomization of two key steps of the algorithms due to Andoni \etal \cite{ASSWZ18} and  Behnezhad \etal \cite{BDELM19}. We will show that in order to achieve $O(\log D + \log\log_{m/n}n)$-rounds deterministic \MPC algorithm for connectivity it is enough to derandomize the random leader contraction process and to solve the problem of finding large matching on a line, and to solve these problems deterministically in a constant number of rounds. Both these tasks rely on a simple random sampling approach and are easy when one allows randomization. The task of derandomization them seems to be more complex. In fact, the recently developed framework of conditional lower bounds for low space \MPC algorithms due to Ghaffari \etal \cite{GKU19} (see also \cite{CDP21d,CDP21a} for deterministic lower bounds) suggests that this may be impossible.

Ghaffari \etal \cite{GKU19} introduced the notion of \emph{component-stable} low-space \MPC algorithms, which are (informally) \MPC algorithms for which the outputs reported by the vertices in different connected components are required to be independent. They have been introduced with the \emph{intuition} that all, or almost all efficient \MPC algorithms are component-stable, and thus the restriction to study such algorithms can be done almost without loss of generality. Then, the main result of Ghaffari \etal \cite{GKU19} is a connection between the complexity of component-stable low-space \MPC algorithms and of computation in the classical distributed computing \LOCAL model. Informally, Ghaffari \etal \cite{GKU19} presented a general technique that for many graph problems translates an $\Omega(T)$-round \LOCAL lower bound into an $\Omega(\log T)$-round lower bound in the low-space \MPC model conditioned on the connectivity conjecture. While the original framework was designed only for randomized algorithms, Czumaj \etal \cite{CDP21d} extended it to deterministic algorithms. In particular, this shows that the two problems we want to derandomize require super-constant number of \MPC rounds, unless we consider non-component-stable \MPC algorithms. This follows from two lower bounds for deterministic algorithms in the \LOCAL model: in a graph with maximum degree 2 there is no deterministic $o(\log^*n)$-time \LOCAL algorithm that returns a constant approximation of the maximum matching \cite{CHW08,LW08}, and no deterministic constant-round \LOCAL algorithm can find an $o(\Delta)$-approximation of a minimum dominating set in $\Delta$-regular graphs \cite{CHW08,LW08} (see also Section 6.4 in \cite{Suomela13}). If one combines these lower bounds with the framework due to Ghaffari \etal \cite{GKU19} and Czumaj \etal \cite{CDP21d}, these results show that there is no $O(1)$-round deterministic component-stable low-space \MPC algorithm for any of these two problems.

While this suggests that the two problems required in our analysis cannot be solved deterministically in a constant number of \MPC rounds, we prove that this is not the case: one can construct $O(1)$-rounds deterministic \MPC algorithms for a constant approximation of a maximum matching in graphs with maximum degree at most 2; and for a good approximation of an appropriate instance of dense dominating set. As it is required by the framework of Ghaffari \etal \cite{GKU19,CDP21d} in view of the lower bounds for the \LOCAL model (see \cite{CHW08,LW08}), our algorithms are not component-stable.

Only very recently it has been shown by Czumaj \etal \cite{CDP21d} that there are low-space \MPC algorithms that are not component-stable and which break the lower bounds delivered by the framework of Ghaffari \etal In this context, our algorithms strengthen the case made in \cite{CDP21d} because we do not only beat the lower bounds for component-stable algorithms, but we also achieve the optimal global space bounds (since we maintain the optimal global space $\gspac = O(m+n)$).


\subsection{Introduction to derandomization}
\label{sec:intro-derandomization}

The derandomization approach in this paper follows the classical approach. We study random sampling algorithms that originally assume full independence between random choices, and in the first step, we reduce the independence of the involved random variables, either to allow their pairwise independence (\Cref{def:k-wise-independent-rvs}), or $k$-wise \eps-approximate independence (\Cref{def:k-eps-wise-independent-rvs}). Then, we set an appropriate target function modeling some required properties of the algorithm. Next, we apply the method of conditional probabilities to the target function to obtain a deterministic \MPC algorithm.

This outline of the approach combines three different aspects: the use of appropriate random variables with some limited independence coming from a small probability space (or equivalently, modeled by a family of hash functions of polynomial size, \Cref{sec:limited-independence}), the use of the method of conditional probabilities to select one good hash functions the previous step (\Cref{sec:conditional-probabilities-on-MPC}), and the design of efficient \MPC algorithm implementing this approach. The last part is less standard and (for one problem) more challenging, and therefore we will pay special care to its description.


\subsubsection{Method of conditional probabilities on \MPC}
\label{sec:conditional-probabilities-on-MPC}

A central tool in our derandomization of algorithms is the \emph{method of conditional probabilities} (or \emph{expectations}) (see, e.g., \cite{AS16,MR95,Raghavan88}). We consider a setting where there is a family of hash functions \HH of size $\poly(n)$ (see \Cref{sec:limited-independence-small-families}) and we know that there exists at least one $h^* \in \HH$ with some desirable properties. Our goal is to find one such $h^*$. This task can easily be performed in $\poly(n)$ time by going through all functions in \HH, but we want to perform it in a constant number of rounds on an \MPC with local space $\spac = O(n^{\delta})$ and linear global space $\gspac = O(n+m)$, which is insufficient to collect information about all evaluations of all functions from \HH on the entire \MPC.

The method of conditional probabilities, developed by Erd{\"{o}}s and Selfridge \cite{ES73}, Raghavan \cite{Raghavan88}, and many others, works by defining $h^*$ by a sequence of repeatedly refining \HH into smaller subsets, $\HH = \HH^{\langle 0 \rangle} \supseteq \HH^{\langle 1 \rangle} \supseteq \dots$, until finally we will end up with a one element set defining $h^*$. The way of refining \HH is led by not increasing (or not decreasing) the conditional expectation of some target function $F$ on \HH, that is, to ensure that $\SEx{h \in \HH_{i+1}}{F(h)} \le \SEx{h \in \HH_i}{F(h)}$. The main difficulty with this approach is in efficiently computing the conditional expectations, and in bounding the number of iterations needed to find a one element subset of \HH.

We will incorporate the approach used recently in \cite{CPS20} and \cite{CDP20a} (see also recent results in \cite{CDP20b,CDP21a,CDP21b}; though a similar approach has been used in earlier, classic works on derandomization, see, e.g., \cite{MNN94} or \cite{Luby93,Raghavan88}, with the only difference being that now we have to process by blocks of $O(\log n/\delta)$ bits, whereas in earlier works one was processing individual bits, one bit at a time\footnote{One can think of the most standard process as being modeled by a binary tree, and we consider a $O(n^{\delta})$-ary tree.}). We will search for function $h^*$ by splitting the $O(\log n)$-bit seed defining it into smaller parts of $\chunk < \log\spac = O(\delta\log n) = O(\log n)$ bits each\footnote{One could think that $\chunk = \lfloor\log\spac\rfloor - 1 \le \log(\spac/2)$, to ensure that every machine has about half of its local space for the analysis of hash functions and another half for its normal data.}, and then processing one part at time, iteratively extending a fixed prefix of the seed until we have fixed the entire seed.

If any $h \in \HH$ is specified by a seed of $\tau$ bits, then we will perform this process in $\lceil \frac{\tau}{\chunk}\rceil$ phases. In phase $i$, $1 \le i \le \lceil \frac{\tau}{\chunk}\rceil$, we determine bits $1+(i-1)\lceil \frac{\tau}{\chunk}\rceil \dots i\lceil \frac{\tau}{\chunk}\rceil$ of $h^*$. For that, at the beginning of phase $i$, we assume that each machine knows the bits of the first $i-1$ parts of $h^*$ (that is, prefix of $(i-1)\chunk$ bits). Next, each machine considers extension of the current bits by all possible seeds of length \chunk. Since $2^{\chunk} < \spac$, this operation can be performed on every machine individually. For every possible \chunk-bit extension of $h^*$, every machine computes locally the expected cost of the target function $F$ with this seed, producing $2^{\chunk}$ different values for the target function, one for each partial seed. Then, one aggregates these values from all machines and chooses the \chunk-bit extension of $h^*$ which minimizes or maximizes the target function using the method of conditional probabilities.

The selection for the target function $F(h)$ is central for this process to work. If we consider a partition of the currently considered subset $\HH^*$ of \HH into $\HH^*_1, \dots, \HH^*_{2^{\chunk}}$, then it is essential that each machine will be able to efficiently compute $\SEx{h \in \HH^*_i}{F(h)}$ for all $1 \le i \le 2^{\chunk}$. In one our application, in the proof of \Cref{thm:deterministic-large-matching-on-a-ring} for a large matching in degree-2 graphs, function $F(h)$ is simply the size of matching corresponding to hash function $h$. Then, as we will show, we can compute exactly $\SEx{h \in \HH^*_i}{F(h)}$, and hence our approach ensures that we can find one $1 \le i \le 2^{\chunk}$ with $\SEx{h \in \HH^*_i}{F(h)} \le \SEx{h \in \HH^*}{F(h)}$.

The approach in the other algorithm, in \Cref{subsec:MPC:alg:deterministic-dominating-set-in-dense-graphs}, is more complicated, since our desired function is a combination of elements that we want to see large and of those that we want to see small. To deal with this difficulty, we define an appropriate cost function $\cost(h)$ for which we can efficiently compute $\SEx{h \in \HH^*_i}{F(h)}$ for all $1 \le i \le 2^{\chunk}$, and such that for any $h$, $\cost(h)$ is very small iff $h$ is the desired hash function. With such a cost function, we can apply our approach to repeatedly refining the original family of hash functions until in a constant number of \MPC rounds we obtain a one-element family, which will give us a desired hash function $h^*$.


\section{Deterministically reducing the number of vertices}
\label{sec:reducing-vertices-in-connectivity-algorithms}

In this section we present an efficient deterministic \MPC algorithm (with low local space and optimal global space) that transforms a given input graph into a graph that maintains its connectivity structure, does not increase the number of edges, and significantly reduce the number of vertices. In particular, for a given undirected graph $G = (V,E)$ we would like to find an undirected simple graph $G^* = (V^*,E^*)$ obtained by a sequence of edge contractions of $G$ such that $|V^*| \ll |V|$.

We first briefly present the randomized \MPC algorithm due to Behnezhad \etal \cite[Section~IV]{BDELM19} that reduces the size of $V$ by a constant factor $\alpha < 1$. Once we have this property, after repeating it $t$ times we will reduce the number of vertices by $\alpha^t$. Later, in \Cref{sec:derandomizing-thm:connectivity->smaller-graph}, we will show how to efficiently derandomize this algorithm while achieving all its key properties.

Behnezhad \etal \cite{BDELM19} proved the following claim.

\begin{theorem}\textbf{\emph{\cite[Section~IV]{BDELM19}}}
\label{thm:connectivity->smaller-graph}
Let $G = (V,E)$ be an arbitrary undirected simple graph with no isolated vertices. Randomized \Cref{alg:connectivity->smaller-graph} can be implemented to run in a constant number of \MPC rounds (with $\spac = O(|V|^{\delta})$ local space and $\gspac = O(|V|+|E|)$ global space) to construct an undirected simple $G^* = (V^*,E^*)$ such that
%
\begin{itemize}
\item $G^*$ is obtained by a sequence of \emph{independent} contractions of edges in $G$ (if a vertex is contracted to a vertex $u$, then vertex $u$ will not be contracted to any other vertex) and
\item $|V^*| \le \frac{99}{100} \cdot |V|$ with probability at least $1 - e^{-|V|/250}$.
\end{itemize}
\end{theorem}

\begin{remark}
\label{remark:reducing-vertices-connectivity}
Since $G^*$ is obtained by a sequence of contractions of edges in $G$, $G^*$ maintains all connected components of $G$. As the result, it is a simple exercise to see that if we solved the connectivity problem for $G^*$ then in a constant number of \MPC rounds (with $\spac = O(|V|^{\delta})$ local space and $\gspac = O(|V|+|E|)$ global space) we could retrieve all connected components in $G$.
\end{remark}


\begin{algorithm}[H]
\SetAlgoLined
\caption{\small Reducing the number of vertices of $G$ by a constant factor}
\label{alg:connectivity->smaller-graph}

\BlankLine
For each $u \in V$ create arc $[v,u\rangle$, where $v$ is a neighbor of $u$ in $G$ with the smallest ID

For each arc $[u,v\rangle$, if there is also arc $[v,u\rangle$ and $u < v$ then delete arc $[u,v\rangle$

For each $u \in V$, if there are at least two arcs $[\cdot,u\rangle$ then delete any outgoing arc $[u,\cdot\rangle$

For each $u \in V$, if there are at least two arcs $[\cdot,u\rangle$ then let $\text{\footnotesize\sc children}(u) = \{v \in V: [v,u\rangle\}$

\quad $\blacktriangleright$ \emph{contract} (in graph $G$) all vertices in $\text{\footnotesize\sc children}(u)$ into vertex $u$

\quad $\blacktriangleright$ delete all arcs $[\cdot,v\rangle$ for $v \in \{u\} \cup \text{\footnotesize\sc children}(u)$

For each arc $[u,v\rangle$ delete $[u,v\rangle$ with probability $\frac23$

For each arc $[u,v\rangle$, if $[u,v\rangle$ is isolated then \emph{contract} $u$ into $v$ (in graph $G$)

\end{algorithm}

\bigskip


\ArturKeep{Maybe later I'll think about putting some figures of representative examples for the run of \Cref{alg:connectivity->smaller-graph} here.}It is not difficult to see that all contractions performed by \Cref{alg:connectivity->smaller-graph} are ``disjoint'' and can be performed without conflicts (that is, if a vertex is contracted to another vertex $u$, then vertex $u$ will not be contracted to any other vertex). With this, it is straightforward to implement this algorithm in a constant number of rounds on an \MPC with $\spac = O(|V|^{\delta})$ local space and $\gspac = O(|V|+|E|)$ global space. Therefore in what follows, we will present key intuitions why the resulting graph has at most $\frac{99|V|}{100}$ vertices with high probability (see to Section~IV in \cite{BDELM19} for more details).

The algorithm can be seen as first choosing some linear number of edges in $G$ and making these edges to form a large directed forest in $G$. Then, one will want to perform a variant of tree contraction to contract a constant fraction of vertices to their parents, ensuring that the contractions are independent (if a vertex is contracted to vertex $u$ then $u$ is not contracted to any other vertex).

Let us discuss some more details behind \Cref{alg:connectivity->smaller-graph}. First, for each vertex $u$ we choose one of its neighbors $v$ that we could later \emph{potentially} contract $u$ into. In order to control the possible contractions, the choice of the neighbor can be seen as selecting an arc (directed edge); an arc $[u,v\rangle$ means that we may consider contracting $u$ into $v$. Further, in order to maintain some desirable structure of the arc-graph, we will use the vertex IDs for the selection: each vertex will choose its neighbor with the smallest ID (step 1); notice that this ensures that all vertices have out-degree 1. A useful feature of such selection is that the obtained arc-graph is \emph{almost acyclic} (all cycles --- even if we ignore the directions of the arcs --- are of length 2), and so the next step (step 2) is to make it acyclic and remove one arc from every cycle of length 2.

The obtained graph has a \emph{forest structure}, but still, we will want to simplify it even further.
%
%
Therefore, if there is a vertex of in-degree at least 2, then we remove its out-going arc (step 3). It can be show that this step reduces the number of arcs by at most half\footnote{One way to prove it is to observe that in a tree with $\ell$ vertices (and hence with $\ell-1$ edges), at most $\frac{\ell-1}{2}$ vertices have at least 2 children, and hence step~3 removes at most half of the edges in each tree.}, and hence the number of arcs remaining after step~3 is at least $\frac{n}{4}$. Notice that after this step the graph looks like a collection of disjoint \emph{long tail stars}, that is, every vertex with an outgoing arc has at most 1 incoming arc. Further, the root of each such tree (vertex with out-degree 0) is either a vertex with in-degree at least 2 or it may have in-degree 1 (in which case the root is the end of a path). We will now consider separately
\begin{inparaenum}[\it (i)]
\item the roots of in-degree greater than 1 and
\item the roots with in-degree equal to 1.
\end{inparaenum}

For each root $u$ of in-degree at least 2, if there are arcs $[v_1,u\rangle, \dots, [v_{\ell},u\rangle$, we will first \emph{contract} $v_1, \dots, v_{\ell}$ to $u$ (step 5) and then disconnect $v_1, \dots, v_{\ell}$ from their children in the tree (step 6).

After this operations, the obtained arc-graph consists of \emph{vertex-disjoint paths} and our goal is to find a matching (isolated arcs) on these paths that involves a constant fraction of arcs and then contract all matching arcs. For that, we first remove each arc with probability $\frac23$ (step 7). This ensures that the paths are getting shorter and with a constant probability every arc from before step 7 is isolated. Therefore in the next step (step 8) we contract all such isolated arcs.

In order to analyze the size of $|V^*|$, let us first observe that there are exactly $n$ arcs after step 1, then step 2 removes at most half of them, and then step 3 removes at most half of the remaining arcs. Hence, after step 3 we have at least $\frac{n}{4}$ arcs and the obtained graph is a collection of disjoint long tail stairs, such that each star has at most one vertex of in-degree greater than 1. Next, if we contracted $\ell$ edges in step~5 (over all vertices $u$) then there are at least $\frac{n}{4} - 2 \ell$ arcs after step~6. Finally, steps 7--8 are designed to find a large matching in a subgraph of maximum degree at most 2, which is exactly the task we consider in the next \Cref{sec:deterministic-approximate-matching}. The analysis in that section, as formalized in Claim \ref{thm:randomized-large-matching-on-a-ring}, shows that if the number of arcs after step~6 is $\kappa$, then with probability at least $1 - e^{-\kappa/54}$ we will perform at least $\frac{2\kappa}{27}$ edge-disjoint contractions in step~8.

With such parameters, one can easily see that the number of contractions is either at least $\frac{n}{100}$ in step~5, or is at least $\frac{n}{100}$ with probability at least $1 - e^{-n/250}$ in step~8.

Let us summarize the analysis above. We started with a graph $G$ and then performed (with high probability) a sequence of at least $\frac{n}{100}$ contractions to obtain a new graph $G^*$. Each contraction reduces the number of vertices by 1, and hence the resulting graph has at most $\frac{99|V|}{100}$ vertices, as promised in \Cref{thm:connectivity->smaller-graph}.
\hfill$\Box$\par\medskip


\subsection{Derandomizing \Cref{alg:connectivity->smaller-graph} (from \Cref{thm:connectivity->smaller-graph})}
\label{sec:derandomizing-thm:connectivity->smaller-graph}

\Cref{alg:connectivity->smaller-graph} is clearly randomized, but the only step which is randomized is step 7, or in fact steps 7--8. One can rephrase the task of steps 7--8 by one of finding a large matching in graph of maximum degree at most 2. The analysis in \cite[Section~IV]{BDELM19} and our claim in \Cref{thm:connectivity->smaller-graph} uses the result from Claim \ref{thm:randomized-large-matching-on-a-ring} that in a graph of maximum degree at most 2 and with $\kappa$ edges with high probability one can find a matching of size at least $\frac{2\kappa}{27}$. However, we can make the entire algorithm deterministic if we could deterministically find such matching, resulting in \Cref{alg:deterministic-connectivity->smaller-graph}.


\medskip

\begin{algorithm}[H]
\SetAlgoLined
\caption{\small Deterministically reducing the number of vertices of $G$ by a constant factor}
\label{alg:deterministic-connectivity->smaller-graph}

\BlankLine
For each $u \in V$ create arc $[v,u\rangle$, where $v$ is a neighbor of $u$ in $G$ with the smallest ID

For each arc $[u,v\rangle$, if there is also arc $[v,u\rangle$ and $u < v$ then delete arc $[u,v\rangle$

For each $u \in V$, if there are at least two arcs $[\cdot,u\rangle$ then delete any outgoing arc $[u,\cdot\rangle$

For each $u \in V$, if there are at least two arcs $[\cdot,u\rangle$ then let $I(u) = \{v \in V: [v,u\rangle\}$

\quad $\blacktriangleright$ \emph{contract} (in graph $G$) all vertices in $I(u)$ into vertex $u$

\quad $\blacktriangleright$ delete all arcs $[\cdot,v\rangle$ for $v \in I(u) \cup \{u\}$

Let $\mathcal{A}$ be the set of remaining arcs (each vertex has in-degree and out-degree at most 1)

Find a large matching $\mathbb{M}$ on the arcs in $\mathcal{A}$ 
(see \Cref{thm:deterministic-large-matching-on-a-ring})

For each arc $[u,v\rangle \in \mathbb{M}$, \emph{contract} $u$ into $v$ (in graph $G$)

\end{algorithm}

\bigskip


In \Cref{sec:deterministic-approximate-matching} we will show a deterministic \MPC algorithm for finding a large matching in graphs of maximum degree 2, proving in \Cref{thm:deterministic-large-matching-on-a-ring} that the size of the matching found in step~8 of \Cref{alg:deterministic-connectivity->smaller-graph} is at least $\frac{\kappa}{8} \ge \frac{2\kappa}{27}$.
\ArturKeep{Just to remember this in the case anybody would be asking: What is the best bound we can hope for $|V \setminus V^*|$? We can only promise the bound $|V \setminus V^*| \ge \ell + \frac{\kappa}{8}$ subject to $\ell, \kappa \ge 0$, $2\ell \le \frac{n}{4}$, and $\kappa \ge \frac{n}{4}-2\ell$. But then,
\begin{align*}
    |V \setminus V^*| &\ge
    \ell + \frac{\kappa}{8} \ge
    \ell + \frac{\frac{n}{4}-2\ell}{8} =
    \frac{n}{32} + \frac{3\ell}{4} \ge
    \frac{n}{32}
    \enspace.
\end{align*}
BTW the bound in \Cref{thm:connectivity->smaller-graph} is (a tiny bit) more tricky, since the bound for $\frac{2\kappa}{27}$ holds only whp.
}If we plug this in the analysis in \Cref{thm:connectivity->smaller-graph}, then we obtain:

\begin{ftheorem}
\label{thm:deterministic-connectivity->smaller-graph}
Let $G = (V,E)$ be an arbitrary undirected simple graph with no isolated vertices. In a constant number of \MPC rounds (with $\spac = O(|V|^{\delta})$ local space and $\gspac = O(|V|+|E|)$ global space) one can deterministically construct an undirected simple graph $G^* = (V^*,E^*)$ such that
%
\begin{itemize}
\item $G^*$ is obtained by a sequence of independent contractions of edges in $G$ and
\item $|V^*| \le \frac{99}{100} |V|$.
\end{itemize}
\end{ftheorem}


\subsection{Deterministically reducing the number of vertices by $\polylog(n)$ factor}
\label{sec:deterministic-reduce-vertices-by-polylog}

\Cref{thm:deterministic-connectivity->smaller-graph} reduces the number of vertices by a constant fraction, and therefore one can easily extend it by repeating \Cref{thm:deterministic-connectivity->smaller-graph} $O(\log\log n)$ many times to ensure that the obtained graph will have at most $\frac{|V|}{\poly(|V|)}$ vertices. (See also Remark \ref{remark:reducing-vertices-connectivity}; the contractions will ensure that the connectivity is maintained, even though the contractions do not have to be independent anymore.)%
\ArturKeep{\revised{Because of \Cref{remark:loglog-base}, maybe in could be better to state a similar claim here: if $|E| \le |V| \log^c|V|$ then in $O(\log\log_{|E|/|V|} |V|)$ \MPC rounds one can deterministically built an undirected simple $G^* = (V^*,E^*)$ such that
\begin{inparaenum}[\it (i)]
\item $G^*$ is obtained by a sequence of contractions of edges in $G$ and
\item the number of non-isolated vertices in $V^*$ is at most $\frac{|V|}{\log^c|V|}$.
\end{inparaenum}}}

\begin{ftheorem}
\label{thm:deterministic-connectivity->smaller-graph-log}
Let $c$ be an arbitrary constant. Let $G = (V,E)$ be an arbitrary undirected simple graph. In $O(\log\log n)$ \MPC rounds (with $\spac = O(|V|^{\delta})$ local space and $\gspac = O(|V|+|E|)$ global space) one can deterministically construct an undirected simple graph $G^* = (V^*,E^*)$ such that
%
\begin{itemize}
\item $G^*$ is obtained by a sequence of contractions of edges in $G$ and
\item the number of non-isolated vertices in $V^*$ is at most $\frac{|V|}{\log^c|V|}$.
\end{itemize}
%
Furthermore, within the same time bounds, one can transform the solution to the connectivity problem for $G^*$ to the solution to the connectivity problem for $G$.
\end{ftheorem}

Let us also make a further remark which is relevant to our connectivity algorithms in \Cref{sec:deterministic-connectivity} (and which was first made in \cite[Section IV]{BDELM19}). In our paper, we will regularly refer to 
\Cref{thm:deterministic-connectivity->smaller-graph-log} that in $O(\log\log n)$ rounds reduces the number of vertices by a polylogarithmic factor in order to ensure that the ratio of the original number of edges (which corresponds to \gspac) over the current number of vertices is $\Omega(\log^{10}n)$. However, \Cref{alg:deterministic-connectivity->smaller-graph} is not needed when $m > n \log^{10} n$. This observation will allow us to obtain round complexities containing an $O(\log\log_{m/n} n)$ term.

\begin{remark}
\label{remark:loglog-base}
If $m < n\log^{10} n$ then the round complexity of running \Cref{alg:deterministic-connectivity->smaller-graph} until $m > n \log^{10}n$ is $\Theta(\log\log_{m/n} n) = 
\Omega(\log\log n)$.
\end{remark}


\section{Deterministically approximating matching in degree 2 graphs}
\label{sec:deterministic-approximate-matching}

In this section we consider the following task required to complete the proof of \Cref{thm:deterministic-connectivity->smaller-graph}: Given a graph $G$ with $n$ vertices, $m$ edges, and with maximum degree $\Delta \le 2$, \emph{deterministically} find a \emph{matching of size $\Omega(m)$} in a \emph{constant number of rounds} on an \MPC with low local space $\spac = O(n^{\delta})$ and global space $\gspac = O(n)$.

Observe that since the input graph has maximum degree at most 2, it consists solely of vertex disjoint cycles and paths. Further, it is easy to see that it always has a matching of size at least $\frac{m}{3}$.
It is also easy to find such maximum matching sequentially in linear time, but any parallel algorithm would seem to require a logarithmic (in $D$) number of rounds. However, it is easy to find a \emph{large matching} in a constant number of rounds using random sampling (this result is folklore).

\begin{claim}
\label{thm:randomized-large-matching-on-a-ring}
Let $G = (V,E)$ be an undirected simple graph with maximum degree $\Delta \le 2$. With probability at least $1 - e^{-m/54}$, one can find a matching $\mathbb{M}$ of $G$ of size at least $\frac{2m}{27} = \Omega(m)$ in a constant number of \MPC rounds with local space $\spac = O(n^{\delta})$ and global space $\gspac = O(n)$.
\end{claim}

\begin{proof}
We rely on the following simple randomized algorithm (see \Cref{alg:randomized-large-matching-on-a-ring}): first remove each edge from $G$ independently at random with probability $p = \frac23$, and then define $\mathbb{M}$ to be the set of all remaining isolated edges (i.e., whose both endpoints have degree 1 now).


\medskip

\begin{algorithm}[H]
\SetAlgoLined
\caption{\small Finding a large matching in graphs with $\Delta \le 2$}
\label{alg:randomized-large-matching-on-a-ring}

\BlankLine
Initialize $\mathbb{M}:= \emptyset$

For each edge $\{u,v\} \in E$, remove $\{u,v\}$ independently at random with probability $p = \frac23$

For each edge $\{u,v\}$ still remaining in $E$, if $u$ and $v$ have degree 1 then add $\{u,v\}$ to $\mathbb{M}$
\end{algorithm}

\bigskip


It is easy to see that $\mathbb{M}$ is a matching and 
that the algorithm can be implemented in a constant number of \MPC rounds within the space bound as given in the theorem. Further, using standard concentration bounds one can show the promised high probability bound that $|\mathbb{M}| = \Omega(m)$. Since the analysis follows the standard approach, we leave it to the reader, and in the next theorem we will present a more elaborate analysis which implicitly supersedes the details here.
\junk{

It is easy to see\Artur{Most likely we will add some comments here, how this can be done.} that $\mathbb{M}$ is a matching and it is not difficult to see that the algorithm can be implemented in a constant number of \MPC rounds within the space bound as given in the theorem. In what follows we will only show the probability bound that $|\mathbb{M}| = \Omega(m)$. The analysis follows the standard approach, but we present the details here, to refer to it later when it will be transferred into a deterministic algorithm using only 2-wise independence.

For any edges $e \in E$, let $Y_e$ be the indicator variable that $e \in \mathbb{M}$; then $|\mathbb{M}| = \sum_{e \in E} Y_e$.

For any edge $e \in E$, let $N_e$ be the set of edges incident to $e$ in $G$; clearly, $|N_e| \le 2$. For any $e \in E$, let $X_e$ be a 0-1 random variable that is 1 if and only if edge $e$ is not removed in step 2. Notice that $Y_e = 1$ if and only if $X_e = 1$ and $X_{e'} = 0$ for all $e' \in N_e$. Therefore, we have
\begin{align*}
    \Pr{Y_e=1} &=
    \Pr{X_e = 1 \land \forall_{e' \in N_e} X_{e'} = 0} =
    (1-p) \cdot p^{|N_e|} \ge
    (1-p) \cdot p^2
    \enspace.
\end{align*}
This immediately implies that
\begin{align*}
    \Ex{|\mathbb{M}|} &= \sum_{e \in E} \Ex{Y_e} = \sum_{e \in E} \Pr{Y_e = 1} \ge (1-p)p^2 m
    \enspace.
\end{align*}
Observe that $(1-p)p^2$ is maximized for $p=\frac23$ (our choice of $p$), in which case $(1-p)p^2 = \frac{4}{27}$.

Next, we apply Chernoff-Hoeffding bound to obtain
\begin{align*}
    \PPr{|\mathbb{M}| \le \tfrac12\Ex{|\mathbb{M}|}} &\le e^{-\Ex{|\mathbb{M}|}/8}
    \enspace.
\end{align*}
This immediately gives,
\begin{align*}
    \PPr{|\mathbb{M}| > \tfrac{2}{27} m} &=
    \PPr{|\mathbb{M}| > \tfrac12\Ex{|\mathbb{M}|}} \ge
    1 - e^{-m/54}
    \enspace.
\end{align*}
}
\end{proof}

Since \Cref{alg:randomized-large-matching-on-a-ring} relies on simple random sampling of the edges in step 2, (as we will show below) it can be described using random sampling with pairwise independent random variables. As the result, it can be derandomized. With unlimited resources this task is rather straightforward (since it is known how to construct a small sample space for pairwise independent random variables of size $O(n^2)$, see \Cref{thm:k-wise-independent-sample-space}). But in our case the goal is more challenging, since we want to perform the derandomization in a constant number of rounds and with only limited global space of size $O(n)$. Still, the following theorem shows that the task of derandomization of the construction in Claim \ref{thm:randomized-large-matching-on-a-ring} can be performed efficiently.

\begin{ftheorem}
\label{thm:deterministic-large-matching-on-a-ring}
Let $G = (V,E)$ be an undirected simple graph with maximum degree $\Delta \le 2$. One can deterministically find a matching $\mathbb{M}$ of $G$ of size at least $\frac{m}{8} = \Omega(m)$ in a constant number of \MPC rounds with local space $\spac = O(n^{\delta})$ and global space $\gspac = O(n)$.
\end{ftheorem}

\begin{proof}
We use \Cref{alg:randomized-large-matching-on-a-ring} except that in step 2 (which is the only step that uses randomization) we replace the independent random choice for the edges by the choice using pairwise independent random variables\footnote{We present our algorithm using pairwise independent random variables (see \Cref{def:k-wise-independent-rvs}), but we notice that a simpler proof of the quality of the approximation found could be obtained using 3-wise independence (and in that case, the size of the matching found would be increased to $\frac{4m}{27}$). However, for the deterministic \MPC implementation, we felt it is simpler and more instructive to rely on pairwise independence.}. We use the following algorithm with $p = \frac14$; we will show how to transfer it into a deterministic $O(1)$-rounds algorithm on a low-space \MPC.


\bigskip

\begin{algorithm}[H]
\SetAlgoLined
\caption{\small Finding a large matching in graphs with $\Delta \le 2$ using pairwise independence}
\label{alg:deterministic-large-matching-on-a-ring}

\BlankLine
Initialize $\mathbb{M}:= \emptyset$

Let $X_e$ with $e \in E$ be pairwise independent 0-1 random variables with $\Pr{X_e=1} = p$

For each edge $e \in E$, remove $e$ from $E$ if $X_e = 0$

For each edge $\{u,v\}$ still remaining in $E$, if $u$ and $v$ have degree 1 then add $\{u,v\}$ to $\mathbb{M}$
\end{algorithm}

\bigskip


For any edge $e \in E$, let $Y_e$ be the indicator variable that $e \in \mathbb{M}$. Then $|\mathbb{M}| = \sum_{e \in E} Y_e$.

For any edge $e \in E$, let $N_e$ be the set of edges incident to $e$ in $G$; clearly, $|N_e| \le 2$. Notice that $Y_e = 1$ if and only if $X_e = 1$ and $X_{e'} = 0$ for all $e' \in N_e$. Therefore, we obtain the following,
\begin{align*}
    \Pr{Y_e = 0 | X_e = 1} &=
    \Pr{\exists_{e' \in N_e} X_{e'} = 1 | X_e = 1} \le
    \sum_{e' \in N_e} \Pr{X_{e'} = 1 | X_e = 1} \le
    |N_e| \cdot p \le
    2 p
    \enspace,
\end{align*}
where the second inequality follows from the assumption that $X_{e'}$ and $X_e$ are independent (for $e \ne e'$). This gives us the following inequality,
\begin{align*}
    \Pr{Y_e = 1} &=
    \Pr{Y_e = 1 | X_e = 1} \cdot \Pr{X_e = 1} + \Pr{Y_e = 1 | X_e = 0} \cdot \Pr{X_e = 0}
        \\&=
    \Pr{Y_e = 1 | X_e = 1} \cdot \Pr{X_e = 1} =
    (1 - \Pr{Y_e = 0 | X_e = 1}) \cdot \Pr{X_e = 1}
        \\&\ge
    (1-2p)p
    \enspace.
\end{align*}
We will apply this bound with $p = \frac14$, to obtain the following,
\begin{align}
    \label{ineq:bound-for-matching-size}
    \Ex{|\mathbb{M}|} &=
    \sum_{e \in E} \Ex{Y_e} =
    \sum_{e \in E} \Pr{Y_e = 1} \ge
    (1-2p)pm =
    \frac{m}{8}
    \enspace.
\end{align}

Observe that by the probabilistic method, (\ref{ineq:bound-for-matching-size}) implies that for an arbitrary family of pairwise independent hash functions $\HH = \{h: \{1,\dots,n\} \rightarrow \{0,1\}^{\ell}\}$ (where in our case $\ell = 2$, since $p = \frac{1}{2^2} = \frac{1}{2^{\ell}}$, see \Cref{def:k-wise-independent-0-1-rvs}), there is a function $h \in \HH$ for which the matching generated is of size at least $\frac{m}{8}$. (We use here the following mapping between $h$ and random variables $X_e$ with $e \in E$: we arbitrarily order the edges $e_1, \dots, e_m$, and then let $X_{e_i} = 1$ iff $h(i) = 0$.) Therefore our goal now is to fix a family of pairwise independent functions $\HH = \{h: \{1,\dots,n\} \rightarrow \{0,1\}^{\ell}\}$ and find one function $h^* \in \HH$ which corresponds to a matching of size at least $\frac{m}{8}$.

We use the approach sketched in \Cref{sec:conditional-probabilities-on-MPC}. For that we first take a family of pairwise independent functions $\HH = \{h: \{1,\dots,n\} \rightarrow \{0,1\}^2\}$ from \Cref{thm:k-wise-independent-sample-space}. By \Cref{thm:k-wise-independent-sample-space}, \HH has size $O(n^2)$, taking any function $h$ from \HH can be done by specifying $B = 2 \log n + O(1)$ bits, and evaluating a function $h$ from \HH takes $\polylog(n)$ time. We partition the bits describing any function from \HH into blocks of size $\chunk < \log\spac = O(\delta\log n)$ bits each, and hence into $q = \frac{B}{O(\delta\log n)} = O(1/\delta)$ blocks. We will determine $h^* \in \HH$ corresponding to a matching of size at least $\frac{m}{8}$ by determining its bits in blocks, in $q$ phases, by defining a refinement of \HH into $\HH = \HH_0 \supseteq \HH_1 \supseteq \dots \supseteq \HH_q$ such that $|\HH_i| = 2^{\chunk} \cdot |\HH_{i+1}|$ and $|\HH_q| = 1$.

At the beginning of phase $i$, $1 \le i \le q$, we will assume that we have determined the first $(i-1) \chunk$ bits $b_1, \dots, b_{(i-1) \chunk}$. Let
\begin{align*}
    \HH_{i-1} = \{h \in \HH: \text{ bit representation of $h$ in \HH has prefix } b_1, \dots, b_{(i-1) \chunk}\}.
\end{align*}
We will additionally ensure the invariant that the expected size of a matching corresponding to $h \in \HH_{i-1}$ is at least $\frac{m}{8}$. (Observe that by our analysis in (\ref{ineq:bound-for-matching-size}) this invariant holds for $i=1$.)

At the beginning phase $i$, we distribute bits $b_1, \dots, b_{(i-1) \chunk}$ to all machines. Next, while we allow the edges of $G$ to be arbitrarily distributed among the machines, for each edge $e$ stored on some machine, we copy to that machine all (at most two) adjacent edges to $e_r$ (we do not move these edges to the machine storing $e$, but that machine knows their indices). This operation can be easily performed in a constant number of rounds using sorting (see Lemma \ref{lemma:constant-time-basic-primitives}). Next, observe that if a machine stores edges $E' \subseteq E$ then for any $h \in \HH$, the machine can locally compute the number of edges from $E'$ that are in the matching generated by $h$. This follows from the fact that for every $e \in E'$, we only have to check if the random variables generated by $h$ satisfy $X_e = 1$ and $X_{e'} = 0$ for all edges $e'$ adjacent to $e$, and this information is locally available on the machine.

Next, let us consider a single machine, say $M_j$. It contains some edges, call this set $E_j$ (sets $E_j$ over all $j$ define a partition of $E$), knows all edges adjacent to $E_j$, and knows the bits $b_1, \dots, b_{(i-1) \chunk}$. Consider all possible \chunk-bits sequences $\bb = b_{(i-1) \chunk + 1}, b_{(i-1) \chunk + 2}, \dots, b_{i \chunk}$. For each such \bb, let $\HH_{i-1}^{\bb}$ be the set of all hash functions $h \in \HH_{i-1}$ whose bit representation in \HH has prefix $b_1, \dots, b_{(i-1) \chunk}$ followed by \bb. (Notice that sets $\HH_{i-1}^{\bb}$ over all \bb define a partition of $\HH_{i-1}$.) Next, we compute $\mu_j^{\bb}$, which we use to denote the expected number of edges from $E_j$ in the matching generated by hash functions from $\HH_{i-1}^{\bb}$. We can do this by taking each function $h \in \HH_{i-1}^{\bb}$ and compute the number of the edges in $E_j$ in the matching generated by $h$ using the approach above. Then, once we know the value for all individual $h \in \HH_{i-1}^{\bb}$, $\mu_j^{\bb}$ is the average of all these numbers.

Observe that the computation of all $\mu_j^{\bb}$ is performed locally on individual machines, and hence can be performed within a single \MPC round.

Next, apply Lemma \ref{lemma:colored-summation} to compute $\mu^{\bb} = \sum_j \mu_j^{\bb}$ for all \chunk-bits sequences \bb. (Observe that this is an instance of the colored summation with $\machines \cdot 2^{\chunk} = O(\gspac)$ objects, where the ``colors'' correspond to all \chunk-bits sequences \bb.) Notice that $\mu^{\bb}$ is the expected size of a matching (for the entire graph) corresponding to $h \in \HH_{i-1}^{\bb}$. By simple averaging arguments, since the expected size of a matching corresponding to $h \in \HH_{i-1}$ is at least $\frac{m}{8}$, there is at least one \chunk-bits sequences \bb with $\mu^{\bb} \ge \frac{m}{8}$. Therefore we can complete the phase and select $\HH_i$ to be any $\HH_{i-1}^{\bb}$ with $\mu^{\bb} \ge \frac{m}{8}$ (ties are broken arbitrarily, but to be consistent among all machines, we can choose the smallest number \bb which maximizes $\mu^{\bb}$).

In this way, after $q$ phases, each phase consisting of a constant number of rounds, we find $\HH_q$ which has a single element $h \in \HH$ for which the size of a matching corresponding to $h$ is at least $\frac{m}{8}$. This is the hash function $h^*$ sought. Since $q = O(1/\delta) = O(1)$, this concludes the proof of \Cref{thm:deterministic-large-matching-on-a-ring}.
\end{proof}


\section{Derandomization of the random leader contraction process}
\label{sec:derandomization-leader-election}

In this section we consider the following formalization of the task required in the random leader contraction process, as used in our deterministic \MPC connectivity algorithm:
\begin{compactitem}
\item Let $b$ and $n$ be integer with $b \le n$. Let $S_1, S_2, \dots, S_n$ be subsets of $\{1,\dots,n\}$ with $|S_i| = b$ and $i \in S_i$, for every $1 \le i \le n$. Our goal is to find a \emph{small} subset $\mathcal{L} \subseteq \{1,\dots, n\}$ (set of ``leaders'') such that
$S_i \cap \mathcal{L} \ne \emptyset$ for every $i \in \{1,\dots, n\}$.\footnote{The intuition: every vertex $i$ has a leader from $\mathcal{L}$ in its own set $S_i$ and then $i$ can be contracted to such a leader vertex, reducing the size of the vertex set from $n$ to $|\mathcal{L}|$. Thus our goal is to find a small $\mathcal{L}$ with the properties above.}
\end{compactitem}

\noindent We will be assuming that $b \ge \log^{10}n$.

(To build more intuitions, we leave it to the reader to verify that our problem is an instance of the set cover problem and can be also seen as an instance of the dominating set problem in graphs.)

\paragraph{The approach:}
We will construct set $\mathcal{L}$ by observing 0-1 variables $X_1, \dots, X_n$, with the meaning that $i \in \mathcal{L}$ if and only if $X_i = 1$. We will present a construction of set $\mathcal{L}$ first by applying some random sampling approach, and then we will show how the approach can be derandomized.

Let us begin with the standard random sampling based algorithm for this task:


\bigskip

\begin{algorithm}[H]
\SetAlgoLined
\caption{\small Random sampling algorithm for random leader contraction}
\label{alg:randomized-dominating-set-in-dense-graphs}

\BlankLine
Let $X_1, \dots, X_n$ be identically distributed 0-1 random variables from distribution $\DISTRnp$

Return $\mathcal{L} := \{ i \in \{1, \dots, n\}: X_i = 1\}$
\end{algorithm}

\bigskip


It is not difficult to see (see, e.g., \cite[Lemma~III.1]{ASSWZ18} and also \cite[Theorem~1.2.2]{AS16} for an existential and tighter result) that if $p \ge \frac{10 \cdot \ln n}{b}$
then with high probability,
\begin{enumerate}[\it (i)]
\item for every $1 \le i \le n$, $\sum_{j \in S_i} X_j > 0$, and
\item $\sum_{i=1}^n X_i \le 2np$,
\end{enumerate}
which implies that if we set $p = \frac{10 \cdot \ln n}{b}$ then random variables $X_1, \dots, X_n$ correspond to a set $\mathcal{L}$ which is the solution to the problem of size at most $\frac{20 n \ln n}{b}$.

This simple construction relies heavily on random sampling and in what follows we will derandomize a construction of set $\mathcal{L}$ with similar properties. We will consider a small $k$-wise \eps-approximately independent family of hash functions \HH for distribution \DISTRnp (see \Cref{def:k-eps-wise-independent-rvs} and \Cref{thm:k-eps-0-1-arbitrary}) and consider the following revision of \Cref{alg:randomized-dominating-set-in-dense-graphs}.


\bigskip

\begin{algorithm}[H]
\SetAlgoLined
\caption{\small ``Random'' leader contraction using $k$-wise \eps-approximate independence}
\label{alg:deterministic-dominating-set-in-dense-graphs}

\BlankLine
Let $X_1, \dots, X_n$ be $k$-wise \eps-approximately independent 0-1 random variables for \DISTRnp

Return $\mathcal{L} := \{ i \in \{1, \dots, n\}: X_i = 1\}$
\end{algorithm}

\bigskip


We will first prove that for an appropriate setting of $p$, $k$, and \eps, if we choose $h$ independently and uniformly at random from \HH, then for the ($k$-wise \eps-approximately independent) random variables $X_1, \dots, X_n$ defined by $X_i = h(i)$, with sufficiently high probability the following holds:
\begin{enumerate}[\it (i)]
\item for every $1 \le i \le n$, $\sum_{j \in S_i} X_j > 0$, and
\item $\sum_{i=1}^n X_i \le \frac{2n}{b^{\beta}}$ for some constant $0 < \beta < 1$.
\end{enumerate}

Then, we demonstrate how to apply this result to transform \Cref{alg:deterministic-dominating-set-in-dense-graphs} using the approach sketched in \Cref{sec:conditional-probabilities-on-MPC} into a constant-rounds \MPC algorithm that deterministically finds set $\mathcal{L}$ with the required properties (\Cref{thm:deterministic-leaders}).


\subsection{Properties of \Cref{alg:deterministic-dominating-set-in-dense-graphs} using $k$-wise \eps-approximately independence}
\label{subsec:alg:deterministic-dominating-set-in-dense-graphs}

We will consider a $k$-wise \eps-approximately independent family of hash functions \HH for distribution \DISTRnp (see \Cref{def:k-eps-wise-independent-rvs} and \Cref{thm:k-eps-0-1-arbitrary}) for $p$, $k$, and \eps to be defined later. Our analysis of \Cref{alg:deterministic-dominating-set-in-dense-graphs} relies on the following lemma which applies tail bounds for $k$-wise \eps-approximate independent random variables (see \Cref{thm:tail-bound-k-eps-wise-approximation-II}). For a proof, see \Cref{sec:proofs:alg:deterministic-dominating-set-in-dense-graphs}.

\begin{flemma}
\label{lemma:aux1-dominating}
Let $c$ be a positive constant and let $\lambda > c$. Let \HH be a $k$-wise \eps-approximately independent family of hash functions for distribution \DISTRnp, with any $p$, even $k \ge 4$, \eps satisfying $0 < p < 1$, $\eps = n^{-\lambda}$, 
$2k \le \sqrt{bp} \le \sqrt{np}$, and $4c\log_{bp}n \le k \le (\lambda - c) \log_{1/p}n$. Let random variables $X_1, \dots, X_n$ be generated\footnote{That is, we select $h \in \HH$ uniformly at random, and then set $X_i = h(i)$ for every $1 \le i \le n$.} at random from \HH.
%
\junk{
Then the following is true:
\begin{align*}
    \PPPr{\sum_{i=1}^n X_i < 2np \text{ and } \sum_{j \in S} X_i > 0 \text{ for every } i \in \{1,2,\dots,n\}} &\ge
    1 - 9 (n+1) \cdot n^{-c}
    \enspace.
\end{align*}
In particular, if $n \ge 10$ then with probability at least $1 - \frac{1}{n^{c-2}}$,
\begin{itemize}
\item for every $1 \le i \le n$, $\sum_{j \in S_i} X_j > 0$, and
\item $\sum_{i=1}^n X_i \le 2np$.
\end{itemize}
}
Then each of the following $n+1$ events hold with probability at least $1 - 9 n^{-c}$:
\begin{itemize}
\item $\sum_{j \in S_i} X_j > 0$ for every $1 \le i \le n$, and
\item $\sum_{i=1}^n X_i \le 2np$.
\end{itemize}
\end{flemma}

Let us polish the assumptions in Lemma \ref{lemma:aux1-dominating} and present them in a more suitable way.

Notice that if we set $p = b^{\frac{c-\lambda}{3c+\lambda}}$, then we will have $4c\log_{bp}n = (\lambda - c) \log_{1/p}n$ (or, if $p \ge b^{\frac{c-\lambda}{3c+\lambda}}$ then $4c\log_{bp}n \le (\lambda - c) \log_{1/p}n$. Therefore, in particular, if in that case $4c\log_{bp}n$ is an integer, then we can set $p = b^{\frac{c-\lambda}{3c+\lambda}}$ and $k = 4c\log_{bp}n = (3c+\lambda)\log_bn$.

This still does not ensure that the assumptions of Lemma \ref{lemma:aux1-dominating} are satisfied, since we still require that $2k \le \sqrt{bp}$. With our choice for $k$, this is the same as that $2(3c+\lambda)\log_bn \le b^{\frac{2c}{3c+\lambda}}$. With $c, \lambda = \Theta(1)$, this approach can work only if $\log_bn = b^{O(1)}$, that is, we require that $b = \log^{\Omega(1)}n$.

Let us summarize this analysis, where we will apply Lemma \ref{lemma:aux1-dominating} with $c = 3$ and $\lambda = 2c = 6$.

\begin{ftheorem}
\label{thm:key-properties-leaders-k-eps}
Let $\log^{10}n \le b \le n$, $k$ be even with $k = 15\log_bn \ge 4$, $\eps = n^{-6}$, and $p = b^{-\frac15}$. Then, if $X_1, \dots, X_n$ are $k$-wise \eps-approximately independent random variables for distribution \DISTRnp,
%
\junk{
then with probability at least $1 - \frac{1}{n}$,
\begin{itemize}
\item for every $1 \le i \le n$, $\sum_{j \in S_i} X_j > 0$, and
\item $\sum_{i=1}^n X_i \le 2np = 2nb^{-\frac15}$.
\end{itemize}
}
Then each of the following $n+1$ events hold with probability at least $1 - 9 n^{-3}$:
\begin{itemize}
\item $\sum_{j \in S_i} X_j > 0$ for every $1 \le i \le n$, and
\item $\sum_{i=1}^n X_i \le 2np = 2nb^{-\frac15}$.
\end{itemize}
\end{ftheorem}

{
\begin{remark}
\label{remark:leader-election-allowing-greater-sets}
Let us observe that while in our problem description we assume that $|S_i| = b$ for every $1 \le i \le n$, clearly the same claim also holds if $|S_i| \ge b$ for every $1 \le i \le n$: just remove all but $b$ elements from each set and apply \Cref{thm:key-properties-leaders-k-eps}.
\end{remark}
}


\subsection{Efficient \MPC implementation of deterministic leader contraction 
}
\label{subsec:MPC:alg:deterministic-dominating-set-in-dense-graphs}

In this section we will show how to leverage on the analysis of \Cref{alg:deterministic-dominating-set-in-dense-graphs} in \Cref{thm:key-properties-leaders-k-eps} and the derandomization framework using small families of $k$-wise \eps-approximate independent hash functions to obtain a deterministic constant rounds \MPC algorithm for leader contraction (\Cref{thm:deterministic-leaders}). Our main focus here is on the case $b < \spac$, though at the end we will consider also the case $b = \Omega(\spac)$.

We use the approach sketched in \Cref{sec:conditional-probabilities-on-MPC} and used in a simpler form earlier, in the proof of \cref{thm:deterministic-large-matching-on-a-ring}. Take a $k$-wise \eps-approximately independent family of hash functions \HH for distribution \DISTRnp from \Cref{thm:k-eps-0-1-arbitrary}, with the following parameters from \Cref{thm:key-properties-leaders-k-eps}: even $k \ge 4$ with $k = \Theta(\log_bn)$, $\eps = n^{-6}$, and $p = b^{-\frac15}$. Observe that by \Cref{thm:k-eps-0-1-arbitrary}, $|\HH| = 2^{O(k + \log(1/\eps) + \log\log n)} = \poly(n)$, taking any function $h$ from \HH can be done by specifying $O(\log n)$ bits, and evaluating a function $h$ from \HH takes $\polylog(n)$ time.

By the probabilistic method, \Cref{thm:key-properties-leaders-k-eps} implies that there is a function $h \in \HH$ for which, if we generate random variables $X_1, \dots, X_n$ by taking a random hash function $h$ from \HH and then defining $X_i = h(i)$, then
\begin{itemize}
\item for every $1 \le i \le n$, $\sum_{j \in S_i} X_j > 0$, and
\item $\sum_{i=1}^n X_i < 2np + 1 = 2nb^{-\frac15} + 1$.
\end{itemize}
Let any hash function $h \in \HH$ satisfying these inequalities above be called \emph{good}. Our goal is to find one good hash function $h^* \in \HH$.

It is easy to find a good $h^*$ if the \MPC has sufficiently many machines (but still, only polynomial in $n$) by checking the inequalities above for all functions in \HH, but we want to do it using optimal global space, that is, $\gspac = O(nb)$. (Note that the input to the problem consists of $n$ sets $S_1, \dots, S_n$ of size $b$ each, of total size $O(nb)$.)

We proceed similarly as in the proof of \cref{thm:deterministic-large-matching-on-a-ring} and partition the bits describing any function from \HH into blocks of size $\chunk < \log\spac = O(\delta\log n)$ bits each, and hence into $q = \frac{\lceil\log_2|\HH|\rceil}{O(\delta\log n)} = \frac{O(\log n)}{O(\delta\log n)} = O(1/\delta)$ blocks. We will determine $h^* \in \HH$ by fixing its bits in blocks, in $q$ phases, by refining \HH into $\HH = \HH_0 \supseteq \HH_1 \supseteq \HH_2 \supseteq \dots \supseteq \HH_q$ with $\HH_t = 2^{\chunk} |\HH_{t+1}|$ for $0 \le t < q$, and $|\HH_q| = 1$, and by ensuring that there is a good $h$ in every $\HH_t$.

We use the approach sketched in \Cref{sec:conditional-probabilities-on-MPC}, but unlike in the proof of \Cref{thm:deterministic-large-matching-on-a-ring}, we need extra care because of the conflict of two types of event: one, for each set $S_i$, aiming to ensure large values of $\sum_{j \in S_i}X_j$, while another requires that $\sum_{i=1}^n X_i$ is small. For that, we will determine the bits in a next block defining $h^*$ using the concept of \emph{pessimistic estimators}. Indeed, we would want to know whether in a subset \HH' of $\HH_t$ currently under consideration there is one good $h^*$, but for that we need to estimate the probability that for a random $h \in \HH'$, for random variables $X_1, \dots, X_n$ defined by $h$, we have for every $1 \le i \le n$, $\sum_{j \in S_i} X_j > 0$, and $\sum_{i=1}^n X_i \le 2nb^{-\frac15}$. But it is difficult to perform this task in a constant number of rounds on an \MPC with low local space. Instead, we use an appropriate estimator for this event.

For any hash function $h \in \HH$ defining 0-1 random variables $X_1, \dots, X_n$ with $X_i = h(i)$ for every $1 \le i \le n$, let us define the following functions:
\begin{itemize}
\item for every $1 \le i \le n$, $TC_i(h) := 1$ if $\sum_{j \in S_i} X_j = 0$ and 0 otherwise, and
\item $L(h) := \sum_{i=1}^n X_i$.
\end{itemize}
Next, we define an aggregation cost function to be
\begin{align}
    \cost(h) &= n \cdot \sum_{i=1}^n TC_i(h) + L(h)
    \enspace.
\end{align}
Our goal is to use function \emph{\cost} as an estimator of the quality of a subset $\HH^*$ of \HH, to determine if $\HH^*$ has many good hash functions. It is easy to see the following properties of our function \cost:

\begin{claim}
For any $h \in \HH$, 
$h$ is good iff $\cost(h) \le 2nb^{-\frac15}$.
\end{claim}

The next claim follows easily from \Cref{thm:key-properties-leaders-k-eps}
\ArturKeep{An easy proof is as follows:
\begin{align*}
    \SEx{h \in \HH}{\cost(h)} &=
    n \cdot \left(\sum_{i=1}^n \SEx{h \in \HH}{TC_i(h)}\right) + \SEx{h \in \HH}{L(h)}
        \\
        &\le
    n \cdot \sum_{i=1}^n \left(0 \cdot \SPPPr{h \in \HH}{\sum_{j \in S_i} X_j > 0} + 1 \cdot \SPPPr{h \in \HH}{\sum_{j \in S_i} X_j = 0} \right) +
        \\
        &\qquad +
    \left(2nb^{-\frac15} \cdot \SPPPr{h \in \HH}{\sum_{i=1}^n X_i \le 2nb^{-\frac15}} + n \cdot \SPPPr{h \in \HH}{\sum_{i=1}^n X_i > 2nb^{-\frac15}}\right)
        \\
        &\le
    n \cdot \sum_{i=1}^n \SPPPr{h \in \HH}{\sum_{j \in S_i} X_j = 0} + 2nb^{-\frac15} + n \cdot \SPPPr{h \in \HH}{\sum_{i=1}^n X_i > 2nb^{-\frac15}}
        \\
        &\le
    n \cdot \sum_{i=1}^n 9n^{-3} + 2nb^{-\frac15} + n \cdot 9n^{-3}
        \\
        &<
    2nb^{-\frac15} + 1
    \enspace.
\end{align*}
}

\begin{claim}
\label{claim:bound-for-E-cost(g)}
If $n \ge 10$ then $\SEx{h \in \HH}{\cost(h)} < 2nb^{-\frac15} + 1$.
\end{claim}

\junk
{
\begin{proof}
\revised{\begin{align*}
    \SEx{h \in \HH}{\cost(h)} &=
    n \cdot \left(\sum_{i=1}^n \SEx{h \in \HH}{TC_i(h)}\right) + \SEx{h \in \HH}{L(h)}
        \\
        &\le
    n \cdot \sum_{i=1}^n \left(0 \cdot \SPPPr{h \in \HH}{\sum_{j \in S_i} X_j > 0} + 1 \cdot \SPPPr{h \in \HH}{\sum_{j \in S_i} X_j = 0} \right) +
        \\
        &\qquad +
    \left(2nb^{-\frac15} \cdot \SPPPr{h \in \HH}{\sum_{i=1}^n X_i \le 2nb^{-\frac15}} + n \cdot \SPPPr{h \in \HH}{\sum_{i=1}^n X_i > 2nb^{-\frac15}}\right)
        \\
        &\le
    n \cdot \sum_{i=1}^n \SPPPr{h \in \HH}{\sum_{j \in S_i} X_j = 0} + 2nb^{-\frac15} + n \cdot \SPPPr{h \in \HH}{\sum_{i=1}^n X_i > 2nb^{-\frac15}}
        \\
        &\le
    n \cdot \sum_{i=1}^n 9n^{-3} + 2nb^{-\frac15} + n \cdot 9n^{-3}
        \\
        &<
    2nb^{-\frac15} + 1
    \enspace.
\end{align*}
}
\end{proof}
}

With Claim \ref{claim:bound-for-E-cost(g)}, now the idea is simple. We will be searching for a good $h \in \HH$ by refining \HH into smaller and smaller parts $\HH = \HH_0 \supseteq \HH_1 \supseteq \HH_2 \supseteq \dots$, as we did earlier in the proof of \cref{thm:deterministic-large-matching-on-a-ring}. Suppose that at some moment we consider hash functions from $\HH^* \subseteq \HH$. Let $\HH^*_1, \dots, \HH^*_q$ be a partition of $\HH^*$. Our key observation is that even though in a constant number of rounds we cannot deterministically compute the probability that a randomly chosen function $h$ from $\HH^*_i$ is good, we can efficiently compute $\SEx{h \in \HH^*_i}{\cost(h)}$. We use this observation to argue that if initially $\SEx{h \in \HH^*}{\cost(h)} < 2nb^{-\frac15} + 1$, then there must be one index $1 \le i \le q$ such that $\SEx{h \in \HH^*_i}{\cost(h)} < 2nb^{-\frac15} + 1$. Therefore, in our partition of $\HH^*$ we will take such $\HH^*_i$ for which $\SEx{h \in \HH^*_i}{\cost(h)} < 2nb^{-\frac15} + 1$, to ensure that there is a hash function $h \in \HH^*_i$ which is good.

Now we are ready to follow the approach from the proof of \cref{thm:deterministic-large-matching-on-a-ring}. We will define a refinement of \HH into $\HH = \HH_0 \supseteq \HH_1 \supseteq \HH_2 \supseteq \dots \supseteq \HH_q$ with $\HH_t = 2^{\chunk} |\HH_{t+1}|$ for $0 \le t < q$, and $|\HH_q| = 1$, where $\HH_t$ is defined in phase $t$, $1 \le t \le q$.

At the beginning of phase $t$, $1 \le t \le q$, we will assume that we have determined the first $(t-1) \chunk$ bits $b_1, \dots, b_{(t-1) \chunk}$ defining any function from \HH. Let
\begin{align*}
    \HH_{t-1} = \{h \in \HH: \text{ bit representation of $h$ in \HH has prefix } b_1, \dots, b_{(t-1) \chunk}\}.
\end{align*}
We will additionally ensure the \emph{invariant} that at the beginning of phase $t$, $1 \le t \le q$, we have $\SEx{h \in \HH_{t-1}}{\cost(h)} < 2nb^{-\frac15} + 1$. Observe that by \Cref{thm:key-properties-leaders-k-eps}, this invariant holds for $t=1$.

At the beginning of phase $t$, we distribute bits $b_1, \dots, b_{(t-1) \chunk}$ to all machines. Next, we assume (and we can ensure this in a constant number of rounds using the methods from \Cref{sec:basic-primitives}) that the elements $\{1,\dots,n\}$ are arbitrarily assigned to the machines in a balance way, the same number to each machine, and that a machine to which element $i$ is assigned has also copies of all elements in $S_i$.
{
(We will later comment how to deal with the case $b \ge \frac{\spac}{2}$.)
}

Next, let us consider a single machine, say $M_j$. It contains some elements $\EPS_j$ from $\{1,\dots,n\}$ and their corresponding sets $S_i$ with $i \in \EPS_j$, and it knows the bits $b_1, \dots, b_{(t-1) \chunk}$. Consider all possible \chunk-bits sequences $\bb = b_{(t-1) \chunk + 1}, b_{(t-1) \chunk + 2}, \dots, b_{t \chunk}$. For each such \bb, let $\HH_{t-1}^{\bb}$ be the set of all hash functions $h \in \HH_{t-1}$ whose bit representation in \HH has prefix $b_1, \dots, b_{(t-1) \chunk}$ followed by \bb. (Notice that sets $\HH_{t-1}^{\bb}$ over all \bb define a partition of $\HH_{t-1}$.) Next, we want to compute $\SEx{h \in \HH_{t-1}^{\bb}}{\cost(h)}$. For that, for each \bb and $h \in \HH_{t-1}^{\bb}$ machine $M_j$ computes locally two numbers:
\begin{itemize}
\item $TC^{\langle j \rangle}(h)$, to be equal to the number of $i \in \EPS_j$ with $\sum_{r \in S_i}X_r = 0$, and
\item $L^{\langle j \rangle}(h) := \sum_{i \in \EPS_j} X_i$.
\end{itemize}

Observe that in this setting, while no single machine can compute $\cost(h)$ on its own, for every $h \in \HH_{t-1}$ we have $\cost(h) = n \cdot \sum_j TC^{\langle j \rangle}(h) + \sum_j L^{\langle j \rangle}(h)$ and therefore by exchanging the appropriate numbers between the machines one can compute $\cost(h)$ for any single $h$. Unfortunately, we cannot compute all values of $\cost(h)$ since there are too many hash functions $h \in \HH_{t-1}$ and aggregating the information needed to compute all these values requires too much communication on an \MPC. However, we can efficiently compute $\sum_{h \in \HH_{t-1}^{\bb}} \cost(h)$ for all \bb, and this immediately allows us to obtain $\SEx{h \in \HH_{t-1}^{\bb}}{\cost(h)}$, since $\SEx{h \in \HH_{t-1}^{\bb}}{\cost(h)} = \frac{1}{|\HH_{t-1}^{\bb}|} \cdot \sum_{h \in \HH_{t-1}^{\bb}} \cost(h)$.

In order to compute $\sum_{h \in \HH_{t-1}^{\bb}} \cost(h)$ for all \bb, first, each machine $M_j$ computes locally $TC^{\langle j \rangle}(h)$ and $L^{\langle j \rangle}(h)$ for all $h \in \HH_{t-1}$, and then uses these number to compute, for every \bb,
\begin{align*}
    \Tcost_j(\bb) &=
    \sum_{h \in \HH_{t-1}^{\bb}}
        \left(n \cdot TC^{\langle j \rangle}(h) + L^{\langle j \rangle}(h) \right)
    \enspace.
\end{align*}
Next, we aggregate these numbers across all machines and use the colored summation problem (see Lemma \ref{lemma:colored-summation}) to compute $\sum_j \Tcost_j(\bb)$ for all \bb. Then, we distribute all $\sum_j \Tcost_j(\bb)$ with all \bb to all machines, after which each machine can locally compute $\SEx{h \in \HH_{t-1}^{\bb}}{\cost(h)}$ since, as we mentioned above, $\SEx{h \in \HH_{t-1}^{\bb}}{\cost(h)} = \frac{1}{|\HH_{t-1}^{\bb}|} \cdot \sum_{h \in \HH_{t-1}^{\bb}} \cost(h)$.

By simple averaging arguments (the probabilistic method, \Cref{sec:conditional-probabilities-on-MPC}), there is at least one \chunk-bits sequences \bb with $\SEx{h \in \HH_{t-1}^{\bb}}{\cost(h)} \le \SEx{h \in \HH_{t-1}}{\cost(h)}$ and therefore we can complete the phase by selecting $\HH_t$ to be any $\HH_{t-1}^{\bb}$ with $\SEx{h \in \HH_{t-1}^{\bb}}{\cost(h)} \le \SEx{h \in \HH_{t-1}}{\cost(h)}$. (As before, ties are broken arbitrarily; e.g., in order to be consistent among all machines, we can choose the smallest number \bb which minimizes $\SEx{h \in \HH_{t-1}^{\bb}}{\cost(h)}$.)

{Let us now remark (as we promised earlier) what can be done when $b = \Omega(\spac) \equiv \Omega(n^{\delta})$. In that case, we achieve a slightly weaker bound, but fully sufficient for our use: we will obtain $\sum_{i=1}^n X_i \le O(n(\frac{\spac}{2})^{-\frac15})$ (instead of $\sum_{i=1}^n X_i \le O(nb^{-\frac15})$). If $b > \frac{\spac}{2}$ then we take arbitrary $\frac{\spac}{2}$ element from each set $S_i$ and consider the problem with all sets of size $\frac{\spac}{2}$ (see also Remark \ref{remark:leader-election-allowing-greater-sets}). Then, for the modified sets $S_i$ we apply the same as above, since we have enough local space to fit each set $S_i$ on a single machine. Therefore, the arguments above suffice to construct 0-1 variables $X_1, \dots, X_n$ such that for every $1 \le i \le n$, $\sum_{j \in S_i} X_j > 0$, and $\sum_{i=1}^n X_i < 2n(\frac{\spac}{2})^{-\frac15} + 1$.}


{In this way, after $q$ phases, each phase consisting of a constant number of rounds, we find $\HH_q$ which has a single element $h \in \HH$ for which $\cost(h) \le \SEx{h \in \HH}{\cost(h)} < 2n(\min\{b,\frac{\spac}{2}\})^{-\frac15} + 1$. This is the hash function $h^*$ sought: it generates a set $\mathcal{L} \subseteq \{1,\dots, n\}$ with $i \in \mathcal{L}$ iff $h^*(i)=1$ such that
\begin{inparaenum}[\it (i)]
\item $S_i \cap \mathcal{L} \ne \emptyset$ for every $1 \le i \le n$ and
\item $|\mathcal{L}| < 2n(\min\{b,\frac{\spac}{2}\})^{-\frac15} + 1$.
\end{inparaenum}
Since $q = O(1/\delta) = O(1)$, this proves the following theorem.}

\begin{ftheorem}
\label{thm:deterministic-leaders}
Let $b$ and be $n$ be integer with $\log^{10} n \le b \le n$. Let $S_1, S_2, \dots, S_n$ be subsets of $\{1,\dots,n\}$ with $|S_i| = b$ and $i \in S_i$, for every $1 \le i \le n$. Then one can deterministically find a subset $\mathcal{L} \subseteq \{1,\dots, n\}$ with $|\mathcal{L}| \le O(n \cdot (\min\{b,\spac\})^{-\frac15})$ such that $S_i \cap \mathcal{L} \ne \emptyset$ for every $i \in \{1,\dots, n\}$ in a constant number of \MPC rounds with local space $\spac = O(n^{\delta})$ and global space $\gspac = O(nb)$.
\end{ftheorem}

\begin{remark}
A reader may notice a discrepancy between the existential result in \Cref{thm:key-properties-leaders-k-eps} and the result in \Cref{thm:deterministic-leaders}, in that our \MPC algorithm achieves a slightly weaker bound for the size of $\mathcal{L}$, $|\mathcal{L}| \le O(n(\min\{b,\spac\})^{-\frac15}$ instead of $|\mathcal{L}| \le O(nb^{-\frac15})$. We observe however that from the point of view of our applications, both results are equally good: since $\spac = O(n^{\delta})$, both results imply $|\mathcal{L}| \le O(nb^{-\Theta(1)})$, which is what is required in our analysis in the next sections.
\end{remark}


\section{Deterministic connectivity in \MPC}
\label{sec:deterministic-connectivity}

In this section we show the main result of this paper and demonstrate how to incorporate our deterministic \MPC algorithms for approximate matching and vertex number reduction from \Cref{sec:reducing-vertices-in-connectivity-algorithms}, and our derandomized leader contraction process from \Cref{sec:derandomization-leader-election} (see \Cref{thm:deterministic-leaders}), to derandomize the \MPC algorithms for connectivity with low local space due to Andoni \etal \cite{ASSWZ18} and Behnezhad \etal \cite{BDELM19}, as sketched in \Cref{sec:intuitions-connectivity-algorithms}.


\subsection{Deterministic $O(\log D \cdot \log\log_{m/n} n)$ connectivity}
\label{sec:deterministic-connectivity-Andoni}


Consider the connectivity algorithm of Andoni \etal \cite{ASSWZ18} (see also \Cref{sec:intuitions-connectivity-algorithms}). In each phase: there are $n_i$ nodes; each node has its degree increased to at least $b$; the random leader contraction picks a set of $\widetilde{\Theta}(n_i/b)$ leaders (where $b = (\frac{m}{n_i})^{1/2}$) and then each node is contracted to some leader in its component. We first note that as a result of \Cref{alg:deterministic-connectivity->smaller-graph}, we can reduce the number of vertices deterministically from $n$ to $n/\polylog (n)$ in $O(\log\log n)$ rounds. This means that the precondition of \Cref{thm:deterministic-leaders}, that $b > \log^{10} n$, is easily satisfiable in that many rounds.

We now modify the algorithm from \cite{ASSWZ18} in the following way. In each phase, we increase the degrees of all vertices in the graph to at least $b = (\frac{m}{n_i})^{1/2}$, as before (keeping in mind that if a vertex is in a connected component of size at most $b$ then we detected its connected component in this step, and so we drop it from any further analysis). We then use our deterministic leader contraction algorithm above (\Cref{thm:deterministic-leaders}) to select and contract into leaders in such a way as to reduce the number of vertices, deterministically in a constant number of rounds, from $n_i$ to $O(\frac{n_i}{(\min\{b,\spac\})^{1/5}})$.
This slightly weaker bound on the size of the set of leaders (comparing to $\widetilde{O}(n_i/b)$ in \cite{ASSWZ18}) does not asymptotically affect the running time, and with $\spac = n^{\Omega(1)}$, after $O(\log\log_{m/n} n)$ phases we can reach $b = O(n_i)$, after which the algorithm correctly determines all connected components.
%
\ArturKeep{This more modest guarantee on the number of leaders selected does not asymptotically affect the running time. To see why, note that $b$ increases throughout the algorithm (from $(\frac{m}{n})^{1/2}$ to $O(n)$), and consider the two ``periods'' of the algorithm: first when $b < \spac$, and then when $b > \spac$. We observe that the algorithm takes $O(\log\log_{m/n} n)$ phases to reach the second period where $b > \spac$ (this is because, as also pointed out in \cite[Remark~I.11]{ASSWZ18}, a reduction in the number of vertices from $n$ to $(\frac{m}{n}^{\Omega(1)})$ in each phase is sufficient for $b$ to reach $O(n)$ in $O(\log\log_{m/n} n)$ phases: $b$ exceeding $O(\spac)$ must take at most that many rounds). Once we reach $b > \spac$, we instead increase $b$ by a factor of $O(\spac^\frac15)$ each round until $b = O(n)$. This second period of the algorithm takes $O(\frac{1}{\delta})$ phases to terminate (recall that $\spac = O(n^\delta)$, and we consider $\delta$ a constant). The algorithm as a whole therefore takes $O(\log\log_{m/n} n) + O(\frac{1}{\delta}) = O(\log\log_{m/n} n)$ phases to complete: the same as in the connectivity algorithm of \cite{ASSWZ18}.\Sam{I want to reword this again, but this has all the requisite information now I think.}}
%
%
Therefore the random leader contraction phase in \cite{ASSWZ18} can be substituted for our own from \Cref{thm:deterministic-leaders} (as long as we reduce the number of vertices by a factor of $\polylog (n)$ beforehand, recalling \Cref{thm:deterministic-connectivity->smaller-graph-log} and Remark~\ref{remark:loglog-base}) to obtain the following:

\begin{ftheorem}
\label{thm:deterministic-connectivity-Andoni}
Let $G$ be an undirected graph with $n$ vertices, $m$ edges, and diameter $D$. In $O(\log D \cdot \log\log_{m/n} n)$ rounds on an \MPC (with $\spac = O(n^\delta)$ local space and $\gspac = O(m+n)$ global space) one can deterministically identify the connected components of $G$.
\end{ftheorem}

While the main focus of this paper is on the linear global space regime, let us also mention that the algorithm of Andoni \etal \cite[Theorem~I.2]{ASSWZ18} achieves a lower complexity when available global space increases. Since the only randomized step of the algorithm is leader contraction, we can use our deterministic leader contraction to obtain the following straightforward extension.

\begin{ftheorem}
\label{thm:deterministic-connectivity-Andoni-large-global-space}
Let $0 \le \gamma \le 2$ be arbitrary. Let $G$ be an undirected graph with $n$ vertices, $m$ edges, and each connected component having diameter at most $D$. One can deterministically identify the connected components of $G$ in $O\left(\min\{\log D \cdot \log(\frac{\log n}{2 + \gamma \log n}), \log n\}\right)$ rounds on an \MPC with $\spac = O(n^\delta)$ local space and $\gspac = O((m+n)^{1+\gamma})$ global space.
\end{ftheorem}

Observe that when $\gamma = \Omega(1)$ then \Cref{thm:deterministic-connectivity-Andoni-large-global-space} gives an asymptotically optimal round complexity of $O(\log D)$ rounds (conditioned on the 1-vs-2 cycles Conjecture \ref{conjecture-1-vs-2-cycles}, see Theorems \ref{thm:lb-in-connectivity-algorithms} and \ref{thm:lb-in-connectivity-algorithms-rev}).


\subsection{Deterministic $O(\log D + \log\log_{m/n} n)$ connectivity}
\label{sec:deterministic-connectivity-logD+loglogn}

For randomized algorithms, Behnezhad \etal \cite{BDELM19} extended the connectivity algorithm of Andoni \etal \cite{ASSWZ18} (derandomized in \Cref{sec:deterministic-connectivity-Andoni}) and improved the complexity to $O(\log D + \log\log_{m/n} n)$ \MPC rounds, with high probability. In this section, we show how to combine the approach of Behnezhad \etal \cite{BDELM19} with our deterministic leader contraction algorithm from \Cref{thm:deterministic-leaders} to get a deterministic $O(\log D + \log\log_{m/n} n)$-rounds \MPC connectivity algorithm.

Similarly to our analysis of the connectivity algorithm of Andoni \etal \cite{ASSWZ18} in \Cref{sec:deterministic-connectivity-Andoni}, we observe that all randomized steps in the algorithm of Behnezhad \etal \cite{BDELM19} can be reduced to the two procedures derandomized earlier in \Cref{thm:deterministic-leaders} and \Cref{thm:deterministic-connectivity->smaller-graph}.

The first step of the randomized algorithm due to Behnezhad \etal \cite{BDELM19} is to ensure that $m \ge n \log^{10} n$ in $O(\log\log_{m/n} n)$ rounds by contracting vertices if necessary; by our analysis in \Cref{thm:deterministic-connectivity->smaller-graph-log} (see also Remark \ref{remark:loglog-base}) this step can be implemented deterministically in $O(\log\log_{m/n}n)$ rounds on an \MPC with $\spac = O(n^\delta)$ local space and $\gspac = O(m+n)$ global space. The remaining task is to determine connected components in a graph with $n$ vertices and $m$ edges with $m \ge n \log^{10} n$ in $O(\log D + \log\log_{m/n}n)$ rounds on an \MPC with $\spac = O(n^\delta)$ and $\gspac = O(m)$.

As mentioned in \Cref{sec:intuitions-connectivity-algorithms}, the main idea behind the speed up in \cite{BDELM19} is to interleave the operations of degree expansions and vertex contractions for vertices of various degrees in the same time, while maintaining the total global space used bounded by \gspac. For that, we quantify individual space available for any single vertex $u$ by assigning to all vertices their \emph{levels} $\ell(u)$ ($0 \le \ell(u) \le O(\log\log_{m/n} n)$) and associated \emph{budgets} $b(u)$ (budget $b(u)$ controls how much space vertex $u$ can use; in \cite{BDELM19} the budget for vertices of level $i$ is $\beta_i = \beta_0^{1.25^i}$, where $\beta_0 = (\frac{m}{n})^{1/2}$, though the analysis works for $\beta_i = (\beta_0)^{c^i}$ for any constant $c > 1$). Then, informally, Behnezhad \etal \cite{BDELM19} showed how to ensure that every vertex $u$ in a constant number of \MPC rounds either increases its level (and hence its budget) or learns its entire 2-hop neighborhood. Since a vertex can increase its level only $O(\log\log_{m/n} n)$ times, once the budget of a vertex reaches $n$, the vertex is able to store all vertices in its allocated local space to determine its connected component. Similarly, learning a 2-hop neighborhood allows to perform edge contractions to (very informally) halve the diameter. With some additional arguments (see \cite{BDELM19,BDELM19a} and also \cite{LTZ20}), this leads to an $O(\log D + \log\log_{m/n}n)$-rounds \MPC connectivity algorithm.

The algorithm due to Behnezhad \etal \cite{BDELM19} runs in a loop the subroutines \ConnectTwoHops (which, informally, performs a round of expansion for every vertex $v$ if $v$ has room to do it), \RelabelInterLevel (which contracts vertices of lower levels into neighboring vertices of higher levels), and \RelabelIntraLevel (which performs leader contraction, for some vertices which have filled up their budgets), until every connected components forms a clique. We remark that \ConnectTwoHops and \RelabelInterLevel are both deterministic (whenever either subroutine makes a choice it may be arbitrary), and so we only consider derandomizing \RelabelIntraLevel (see \Cref{alg:RelabelIntraLevel-Behnezhad-etal} and \cite[p.~1620]{BDELM19}).


\bigskip

\begin{algorithm}[H]
\SetAlgoLined
\caption{\small \RelabelIntraLevel$(G,b,\ell)$ (Randomized algorithm from \cite[p.~1620]{BDELM19})}
\label{alg:RelabelIntraLevel-Behnezhad-etal}
\KwIn{\small $b(u)$ denotes the \emph{budget} and $\ell(u)$ the \emph{level} of vertex $u$}

\BlankLine
Mark an active vertex $v$ as \emph{``saturated''} if it has at least $b(v)$ active neighbors that have the same level as $v$.

If an active vertex $v$ has a neighbor $u$ with $\ell(u) = \ell(v)$ that is marked as saturated, $v$ is also marked as saturated.

Mark every saturated vertex $v$ as a \emph{``leader''} independently with probability $\min\{\frac{3 \log n}{b(v)}, 1\}$.

For every leader vertex $v$, set $\ell(v) := \ell(v) + 1$ and $b(v) := b(v)^{1.25}$.

Every non-leader saturated vertex $v$ that sees a leader vertex $u$ of the same level (i.e., $\ell(u) = \ell(v)$) in its 2-hop (i.e., $\text{dist}(v,u) \le 2$), chooses one as its leader arbitrarily.

Every vertex is contracted to its leader. That is, for any non-leader vertex $v$ with leader $u$, every edge $\{v,w\}$ is replaced with an edge $\{u,w\}$. Then remove vertex $v$ from the graph.

Remove duplicate edges or self-loops and remove saturated/leader flags from the vertices.
\end{algorithm}

\bigskip


\Cref{alg:RelabelIntraLevel-Behnezhad-etal} \RelabelIntraLevel works solely with a subset of vertices that we call \emph{active} and it is the only procedure that increases the budgets and the levels.

In \RelabelIntraLevel, the only randomized step is Step 3, which takes some subset of vertices and marks each vertex $v$ as a ``leader'' independently at random with probability $\min\{\frac{3 \log n}{b(v)}, 1\}$. All selected leaders are then be promoted to the next level ($\ell(v) := \ell(v)+1$ in Step 4), and all neighbors of a leader at the same level will be contracted to one (arbitrary) such a leader (Steps 5--6). The entire analysis of the algorithm of Behnezhad \etal \cite{BDELM19} relies on this randomized process.

For our analysis it is useful to observe that the ``leader contraction'' processes in \RelabelIntraLevel as described above is \emph{independent for each level}: it is easy to see that two vertices which do not start \RelabelIntraLevel with the same level do not interact with each other (edges between them are irrelevant to the execution of \RelabelIntraLevel).

There are two principal technical challenges to consider when derandomizing this process. Firstly, our deterministic leader contraction process for a given level $i$ only reduces the number of vertices from $N_i$ to $O(\frac{N_i}{(\min\{\beta_i,\spac\})^{1/5}})$ (\cite{BDELM19} reduces to $O(\frac{N_i \log n}{\beta_i})$), and unlike the connectivity algorithm of \cite{ASSWZ18}, we can no longer rely on the ``phase'' structure (and the framework of double exponential speed problem size reduction) to overcome this. Secondly, vertices of multiple levels might be performing leader contractions simultaneously, and we need to ensure that we can perform these leader contractions in parallel and in linear global space. We address these issues in the following two lemmas (see \Cref{sec:proof:deterministic-connectivity-logD+loglogn} for the proofs).

\begin{lemma}
\label{lemma:changing-constants-in-behnezhad}
Let $S_i$ denote the set of saturated vertices at level $i$ after Step 2 of \RelabelIntraLevel, let $L_i$ denote the set of selected leaders at level $i$ after Step 3 of the same execution of \RelabelIntraLevel, let $\beta_i$ denote the budget of vertices at level $i$, let $b(v)$ denote the budget of vertex $v$, and let $\gamma$ be an arbitrary constant such that $0 < \gamma \le 1$.
If we make the following modifications to \RelabelIntraLevel:
\begin{itemize}
\item set $\beta_{i+1} := \beta_i \cdot (\min\{\beta_i,\spac\})^{\gamma/4}$,
\item replace Step 3 of \RelabelIntraLevel with any \MPC algorithm that in $O(1)$ rounds selects $|L_i| = O(\frac{|S_i|}{(\min\{\beta_i,\spac\})^\gamma})$ leaders for each level $i$ with high probability or deterministically, and
\item replace the budget update rule in Step 4 of \RelabelIntraLevel with $b(v) := b(v) \cdot (\min\{b(v),\spac\})^{\gamma/4}$,
\end{itemize}
then the connectivity algorithm of \cite{BDELM19} remains correct with the same asymptotic complexity.
\end{lemma}

\begin{lemma}
\label{lemma:RelabelIntraLevel-number-of-leaders}
Let $S_i$ denote the set of saturated vertices at level $i$ after Step 2 of an execution of  \RelabelIntraLevel (in \cite{BDELM19}), and let $\beta_i$ denote the budget of vertices at level $i$.

Copies of \Cref{alg:deterministic-dominating-set-in-dense-graphs} can be run in parallel, for each possible level and in a constant number of \MPC rounds, to deterministically select $O(|S_i| \cdot (\min\{\beta_i,\spac\})^{-\frac15})$ leaders for each level.
\end{lemma}

By combining Lemmas \ref{lemma:changing-constants-in-behnezhad} and \ref{lemma:RelabelIntraLevel-number-of-leaders}, we can conclude that using a deterministic leader contraction algorithm which selects $O(\frac{|S_i|}{\min\{\beta_i,\spac\}^\gamma})$ leaders for a constant $0 < \gamma \le 1$, for all levels, in a constant number of rounds, does not affect the proof of correctness or the asymptotic running time of the connectivity algorithm of \cite{BDELM19}. Further, by \Cref{thm:deterministic-leaders}, \Cref{alg:deterministic-dominating-set-in-dense-graphs} can be run in parallel for all levels in a constant number of rounds and be implemented on \MPC for $\gamma = \frac15$. We combine these with our deterministic \Cref{alg:deterministic-connectivity->smaller-graph} for reducing the number of vertices by a constant fraction in a constant number of rounds, and recalling Remark \ref{remark:loglog-base} we obtain the following.


\begin{ftheorem}
\label{thm:deterministic-connectivity-logD+loglogn}
Let $G$ be an undirected graph with $n$ vertices, $m$ edges, and diameter $D$. In $O(\log D + \log\log_{m/n} n)$ rounds on an \MPC (with $\spac = O(n^\delta)$ local space and $\gspac=O(m+n)$ global space) one can deterministically identify the connected components of $G$.
\end{ftheorem}

\ArturKeep{Let us also notice that as in \cite{BDELM19}, our algorithm does not require prior knowledge of $D$.}

Similarly as in \Cref{thm:deterministic-connectivity-Andoni-large-global-space}, one can easily (using the same approach as above) incorporate our derandomization framework to the algorithm of Behnezhad \etal \cite{BDELM19} with superlinear global space to obtain the following theorem.

\begin{ftheorem}
\label{thm:deterministic-connectivity-Behnezhad-large-global-space}
Let $0 \le \gamma \le 2$ be arbitrary. Let $G$ be an undirected graph with $n$ vertices, $m$ edges, and each connected component having diameter at most $D$. One can deterministically identify the connected components of $G$ in $O(\log D + \log(\frac{\log n}{2 + \gamma \log n}))$ rounds on an \MPC with $\spac = O(n^\delta)$ local space and $\gspac = O((m+n)^{1+\gamma})$ global space.
\end{ftheorem}

As in \Cref{thm:deterministic-connectivity-Andoni-large-global-space}, for $\gamma = \Omega(1)$ \Cref{thm:deterministic-connectivity-Behnezhad-large-global-space} gives an asymptotically optimal round complexity of $O(\log D)$ rounds (conditioned on the 1-vs-2 cycles Conjecture \ref{conjecture-1-vs-2-cycles}, see Theorems \ref{thm:lb-in-connectivity-algorithms} and \ref{thm:lb-in-connectivity-algorithms-rev}).


\subsection{Implementation details: Proofs of Lemmas \ref{lemma:changing-constants-in-behnezhad} and \ref{lemma:RelabelIntraLevel-number-of-leaders}}
\label{sec:proof:deterministic-connectivity-logD+loglogn}

In this section we complete our analysis of a deterministic version of the connectivity \MPC algorithm due to Behnezhad \etal \cite{BDELM19} in \Cref{thm:deterministic-connectivity-logD+loglogn} and prove our auxiliary Lemmas \ref{lemma:changing-constants-in-behnezhad} and \ref{lemma:RelabelIntraLevel-number-of-leaders}.

\begin{proofof}{Lemma \ref{lemma:changing-constants-in-behnezhad}}
For brevity, we consider all claims and lemmas in the proof of correctness where either the quantity of leaders selected or the budget update rule are relevant.
The analysis is identical to the analysis in \cite{BDELM19} except that we have to modify the arguments to prove the following three key properties:
\begin{enumerate}[(a)]
\item\label{prop-a} the level $\ell(v)$ of any vertex $v$ never exceeds $O(\log\log_{m/n}n)$ (cf. \cite[Lemma~15]{BDELM19}),
\item\label{prop-b} the total space bound does not exceed $O(\gspac)$ (cf. \cite[Lemma~17]{BDELM19}), and
\item\label{prop-c} the sum of the squares of the budgets does not exceed $O(\gspac)$ (cf. \cite[Lemma~21]{BDELM19}).
\end{enumerate}

We will prove all these three properties in our setting.


\begin{enumerate}[(a)]
\item To see property (\ref{prop-a}), observe that while $b(v) < \spac$, we have $b(v) = (\frac{m}{n})^{(1+\gamma/4)^{\ell(v)}}$ and once $b(v) > \spac$, the update rule changes to $b(v) := b(v) \spac^{\gamma/4}$. Recalling that $\spac = O(n^\delta)$ for a constant $\delta$, there are $O(\log\log_{m/n} n)$ levels until $b(v) > \spac$, and then $O(\frac{1}{\gamma \delta})$ more levels before $b(v) = \Omega(n)$. Therefore a vertex can have at most $O(\log\log_{m/n} n) + O(\frac{1}{\gamma\delta}) = O(\log\log_{m/n} n)$ levels.


\item To see property (\ref{prop-b}), as in the proof of \cite[Lemma~17]{BDELM19}), it suffices to show that for every $0 \le i \le O(\log\log_{m/n} n)$ we have $\beta_{i+1} \cdot n_{i+1} \le \beta_i \cdot n_i$, where $n_i$ denotes the number of vertices which \emph{ever} reach level $i$ over the course of the algorithm.  We rely on the ideas from \cite{BDELM19}.

Recall that $\beta_{i+1} = \beta_i \cdot (\min\{\beta_i, \spac\})^{\gamma/4}$. Any leaders of level $i$ are selected only if there are at least $|S_i| = \beta_i$ saturated vertices. Since $|L_i| = O(\frac{|S_i|}{\min\{\beta_i,\spac\}^\gamma})$ and $\beta_i \ge \polylog(n)$, this gives us:\ArturKeep{Notice that $\Omega((\min\{\beta_i, \spac\})^{\gamma}) \gg (\min\{\beta_i, \spac\})^{\gamma/2}$ follows from our assumption that $m/n \ge \polylog(n)$.}
\begin{align*}
    \frac{|S_i \setminus L_i|}{|L_i|} &=
    \frac{|S_i|-|L_i|}{|L_i|} =
    \frac{|S_i|}{|L_i|} - 1 =
    \Omega((\min\{\beta_i, \spac\})^{\gamma}) \gg
    (\min\{\beta_i, \spac\})^{\gamma/2}
\end{align*}
vertices removed from the graph per leader. Since vertices of level $i$ marked as saturated but not marked as leaders never reach level $i+1$, we get $n_{i+1} < n_i \cdot (\min\{\beta_i, \spac\})^{-\gamma/2}$. Hence,
\begin{align*}
    \beta_{i+1} n_{i+1}
        &=
    \left(\beta_i (\min\{\beta_i, \spac\})^{\gamma/4}\right) n_{i+1}
        <
    \left(\beta_i (\min\{\beta_i, \spac\})^{\gamma/4}\right) \left(n_i (\min\{\beta_i, \spac\})^{-\gamma/2}\right)
        \le
    \beta_i n_i
\end{align*}
as required.
Once we have proven that $\beta_{i+1} \cdot n_{i+1} \le \beta_i \cdot n_i$, the same arguments as those in \cite{BDELM19} imply that the global space is $O(\gspac)$.


\item We prove property (\ref{prop-c}) that at any given moment $\sum_v (b(v))^2= O(\gspac)$ by induction on the number of steps of the algorithm. The base case of the induction remains correct since at the beginning $\sum_{v \in V} (b(v))^2 = n \cdot \beta_0 = n \cdot (\frac{m}{n})^{1/2} = m = O(\gspac)$.

When we increase the budget of a vertex in \RelabelIntraLevel, from $\beta_i$ to $\beta_{i+1}$, we increase the sum of squares of the budgets by at most $(\beta_i \cdot (\min\{\beta_i, \spac\})^{\gamma/4})^2 = (\beta_i)^2 \cdot (\min\{\beta_i, \spac\})^{\gamma/2}$. But recalling our argument from part (\ref{prop-b}), we remove at least $(\min\{\beta_i, \spac\})^{\gamma/2}$ vertices per leader, and so, denoting a set of vertices removed in such a way by $U$, we decrease the sum of the squares of the budgets by at least $|U| \cdot (\beta_i)^2 \ge (\beta_i)^2 \cdot (\min\{\beta_i, \spac\})^{\gamma/2}$. Therefore the sum of the squares of the budgets does not increase overall, and the inductive argument holds in this case.
\end{enumerate}


All relevant aspects of the proof of correctness have been adjusted in comparison to \cite{BDELM19}, and so the algorithm remains correct. The running time is not affected as, by Lemma~15 in \cite{BDELM19}, the number of iterations of  algorithm asymptotically takes the same number of rounds, and the only algorithmic change is replacing a line in \RelabelIntraLevel with an algorithm which takes asymptotically the same number of rounds.
\end{proofof}

\begin{proofof}{Lemma \ref{lemma:RelabelIntraLevel-number-of-leaders}}
First, note that each vertex has a unique level, and when considering neighboring vertices for the purposes of leader contraction we only consider vertices of the same level. Therefore the leader contraction processes for each level are fully independent of each other (no edge or vertex is used by the leader contraction process of more than one level).

Suppose we have $O(m)$ total space on some machines whose sole purpose (without loss of generality) is to perform these leader contractions. This is enough to perform all of the leader contraction algorithms, as for each one we need $O(n_i \cdot \beta_i)$ space (because there are at most $n_i$ vertices which have reached level $i$, and each saturated vertex needs to know of $\beta_i$ other saturated vertices of level $i$ in its component), and from arguments used in the proof of \cite[Lemma~17]{BDELM19} (following on directly from our showing that $\beta_{i+1} \cdot n_{i+1} \le \beta_i \cdot n_i$ for each level $i$ in the proof of Lemma~\ref{lemma:changing-constants-in-behnezhad}), we know that $\sum_{i=1}^{O(\log\log_{m/n} n)} n_i \cdot \beta_i = O(m)$. It is a simple exercise to show that distributing the executions of \Cref{alg:deterministic-dominating-set-in-dense-graphs} among these machines in a constant number of rounds is possible.

We note that the process of selecting leaders from the set of saturated vertices is subtly different to \Cref{alg:deterministic-dominating-set-in-dense-graphs}. Specifically, not every vertex that is saturated and at level $i$ has at least $\beta_i$ neighbors at the same level. Vertices can be marked as saturated if they \emph{have a neighbor (at the same level) that has $\beta_i$ neighbors of the same level}. However, this means that all vertices at level $i$ which are marked as saturated have at least $\beta_i$ vertices at level $i$ in their 2-hop neighborhood.

With this in mind (and recalling that in line 5 of \RelabelIntraLevel vertices are only required to find a leader in their 2-hop neighborhood), it suffices that each vertex of level $i$ is made aware of $\beta_i$ vertices of the same level in its 2-hop neighborhood. This can easily be done in a constant number of rounds (an almost identical procedure is used in \ConnectTwoHops in \cite{BDELM19}), and as argued above, even across all levels, this does not require any extra space. We can then perform \Cref{alg:deterministic-dominating-set-in-dense-graphs} in parallel across all levels.
This completes the proof of Lemma \ref{lemma:RelabelIntraLevel-number-of-leaders}.
\end{proofof}


\section{Lower bound of $\Omega(\log D)$ for connectivity}
\label{sec:lb-in-connectivity-algorithms}

It is known that with sufficiently large local space per machine one can determine all connected components of a graph in a constant number of \MPC rounds, even deterministically (see, e.g., \cite{LMSV11,Nowicki21}). However, we do not expect the same result for \MPC{}s in the $O(n^{\delta})$ local space regime: already for 2-regular graphs the tasks seems to require $\Omega(\log n)$ \MPC rounds, as stated in the 1-vs-2 cycles Conjecture (see \Cref{sec:1-2-cycle}). However, the 1-vs-2 cycles Conjecture only claims that the connectivity problem is hard for graphs with large diameter and in does not exclude the situation that for graphs with $o(n)$ diameter this problem is \emph{much easier}. For example, one could hope that for graphs with low diameter, say, $D = \polylog(n)$ or $D = n^{1 - \Omega(1)}$, constant rounds \MPC algorithm in the $n^{\delta}$ space regime exists. Behnezhad \etal \cite{BDELM19} partially quashed this possibility and proved that unless the 1-vs-2 cycles Conjecture fails, if $D \ge \log^{1+\delta}n$ then the best one can hope for is $\Omega(\log D)$ \MPC rounds.

\begin{theorem}\textbf{\emph{\cite{BDELM19}}}
\label{thm:lb-in-connectivity-algorithms}
Let $\delta$ and $\gamma$ be arbitrary constants $0 < \beta, \gamma < 1$. Let $D \ge \log^{1+\gamma}n$. Consider an \MPC with local space $\spac = O(n^{\beta})$ and global space $\gspac = \Omega(n+m)$. Any \MPC algorithm that with probability at least $1 - \frac{1}{n^3}$ determines each connected component of any given $n$-vertex graph with diameter at most $D$ requires $\Omega(\log D)$ rounds (on an \MPC with local space \spac and global space \gspac), unless there is an \MPC algorithm that with probability at least $1 - \frac{1}{n}$ distinguishes in $o(\log n)$ rounds (on an \MPC with local space \spac and global space \gspac) whether the input graph is an $n$-vertex cycle or consists of two $\frac{n}{2}$-vertex cycles.
\end{theorem}

We use the success probability in \Cref{thm:lb-in-connectivity-algorithms} in a flexible way, and the high probability bounds can be adjusted. Further, let us notice that the \MPC local space and the number of machines are the same for both, the connectivity problem and the one vs. two cycles problem.

Rather surprisingly, the arguments in
\Cref{thm:lb-in-connectivity-algorithms} require that $D$ is large, at least $\log^{1+\Omega(1)}n$. We complement \Cref{thm:lb-in-connectivity-algorithms} by extending the analysis to arbitrarily small values of $D$ as follows.

\begin{ftheorem}
\label{thm:lb-in-connectivity-algorithms-rev}
Let $\delta$ be an arbitrary constant $0 < \delta < 1$ and let $D \le \log^{3/2}n$. Consider an \MPC with local space $\spac = O(n^{\delta})$ and global space $\gspac = \Omega(n+m)$. Any \MPC algorithm that with probability at least $1 - \frac{1}{n^3}$ for any given $n$-vertex graph $G$ correctly determines all connected components of $G$ with diameter at most $D$ requires $\Omega(\log D)$ rounds (on an \MPC with local space \spac and global space \gspac), unless there is an \MPC algorithm that with probability at least $1 - \frac{1}{n}$ distinguishes in $o(\log n)$ rounds (on an \MPC with local space \spac and global space \gspac) if the input graph is an $n$-vertex cycle or consists of two $\frac{n}{2}$-vertex cycles.
\end{ftheorem}

\begin{remark}
Observe that we do not prove exactly the same claim as \cite{BDELM19} in \Cref{thm:lb-in-connectivity-algorithms} for small values of $D$, since we consider a slightly different problem: we require that the algorithm correctly determines all connected components with diameter at most $D$, and hence, the input to the algorithm can have larger diameter; we just ignore the output of larger connected components.
\end{remark}

Our approach follows similar arguments to these in \cite{BDELM19}: we show that if we are given as a black-box an \MPC algorithm that with high probability correctly determines all connected components of $G$ with diameter at most $D$ in $\Omega(\log D)$ rounds (on an \MPC with local space \spac and \machines machines), then we can use such algorithm to distinguish in $o(\log n)$ rounds (on an \MPC with local space \spac and \machines machines) if the input graph is an $n$-vertex cycle or consists of two $\frac{n}{2}$-vertex cycles.

The proof of \Cref{thm:lb-in-connectivity-algorithms-rev} is by contradiction. Suppose that there is an algorithm \ALGs that for a given graph $G$ on $n$ vertices determines with probability at least $1 - \frac{1}{n^3}$ all connected components of $G$ with diameter at most $D$ in $o(\log D)$ \MPC rounds with local space $\spac = n^{\delta}$ and on \machines machines. We will show how to use \ALGs to disprove the 1-vs-2 cycles \Cref{conjecture-1-vs-2-cycles} and design an algorithm that in $o(\log n)$ rounds on an \MPC with local space $\spac = n^{\delta}$ and \machines machines, with probability at least $1 - \frac{1}{n}$ distinguishes whether the input graph is an $n$-vertex cycle or consists of two $\frac{n}{2}$-vertex cycles.

The idea is as follows. Let $G = (V,E)$ either be an $n$-vertex cycle or consists of two disjoint $\frac{n}{2}$-vertex cycles. We will use \ALGs to distinguish between these two cases, when $G$ is an $n$-vertex cycle and when $G$ consists of two disjoint $\frac{n}{2}$-vertex cycles, by repeatedly applying the following procedure (see \Cref{alg:connectivity->lb-sub-algorithm-rev}):
\begin{enumerate}[\it (i)]
\item split each cycle into paths of length \emph{typically} at most $D$ by ``temporary'' removing some edges,
\item determine all vertices on each path (i.e., in each connected component) using \ALGs,
\item contract each path into a single vertex, and
\item reverse the temporary removal of the edges to obtain a smaller representation of the original graph.
\end{enumerate}
The underlying idea is that if we started with graph $G$ on $n$-vertices, then steps \emph{(i)}, \emph{(iii)}, and \emph{(iv)} can be performed in a constant number of \MPC rounds, step \emph{(ii)} can be done in $o(\log D)$ \MPC rounds (using \ALGs), and we reduced the problem to an instance on $\widetilde{O}(n/D^{\Theta(1)})$ vertices. Thus, in short, after $o(\log D)$ \MPC rounds we reduce the size from $n$ to $O(n/D^{\Theta(1)})$, then we repeat this step the entire graph can be stored on a single \MPC machine, in which case the problem can be solved in a single \MPC round. Since this approach will reduce $G$ into a graph on at most $n^{\delta}$ vertices in $O(\log_D n)$ repetitions, this will give an algorithm that can determine whether the original $G$ is an $n$-vertex cycle or consists of two disjoint $\frac{n}{2}$-vertex cycles in $o(\log D) \cdot O(\log_D n) = o(\log n)$ \MPC rounds, contradicting the 1-vs-2 cycles \Cref{conjecture-1-vs-2-cycles}.


\bigskip

\begin{algorithm}[H]
\SetAlgoLined
\caption{\small Reducing length of each cycle}
\label{alg:connectivity->lb-sub-algorithm-rev}
\KwIn{Graph $H$ with maximum degree at most 2; parameters $n$ and $D$}

\BlankLine
Temporarily remove each edge in $H$ with probability $p = \frac{1}{\sqrt{D}}$, obtaining graph $H^*$

Find all connected components of $H^*$ using algorithm \ALGs

Determine the \emph{correctness} of all connected components in $H^*$ returned by \ALGs

Contract each connected component of $H^*$ correctly classified by \ALGs to a single vertex in $H$ and return the obtained graph (including all edges removed in step~1)
\end{algorithm}

\bigskip


The main differences between the approach from \cite{BDELM19} and \Cref{alg:connectivity->lb-sub-algorithm-rev} is that in step~1, we have the probability of dropping an edge $p = \frac{1}{\sqrt{D}}$ independent of $n$, and that (because of the change in step 1) we use \ALGs that will process all connected components of diameter at most $D$, but which also allows as its input some connected components with larger diameter. (The reason the arguments from \cite{BDELM19} were weaker in this step is that the proof wanted to ensure that with high probability \emph{all} components have diameter $O(\log D)$, and this requires some dependency on $n$; we avoid such dependency.) The use of step 3 is to determine all connected components in $H^*$ of diameter at most $D$, which are contract into a single vertex in $H$, but to possibly do nothing to vertices in larger connected components in $H^*$, which are likely to be misclassified by \ALGs. The idea here is as follows: Consider an arbitrary connected component $C$ of $H^*$ and consider the output of \ALGs on $C$. If \ALGs correctly classifies all vertices in $C$ (as being in the same connected component and disjoint from the rest of $H^*$) then all these vertices will be contracted to a single vertex in $H$; otherwise, all edges from $C$ will stay in $H$.

The following claim describes key properties of \Cref{alg:connectivity->lb-sub-algorithm-rev} (see \Cref{sec:proof:claim:connectivity->lb-sub-algorithm-rev} for a proof).

\begin{claim}
\label{claim:connectivity->lb-sub-algorithm-rev}
If the input graph $H$ in \Cref{alg:connectivity->lb-sub-algorithm-rev} has maximum degree at most 2, has $N \le n$ vertices, and each connected component 
is of size greater than $n^{\delta}$, and if $9 < D \le \log^{3/2}n$, then with probability at least $1 - \frac{1}{n^2}$ the following three conditions hold:
\begin{enumerate}[\it (i)]
\item the total number of vertices in all connected components of $H^*$ of diameter greater than $D$ is at most $\frac{N}{\sqrt{D}}$,
\item \Cref{alg:connectivity->lb-sub-algorithm-rev} runs in $o(\log D)$ \MPC rounds,
\item the graph returned by \Cref{alg:connectivity->lb-sub-algorithm-rev} has at most $\frac{3 N}{\sqrt{D}}$ edges.
\end{enumerate}
\end{claim}

With Claim \ref{claim:connectivity->lb-sub-algorithm-rev} at hand, we are now ready to move to the prove of \Cref{thm:lb-in-connectivity-algorithms-rev}.

\begin{proofof}{\Cref{thm:lb-in-connectivity-algorithms-rev}}
The claim is trivial for $D \le 9$ and so let us assume that $D > 9$.

We start with $G$ with $n$ vertices and with the promise that either $G$ is an $n$-vertex cycle or $G$ consists of two $\frac{n}{2}$-vertex cycles. We apply \Cref{alg:connectivity->lb-sub-algorithm-rev} repeatedly to $G$, so that by Claim \ref{claim:connectivity->lb-sub-algorithm-rev}, after $t$ repetitions, with probability at least $1 - \frac{t}{n^2}$,
\begin{enumerate}[\it (i)]
\item the repetitions can be implemented in $o(t \cdot \log D)$ \MPC rounds, and
\item either $G$ was contracted to a graph that has a connected component with at most $n^{\delta}$ vertices,
\item or $G$ was contracted to a graph with at most $n \cdot \left(\frac{3}{\sqrt{D}}\right)^t$ edges and vertices.
\end{enumerate}
This implies that if $D > 9$, then for any $t \ge \frac{(1-\delta) \log n}{\log(\sqrt{D}/3)}$ we have a guarantee (with probability at least $1-\frac{t}{n^2}$) that after at most $t$ repetition the contracted graph fits a single \MPC machine, and therefore then we can determine whether the contracted graph is a single cycle or two cycles, determining whether $G$ is an $n$-vertex cycle or consists of two $\frac{n}{2}$-vertex cycles. Observe that our choice of $t$ ensures that the algorithm runs (with probability at least $1 - \frac{t}{n^2}$) in $o\left(\frac{\log n}{\log D} \cdot \log D\right) \equiv o(\log n)$ rounds on an \MPC with local space $\spac = n^{\delta}$.

This implies that if \ALGs can take any graph on at most $n$ vertices and then determine with probability at least $1 - \frac{1}{n^3}$ all connected components with diameter at most $D$ in $o(\log D)$ \MPC rounds with local space $\spac = n^{\delta}$ and global space \gspac, then the scheme above can use \ALGs to design an \MPC algorithm that in $o(\log n)$ rounds on an \MPC with local space $\spac = n^{\delta}$ and global space \gspac, with probability at least $1 - \frac{1}{n}$ distinguishes whether the input graph is an $n$-vertex cycle or consists of two $\frac{n}{2}$-vertex cycles. This completes the proof of \Cref{thm:lb-in-connectivity-algorithms-rev}.
\end{proofof}

\junk{
\begin{remark}
A similar theorem is proven by Ghaffari \etal \cite{GKU19}, though with two differences: the paper uses a large (but still polynomial in $n$) number of machines and relies on another problem for the reduction, $D$-diameter $s$-$t$ connectivity instead of the connectivity.\Artur{\cite[Lemma IV.1]{GKU19} (which will be presented somewhere in this text too): \sl Suppose that every \MPC algorithm with local memory $n^{\alpha}$ for a constant $\alpha \in (0,1)$ and global memory $\poly(n)$ that can distinguish one $n$-vertex cycle from two $n/2$-vertex cycles requires $\Omega(\log n)$ rounds. Then, any \MPC algorithm with local memory $n^{\alpha}$ for a constant $\alpha \in (0,1)$ and global memory $\poly(n)$ that solves $D$-diameter $s$-$t$ connectivity for $D \le \log^{\gamma}n$ for a constant $\gamma \in (0,1)$ requires $\Omega(\log D)$ rounds.)}
\end{remark}
}


\subsection{Proof of Claim \ref{claim:connectivity->lb-sub-algorithm-rev} 
}
\label{sec:proof:claim:connectivity->lb-sub-algorithm-rev}

In this section we prove Claim \ref{claim:connectivity->lb-sub-algorithm-rev}, describing central properties of \Cref{alg:connectivity->lb-sub-algorithm-rev} in our lower bound arguments from \Cref{thm:lb-in-connectivity-algorithms-rev}.

\begin{proofof}{Claim \ref{claim:connectivity->lb-sub-algorithm-rev}}
The analysis is similar to the proof of \Cref{thm:lb-in-connectivity-algorithms} from \cite{BDELM19} with some notable exceptions.

We observe that since each connected in $H$ is of size greater than $n^{\delta}$, the probability that there is a cycle in $H$ which survives in $H^*$ is negligible, is at most $n^{1-\delta} \cdot (1-p)^{n^{\delta}} \ll \frac{1}{n^3}$. Therefore from now on, we will condition on that and will assume that every connected component in $H^*$ is a path.

Next, notice that \emph{(ii)} follows immediately from the definition of \ALGs and from the fact that since $D$ is small, we can verify for each connected the correctness of the solution found. Indeed, if $D \le \log^{3/2}n$ then all vertices output by \ALGs as belonging to the same connected component can be redistributed to be stored on a single machine and then we can easily check whether they indeed form a connected component.

In order to prove property \emph{(i)}, notice first that the same arguments as in the proof of \Cref{thm:lb-in-connectivity-algorithms} from \cite{BDELM19} show that the expected diameter of any connected component (path) in $H^*$ is $1/p = \sqrt{D}$. However, we may have some longer paths and in fact we cannot ensure that all paths are of length at most~$D$.

Let us first bound the number of connected components in $H^*$. Observe that the number of connected components in $H^*$ is at most the number of connected components in $H$ (which is at most $N/n^{\delta}$) plus the number of edges removed in step 1. Since the number of edges removed in step 1, which we denote by $R$, has binomial distribution with the success probability $p$, we obtain that $\Ex{R} = Np$ and by Chernoff-Hoeffding bound, we obtain
\begin{align*}
    \PPr{R \ge \tfrac32Np} &=
    \PPr{R \ge \tfrac32\Ex{R}} \le
    e^{-\Ex{R}/12} =
    e^{-Np/12} \ll
    \frac{1}{n^3}
    \enspace.
\end{align*}

This implies that with probability at least $1 - \frac{1}{n^3}$, the number of paths (connected components) in $H^*$ is at most $N/n^{\delta} + \frac32Np \le 2Np$. From now on, we will condition on this bound.

Let us consider an arbitrary path in $H^*$ and let $X_i$ be the length of path $i$ in $H^*$. Observe that each $X_i$ has almost geometric distribution\footnote{Notice that if the path corresponds to the suffix of some path in $H$ then possibly the length of that path is determined by the last edge on the corresponding path in $H$, which is why we do no necessarily have $\Pr{X_i = k} = (1-p)^k p$, as with geometrically distributed random variables, but occasionally we may have $\Pr{X_i = k} = (1-p)^k$. Therefore we will only bound the probability by $\Pr{X_i = k} \le (1-p)^k$.} and for every integer $k$ we have $\Pr{X_i = k} \le (1-p)^k$.
\junk{
Further, observe that
\begin{align*}
    \Pr{X_i \ge D} &\le
    (1-p)^D =
    \left(1-\frac{1}{\sqrt{D}}\right)^D <
    e^{-\sqrt{D}}
    \enspace.
\end{align*}
}

For $1 \le i \le 2Np$, let $L_i$ be independent integer random variable such that
\begin{align*}
    \Pr{L_i = k} &=
    \begin{cases}
        (1-p)^k & \text{ if } k \ge D, \\
        0 & \text{ otherwise.}
    \end{cases}
\end{align*}
%

Notice that while the lengths of all the paths $X_i$ are potentially dependent (if a path in $H$ has length $k$ then its subpaths cannot have length larger than $k$), but we can majorize them by \emph{independent random variables}, and each $X_i$ is stochastically dominated by $L_i$, that is, for every $k \in \NN$, we have $\Pr{k \le X_i} \le \Pr{k \le L_i}$. With this notation, we can now stochastically bound the total number of vertices in all connected components of $H^*$ of diameter greater than $D$ to be $\sum_{i=1}^{2Np} L_i$. Let $L = \sum_{i=1}^{2Np} L_i$.

It is easy to compute the expected value of $L$:
\begin{align*}
    \Ex{L} &=
    \sum_{i=1}^{2Np} \Ex{L_i} =
    2Np \sum_{k=D}^{\infty} k (1-p)^k =
    2Np(1-p)^D \sum_{k=0}^{\infty} (D+k) (1-p)^k =
        \\&=
    2Np(1-p)^D \left(
        D \sum_{k=0}^{\infty} (1-p)^k +
        \sum_{k=0}^{\infty} k (1-p)^k
        \right)=
        \\&=
    2Np(1-p)^D \left(
        \frac{D}{p} +
        \frac{1-p}{p^2}
        \right) =
    2N(1-p)^D \left(D + \frac{1-p}{p}\right)
    \enspace.
\end{align*}
\junk{
\begin{align*}
    \Ex{L} &=
    \sum_{i=1}^{2Np} \Ex{L_i} =
    \sum_{i=1}^{2Np} \sum_{k=D}^{\infty} k \cdot \Pr{L_i=k} =
    \sum_{i=1}^{2Np} \sum_{k=0}^{\infty} (D+k) \cdot \Pr{L_i=D+k}
        \\&=
    \sum_{i=1}^{2Np} \left(
        D \cdot \sum_{k=0}^{\infty} \Pr{L_i=D+k} +
        \sum_{k=0}^{\infty} k \cdot \Pr{L_i=D+k}
        \right)
        \\&=
    \sum_{i=1}^{2Np} \left(
        D \cdot \Pr{L_i \ge D} +
        \sum_{k=0}^{\infty} \Pr{L_i \ge D+k}
        \right)
        \\&=
    2Np \left(
        D \cdot \Pr{L_1 \ge D} +
        \sum_{k=0}^{\infty} \Pr{L_1 \ge D+k}
        \right)
        \\&=
    2N(1-p)^D \left(D + \frac{1-p}{p}\right)
    \enspace.
\end{align*}
}
Since $p = \frac{1}{\sqrt{D}}$, we can bound this by $\Ex{L} < 3ND(1-p)^D$.

Next, we claim the following.

\begin{claim}
\label{claim:bounding-total-length-of-long-paths}
$\PPr{L \ge 8ND(1-p)^D} \le 2e^{-2 N (1-p)^D / 3}$.
\end{claim}

\begin{proof}
We follow the standard approach of studying geometric random variables by comparing them to Bernoulli (or with negative binomial distribution) random variables. The idea of the analysis is in two steps. First, we show that the number of $i$s with $L_i > 0$ is at most w.h.p. Then, conditioned on that bound, we upper bound the values of such $L_i$.

Let $Y_i$ be the indicator random variable that $L_i > 0$ (that is, $\Ex{Y_i} = \sum_{k=D}^{\infty}(1-p)^k = \frac{(1-p)^D}{p}$) and let $Y = \sum_{i=1}^{2Np}Y_i$. Observe that $\Ex{Y} = 2N(1-p)^D$. Since $Y_1, \dots, Y_{2Np}$ are independent 0-1 random variables, we can use Chernoff-Hoeffding bound to obtain,
\begin{align*}
    \Pr{Y \ge 2\Ex{Y}} &\le e^{-\Ex{Y}/3}
    \enspace,
\end{align*}
\junk{
Next, notice that since with our choice of $p = \frac{1}{\sqrt{D}}$, we have $(1-p)^D = (1-\frac{1}{\sqrt{D}})^D \ge e^{-2\sqrt{D}}$
%
%
and $\exp(-\frac{2N}{3e^{2\sqrt{D}}}) \le \frac{1}{n^3}$ (since $D \le \log^{3/2}n \le \frac14 \ln^2\left(\frac{2n^{\delta}}{9\ln n}\right) \le \frac14 \ln^2\left(\frac{2N}{9\ln n}\right)$), we obtain\Artur{So far $QQQ \le 2Ne^{-2\sqrt{D}}$.}
\begin{align*}
    \PPr{Y \ge QQQ} &\le
    \Pr{Y \ge 2\Ex{Y}} \le
    e^{-\Ex{Y}/3} \le
    \frac{1}{n^3}
    \enspace.
\end{align*}
}
or equivalently, if we set $\lambda = 4N(1-p)^D$, that
\begin{align*}
    \PPr{Y \ge \lambda} &\le
    e^{-\lambda/6} =
    e^{-2 N (1-p)^D / 3}
    \enspace.
\end{align*}

Now, let us condition on the fact that $Y \le \lambda$. Then, $L = \sum_{i=1}^{2Np}L_i$ is stochastically dominated by the random variable $Q = \sum_{i=1}^{\lambda}(D+Q_i)$, where each $Q_i$ is an integer random variable with geometric distribution with the success probability $p$, that is, for any integer $k \ge 0$, we have $\Pr{Q_i=k} = (1-p)^k p$. Notice that $Q = \lambda D + \sum_{i=1}^{\lambda}Q_i$, and so we only have to bound the sum of $\lambda$ independent random variables with geometric distribution.

Let us observe that the probability that $\sum_{i=1}^{\lambda}Q_i \ge S$ is upper bounded by the probability that $S$ random coin tosses with the success probability $p$ will have at least $\lambda$ successes, that is, that the binomial random variable $\mathbb{B}(S,p)$ is at least $\lambda$. Observe that $\Ex{\mathbb{B}(S,p)} = S p$ and by Chernoff-Hoeffding bound we obtain,
\begin{align*}
    \Pr{\mathbb{B}(S,p) \ge 2\Ex{\mathbb{B}(S,p)}} &\le e^{-\Ex{\mathbb{B}(S,p)}/3}
    \enspace.
\end{align*}
This gives us the following (by setting $S = \frac{\lambda}{2p}$),
\begin{align*}
    \PPr{Q \ge \lambda D + \frac{\lambda}{2p}} &=
    \PPr{\sum_{i=1}^{\lambda}Q_i \ge \frac{\lambda}{2p}} \le
    \PPPr{\mathbb{B}\left(\frac{\lambda}{2p},p\right) \ge \lambda}
        \\&=
    \PPPr{\mathbb{B}\left(\frac{\lambda}{2p},p\right) \ge 2\EEx{\mathbb{B}\left(\frac{\lambda}{2p},p\right)}} \le
    e^{-\Ex{\mathbb{B}(\lambda/(2p),p)}/3} =
    e^{-\lambda/6}
    \enspace.
\end{align*}

Now, notice that with our choice $\lambda = 4 N (1-p)^D$ and $p = \frac{1}{\sqrt{D}}$, we obtain that $\frac{\lambda}{2p} \le \lambda D$ and $4 N e^{-2\sqrt{D}} \le \lambda \le 4 N e^{-\sqrt{D}}$. If we plug this in in the inequality above, then we obtain,
\begin{align*}
    \PPr{Q \ge 8 N D (1-p)^D} &\le
    \PPr{Q \ge 2 \lambda D} \le
    \PPr{Q \ge \lambda D + \frac{\lambda}{2p}} \le
    e^{-\lambda/6} =
    e^{-2 N (1-p)^D / 3}
    \enspace.
\end{align*}

Now we ready to summarize the analysis. We first showed that with probability at most $e^{-\lambda/6} = e^{-2 N (1-p)^D / 3}$ we have that $Y \ge \lambda$, that is, that there are at least $\lambda$ paths longer than (or equal to) $D$ in $H^*$. Then, conditioned on this event, we showed that with probability at most $e^{-\lambda/6} = e^{-2 N (1-p)^D / 3}$ the sum of the lengths of all paths longer than or equal to $D$ is at least $2 \lambda D = 8 N D (1-p)^D$. These two claims together yield the proof of \Cref{claim:bounding-total-length-of-long-paths}.
\end{proof}

Now, we can conclude the analysis of property \emph{(i)} in Claim \ref{claim:connectivity->lb-sub-algorithm-rev}. By our arguments above, the total number of vertices in all connected components of $H^*$ of diameter greater than $D$ is stochastically dominated by the random variable $L$, which in turns can be analyzed using \Cref{claim:bounding-total-length-of-long-paths}, to obtain $\PPr{L \ge 8ND(1-p)^D} \le 2e^{-2 N (1-p)^D / 3}$.

Since with our setting of $p = \frac{1}{\sqrt{D}}$ we have $e^{-2\sqrt{D}} \le (1-\frac{1}{\sqrt{D}})^D \le e^{-\sqrt{D}}$, because we have $e^{-2 N (1-p)^D / 3} \le \exp(-\frac{2N}{3e^{2\sqrt{D}}}) \le \frac{1}{n^3}$ (since $D \le \log^{3/2}n \le \frac14 \ln^2\left(\frac{2n^{\delta}}{9\ln n}\right) \le \frac14 \ln^2\left(\frac{2N}{9\ln n}\right)$), and since $8ND(1-p)^D \le 8NDe^{-\sqrt{D}} \le \frac{N}{\sqrt{D}}$ for $D \ge 9$, we can conclude that the probability that more than $\frac{N}{\sqrt{D}}$ vertices are in connected components of $H^*$ of diameter greater than $D$ is at most $\frac{2}{n^3}$.

In order to analyze property \emph{(iii)}, notice that (similarly as in the proof of \Cref{thm:lb-in-connectivity-algorithms} from \cite{BDELM19})
every edge output by \Cref{alg:connectivity->lb-sub-algorithm-rev} is also an edge removed in step 1 of \Cref{alg:connectivity->lb-sub-algorithm-rev}, except that some edges in $H^*$ will not be contracted and they will stay in the output graph: these are the edges in the path longer than $D$ in $H^*$. However, by property \emph{(i)}, with probability at least $1 - \frac{2}{n^3}$ the total number of such edges is at most $\frac{N}{\sqrt{D}}$. In view of that, the total number of edges in the graph output by \Cref{alg:connectivity->lb-sub-algorithm-rev} is equal to the number of edges removed in step 1 of  \Cref{alg:connectivity->lb-sub-algorithm-rev} plus the total number of edges in all connected components of $H^*$ of diameter greater than $D$. By property \emph{(iii)}, with probability at least $\frac{1}{n^3}$, the second number is upper bounded by $\frac{N}{\sqrt{D}}$. On the hand, the number of edges removed in step 1 of  \Cref{alg:connectivity->lb-sub-algorithm-rev} has binomial distribution $\mathbb{B}(N,p)$ and hence it can be analyzed using Chernoff-Hoeffding bound in the same way as it was done in the proof of \Cref{thm:lb-in-connectivity-algorithms} from \cite{BDELM19}. Such analysis gives that the probability that there are more than $2Np = \frac{2N}{\sqrt{D}}$ edges removed in step 1 in \Cref{alg:connectivity->lb-sub-algorithm-rev} is at most $e^{-Np/3} \le \frac{1}{n^3}$. If we combine these two bounds together, we obtain that the probability that the graph returned by \Cref{alg:connectivity->lb-sub-algorithm-rev} has more than $\frac{3 N}{\sqrt{D}}$ edges is at most $\frac{2}{n^3}$.
\end{proofof}


\section{Implications for spanning forest and minimum spanning forest}
\label{sec:extentions-to-spanning-forest-etc}

Now that we have shown that the important primitive of leader contraction can be derandomized while maintaining the same global space and round complexity, we extend this result to some other fundamental graph problems presented in \cite{ASSWZ18} (we present some sample applications only).%

\begin{ftheorem}
\label{thm:ST}
Let $G$ be an undirected graph with $n$ vertices, $m$ edges, and diameter $D$. In $O(\log D \cdot \log \log_{m/n} n)$ rounds on an \MPC (with $\spac = O(n^\delta)$ local space and $\gspac = O(n+m)$ global space) one can deterministically compute a rooted spanning forest of $G$.

With global space $\gspac = (O(n+m))^{1+\gamma}$ for $0 \le \gamma \le 2$, and the same local space $\spac = O(n^\delta)$, one can reduce the number of rounds to
$O\left(\min\left\{\log D \cdot \log(\frac{\log n}{2 + \gamma \log n}), \log n\right\}\right)$.
\end{ftheorem}

\begin{proof}
This is a simple extension of our techniques from \Cref{thm:deterministic-connectivity-Andoni} (relying on the derandomization results in \Cref{thm:deterministic-connectivity->smaller-graph}, \Cref{thm:deterministic-leaders}) when applied to the algorithm by Andoni \etal \cite{ASSWZ18} (with a full analysis provided in \cite{ASSWZ18a}). The only randomized procedures in the spanning forest algorithm presented in \cite[Algorithm~11]{ASSWZ18a} (or any of the subroutines referenced from there) are the steps which select leaders in the same manner as in the connectivity algorithm in \cite{ASSWZ18}, and therefore we can incorporate our framework above (\Cref{thm:deterministic-connectivity-Andoni}).

For our deterministic leader contraction algorithm to work, we need that $m > n \log^{10}n$, and so we potentially need to run \Cref{alg:deterministic-connectivity->smaller-graph} $O(\log\log_{m/n} n)$ times. We note that these contractions that are performed in \Cref{alg:deterministic-connectivity->smaller-graph} are independent of each other. Because of this, the $O(\log\log_{m/n} n)$ iterations of \Cref{alg:deterministic-connectivity->smaller-graph} can be easily modified to maintain the information which is maintained in \cite[Algorithm~11]{ASSWZ18a} (specifically, the set of edges in the spanning forest, and the map between edges in the current graph and edges in the original graph).

Finally, as the edge orientation algorithm presented in \cite[Algorithm~12]{ASSWZ18a} is deterministic and takes the output of the spanning forest algorithm as input (which we can easily preserve), a \emph{rooted} spanning forest can be computed deterministically in $O(\log D \cdot \log \log_{m/n} n)$ rounds.

The extension to a larger global space uses the same approach but it 
relies on \Cref{thm:deterministic-connectivity-Andoni-large-global-space}.
\end{proof}

Andoni \etal \cite{ASSWZ18} showed that efficient \MPC algorithms for graph connectivity and spanning forest can be used as black boxes to design fast \MPC algorithms for several graph optimization problems, including most notably the minimum spanning forest and bottleneck spanning forest problems (see Theorems I.7 and I.9 in \cite{ASSWZ18}, and a more detailed exposition in Appendix~G in \cite{ASSWZ18a}). Since these procedures contain no randomized procedures other than the connectivity and spanning forest algorithms, we can substitute in our deterministic algorithms for connectivity and spanning forest (Theorems
\ref{thm:deterministic-connectivity-Andoni-large-global-space},
\ref{thm:deterministic-connectivity-Behnezhad-large-global-space}, and \ref{thm:ST}) and obtain the following corollaries.
%
\footnote{In comparison to the results from Theorems I.7 and I.9 in \cite{ASSWZ18}, we simplified the bound by observing that since these problems can be solved deterministically in $O(\log n)$ time on an EREW \PRAM with $O(n+m)$ total work \cite{CHL01}, one can always solve them in $O(\log n)$ \MPC rounds even with linear $\gspac = O(n+m)$ global space.}
\junk{
\begin{inparaenum}[\it (i)]
\item since these problems can be solved in $O(\log n)$ time on a \PRAM, one can always solve them in $O(\log n)$ \MPC rounds;
\item the parameter $\gamma$ can be assumed to be $\Omega(\frac{1}{\log n})$ (since we assume $\gspac = \Omega(n+m)$), in which case the parameter $\gamma' = \frac{\gamma}{2} + \Theta(\frac{1}{\log n})$ becomes $\Theta(\gamma)$; this simplifies the boundary conditions of the bounds in \cite{ASSWZ18}.
\end{inparaenum}}

\begin{fcorollary}
\label{corollary:MSF}
Let $G$ be an undirected graph with weights $w: V \rightarrow \mathbb{Z}$ such that $w(e) \le \poly(n)$ for all $e \in E$. Let $D_{\textsf{\tiny MFS}}$ be the diameter of a minimum spanning forest of $G$. For any $0 \le \gamma \le 2$, one can deterministically find a minimum spanning forest of $G$ in
\begin{align*}
    O\left(\min\left\{\left(\log D_{\textsf{\tiny MFS}} + \log(\tfrac{\log n}{2 + \gamma \log n})\right) \cdot \tfrac{\log n}{2 + \gamma \log n}, \log n\right\}\right)
\end{align*}
rounds on an \MPC with $\spac = O(n^{\delta})$ local space and $\gspac = (O(n+m))^{1+\gamma}$ global space.
\end{fcorollary}

\begin{fcorollary}
\label{corollary:BSF}
Let $G$ be an undirected graph with weights $w: V \rightarrow \mathbb{Z}$ such that $w(e) \le \poly(n)$ for all $e \in E$. Let $D_{\textsf{\tiny MFS}}$ be the diameter of a minimum spanning forest of $G$. For any $0 \le \gamma \le 2$, one can deterministically find a bottleneck spanning forest of $G$ in
\begin{align*}
    O\left(\min\left\{\left(\log D_{\textsf{\tiny MFS}} + \log(\tfrac{\log n}{2 + \gamma \log n})\right) \cdot \log(\tfrac{\log n}{2 + \gamma \log n}), \log n\right\}\right)
\end{align*}
rounds on an \MPC with $\spac = O(n^{\delta})$ local space and $\gspac = (O(n+m))^{1+\gamma}$ global space.
\end{fcorollary}

Observe that when $\gspac = O(m+n)$, that is, when $\gamma \le O(\frac{1}{\log n})$, then the bound in Corollary~\ref{corollary:MSF} is just as good as on a \PRAM, giving only a deterministic $O(\log n)$ rounds \MPC algorithm (bound previously known).
But when the global space is $\gspac = (n+m)^{1+\Omega(1)}$,
then Corollaries \ref{corollary:MSF} and \ref{corollary:BSF} imply that both
these
problems have $O(\log D_{\textsf{\tiny MFS}})$-rounds deterministic \MPC algorithms.


\section{Conclusions}
\label{sec:conclusions}

In this paper we show that the recently developed powerful technique of $o(\log n)$-rounds randomized connectivity \MPC algorithms for graphs with low diameter due to Andoni \etal \cite{ASSWZ18} and Behnezhad \etal \cite{BDELM19} can be efficiently derandomized without any asymptotic loss.

In addition to the main concrete result of this paper,
%
\Cref{thm:deterministic-connectivity-logD+loglogn},
we believe that the main take-home message of this work is that many powerful randomized \MPC algorithms can be efficiently derandomized in the \MPC setting, even with low local space and optimal global space utilization. We have already seen in the past some examples demonstrating the strength of deterministic \MPC, and we hope that our work contributes to the advances in this area, showing that even so fundamental problem as graph connectivity can be solved deterministically as good as the state-of-the-art randomized algorithms. One interesting feature of our work is that the derandomization techniques incorporated in our work are highly non-component-stable. As the result, our work is another example (see also \cite{CDP20a,CDP21a,CDP21b}) showing that the powerful framework of conditional \MPC lower bounds due to Ghaffari \etal \cite{GKU19} (which as for now, is arguably the most general framework of lower bounds known for \MPC algorithms) is not always suitable, and a lower bound (even if only conditioned on the 1-vs-2 cycles \Cref{conjecture-1-vs-2-cycles}) for component-stable \MPC algorithms may not preclude the existence of more efficient non-component-stable \MPC algorithms.

We believe that the techniques developed in our paper are versatile and that our framework of deterministic \MPC algorithms can be applied more broadly in similar settings. For example, in addition to the bounds mentioned in \Cref{sec:extentions-to-spanning-forest-etc}, our framework could also be applied to approximate minimum spanning forest algorithms (see Andoni \etal \cite[Theorem I.8]{ASSWZ18}), but we hope that our approach will find many more applications. In particular, we would not be surprised if a number of fundamental graph problems could be solved (even deterministically) in $\widetilde{O}(\log D)$ \MPC rounds.

\ArturKeep{The bounds in Corollaries \ref{corollary:MSF} and \ref{corollary:BSF} have the complexity logarithmic in $D_{\textsf{\tiny MFS}}$ and not in $D$, and so they are not yet solving these problems in $\widetilde{O}(\log D)$ rounds. Achieving the \MPC complexity of $\widetilde{O}(\log D)$ rounds (first, even for randomized algorithms) is an interesting open question.
}

\ArturKeep{
Let us notice though that in the applications of our framework beyond the basic connectivity problem (see \Cref{sec:extentions-to-spanning-forest-etc}), including the most basic problem of constructing a spanning forest, we do not know how to match the bound for the connectivity problem of $O(\log D + \log\log_{m/n}n)$, with the optimal global space utilization $\gspac = O(n+m)$. We believe that matching this gap is an interesting open problem.
}

Finally, the recent advances in the randomized algorithms for the connectivity in the \MPC model led also to similar results for the classic \PRAM model \cite{LTZ20}, who gave an $O(\log D + \log\log_{m/n} n)$-time randomized algorithm on an ARBITRARY CRCW \PRAM using $O(m)$ processors whp. It is an interesting open problem whether a similar bound can be achieved deterministically. The approach present in our paper seems to require to many resources 
to achieve this task.


\section*{Acknowledgements}

The second author would like to thank Peter Davies and Merav Parter for many fascinating conversations about derandomization of \MPC algorithms and their helpful input. The proof of \Cref{thm:tail-bound-k-eps-wise-approximation-II} has been influenced by discussions with Peter Davies about similar tail bounds for related notions of $k$-wise \eps-approximate independence.


\phantomsection
\addcontentsline{toc}{section}{References}


\newcommand{\Proc}{Proceedings of the~}
\newcommand{\ALENEX}{Workshop on Algorithm Engineering and Experiments (ALENEX)}
\newcommand{\BEATCS}{Bulletin of the European Association for Theoretical Computer Science (BEATCS)}
\newcommand{\CCCG}{Canadian Conference on Computational Geometry (CCCG)}
\newcommand{\CIAC}{Italian Conference on Algorithms and Complexity (CIAC)}
\newcommand{\COCOON}{Annual International Computing Combinatorics Conference (COCOON)}
\newcommand{\COLT}{Annual Conference on Learning Theory (COLT)}
\newcommand{\COMPGEOM}{Annual ACM Symposium on Computational Geometry}
\newcommand{\DCGEOM}{Discrete \& Computational Geometry}
\newcommand{\DISC}{International Symposium on Distributed Computing (DISC)}
\newcommand{\ECCC}{Electronic Colloquium on Computational Complexity (ECCC)}
\newcommand{\ESA}{Annual European Symposium on Algorithms (ESA)}
\newcommand{\FOCS}{IEEE Symposium on Foundations of Computer Science (FOCS)}
\newcommand{\FSTTCS}{Foundations of Software Technology and Theoretical Computer Science (FSTTCS)}
\newcommand{\ICALP}{Annual International Colloquium on Automata, Languages and Programming (ICALP)}
\newcommand{\ICCCN}{IEEE International Conference on Computer Communications and Networks (ICCCN)}
\newcommand{\ICDE}{IEEE International Conference on Data Engineering (ICDE)}
\newcommand{\ICDCS}{International Conference on Distributed Computing Systems (ICDCS)}
\newcommand{\ICML}{International Conference on Machine Learning (ICML)}
\newcommand{\IJCGA}{International Journal of Computational Geometry and Applications}
\newcommand{\INFOCOM}{IEEE INFOCOM}
\newcommand{\IPCO}{International Integer Programming and Combinatorial Optimization Conference (IPCO)}
\newcommand{\IPL}{Information Processing Letters}
\newcommand{\ISAAC}{International Symposium on Algorithms and Computation (ISAAC)}
\newcommand{\ISTCS}{Israel Symposium on Theory of Computing and Systems (ISTCS)}
\newcommand{\JACM}{Journal of the ACM}
\newcommand{\JAlgorithms}{Journal of Algorithms}
\newcommand{\LNCS}{Lecture Notes in Computer Science}
\newcommand{\NIPS}{Conference on Neural Information Processing Systems (NeurIPS)}
\newcommand{\OPODIS}{International Conference on Principles of Distributed Systems (OPODIS)}
\newcommand{\OSDI}{Conference on Symposium on Operating Systems Design \& Implementation (OSDI)}
\newcommand{\PVLDB}{Proceedings of the VLDB Endowment}
\newcommand{\PODS}{ACM SIGMOD Symposium on Principles of Database Systems (PODS)}
\newcommand{\PODC}{ACM Symposium on Principles of Distributed Computing (PODC)}
\newcommand{\RANDOM}{International Workshop on Randomization and Approximation Techniques in Computer Science (RANDOM)}
\newcommand{\RSA}{Random Structures and Algorithms}
\newcommand{\SICOMP}{SIAM Journal on Computing}
\newcommand{\SIJDM}{SIAM Journal on Discrete Mathematics}
\newcommand{\SIROCCO}{International Colloquium on Structural Information and Communication Complexity (SIROCCO)}
\newcommand{\SODA}{Annual ACM-SIAM Symposium on Discrete Algorithms (SODA)}
\newcommand{\SOSA}{Symposium on Simplicity in Algorithms (SOSA)}
\newcommand{\SPAA}{Annual ACM Symposium on Parallel Algorithms and Architectures (SPAA)}
\newcommand{\STACS}{Annual Symposium on Theoretical Aspects of Computer Science (STACS)}
\newcommand{\STOC}{Annual ACM Symposium on Theory of Computing (STOC)}
\newcommand{\SWAT}{Scandinavian Workshop on Algorithm Theory (SWAT)}
\newcommand{\TALG}{ACM Transactions on Algorithms}
\newcommand{\UAI}{Conference on Uncertainty in Artificial Intelligence (UAI)}
\newcommand{\WADS}{Workshop on Algorithms and Data Structures (WADS)}
\newcommand{\TCS}{Theory of Computing Systems}

\COMMENTED{
\renewcommand{\Proc}{Proc.~}
\renewcommand{\Proc}{}
\renewcommand{\BEATCS}{Bull.~EATCS}
\renewcommand{\CIAC}{CIAC}
\renewcommand{\COCOON}{COCOON}
\renewcommand{\COLT}{COLT}
\renewcommand{\COMPGEOM}{SoCG}
\renewcommand{\DISC}{DISC}
\renewcommand{\ESA}{ESA}
\renewcommand{\FOCS}{FOCS}
\renewcommand{\FSTTCS}{FSTTCS}
\renewcommand{\ICALP}{ICALP}
\renewcommand{\IPCO}{IPCO}
\renewcommand{\ISAAC}{ISAAC}
\renewcommand{\ISTCS}{ISTCS}
\renewcommand{\JACM}{JACM}
\renewcommand{\LNCS}{LNCS}
\renewcommand{\OPODIS}{OPODIS}
\renewcommand{\PODC}{PODC}
\renewcommand{\PODS}{PODS}
\renewcommand{\SODA}{SODA}
\renewcommand{\SPAA}{SPAA}
\renewcommand{\STACS}{STACS}
\renewcommand{\STOC}{STOC}
\renewcommand{\SWAT}{SWAT}

\newcommand{\SoCG}{\COMPGEOM}
\newcommand{\GEOMETRY}{\COMPGEOM}

\renewcommand{\Proc}{}
}



\bibliographystyle{alpha}

\bibliography{References}


\appendix

\bigskip
\centerline{\huge \textbf{Appendix}}


\section{Limited independence and small families of hash functions}
\label{sec:limited-independence}

An important tool in our analysis is the concept of approximations of product distributions (distributions that are the product of $n$ independent distributions). In our paper we will consider individual distributions to be identical on $\{0,1\}$, and therefore a product distribution is determined by a probability $p$, $0 \le p \le 1$. The distribution with an arbitrary $p$ will be denoted by \DISTRnp. For simplicity, in this paper we will consider only the case when $1/p = 2^\ell$ for integer $\ell$, in which case one can model \DISTRnp by $n$ variables with the uniform distribution on $\{0,1\}^{\ell}$. In view of that, we will model \DISTRnp by a family of hash functions $\Hl = \{h: \{1,\dots,n\} \rightarrow \{0,1\}^{\ell}\}$. Then choosing a sequence of random 0-1 variables $X_1, \dots, X_n$ according to \DISTRnp is the same as choosing a random hash function $h \in \Hl$ and then generating each variable $X_1, \dots, X_n$ as $X_i = 1$ iff $h(i) = 0$ (we represent here a length-$\ell$ bit sequence by the corresponding integer in $\{0,\dots,2^{\ell}-1\}$).
We will use these two interpretations interchangeably.

In our applications, we will want to approximate the large family of hash function \Hl by a significantly smaller families that approximates \DISTRnp (with $p = 2^{-\ell}$). For that, let us recall standard definitions of pairwise, $k$-wise, and $k$-wise \eps-approximate independent families of hash functions.

\begin{fdefinition}{\rm (\textbf{$k$-wise independence})}
\label{def:k-wise-independent-rvs}
Let $k, n, \ell \in \NN$ with $k \le n$. A family of hash functions $\HH = \{h: \{1,\dots,n\} \rightarrow \{0,1\}^{\ell}\}$ is called \emph{$k$-wise independent} if for all $I \subseteq \{1,\dots,n\}$ with $|I| \le k$, the random variables $h(i)$ with $i \in I$ are independent and uniformly distributed in $\{0,1\}^{\ell}$ when $h$ is chosen randomly from \HH.
If $k=2$ then \HH is called \emph{pairwise independent}.

Random variables sampled from a pairwise independent family of hash functions\footnote{That is, random variables $X_1, \dots, X_n$ such that $X_i = h(i)$ with $h$ selected uniformly at random from \HH.} are called \emph{pairwise independent random variables}.
\end{fdefinition}

Our next notion introduces some flexibility in estimating the probabilities.

\begin{fdefinition}{\rm (\textbf{$k$-wise \eps-approximate independence})}
\label{def:k-eps-wise-independent-rvs}
Let $0 \le \eps \le 1$ and let $k, n, \ell \in \NN$ with $k \le n$. A family of hash functions $\HH = \{h: \{1,\dots,n\} \rightarrow \{0,1\}^{\ell}\}$ is called \emph{$k$-wise \eps-approximately independent} if for any $t \le k$, any distinct $i_1, \dots, i_t \in \{1,\dots,n\}$ and any (not necessarily distinct) $b_1, \dots, b_t \in \{0,1\}^{\ell}$, we have
\begin{align*}
    \left|
        \PPr{h(i_j) = b_j, 1 \le j \le t}
            -
        \left(\frac{1}{2^{\ell}}\right)^t
    \right|
        &\le
        \eps
    \enspace.
\end{align*}
Random variables sampled from a $k$-wise \eps-approximately independent family of hash functions\footnote{That is, random variables $X_1, \dots, X_n$ such that $X_i = h(i)$ with $h$ selected uniformly at random from \HH.} are called \emph{$k$-wise \eps-approximately independent random variables}.
\end{fdefinition}

Observe that for $\eps = 0$, $k$-wise \eps-approximation corresponds to $k$-wise independence.

Since we use the family of hash functions \HH to model probability distribution \DISTRnp over $\{0,1\}^n$, we will also extend the notation in \Cref{def:k-wise-independent-rvs} and \Cref{def:k-eps-wise-independent-rvs} as follows.

\begin{definition}
\label{def:k-wise-independent-0-1-rvs}
Random variables sampled from a pairwise independent family of hash functions $\HH = \{h: \{1,\dots,n\} \rightarrow \{0,1\}^{\ell}\}$ are called \emph{pairwise independent random variables for \DISTRnp} with $p = 2^{-\ell}$.

Similarly, random variables sampled from a $k$-wise \eps-approximately independent family of hash functions $\HH = \{h: \{1,\dots,n\} \rightarrow \{0,1\}^{\ell}\}$ are called \emph{$k$-wise \eps-approximately independent random variables for \DISTRnp} with $p = 2^{-\ell}$.
\end{definition}


\subsection{Small limited independence families of hash functions}
\label{sec:limited-independence-small-families}

The classical references about pairwise and $k$-wise independence are the books by Motwani and Raghavan \cite{MR95} and by Alon and Spencer \cite{AS16}, and a survey by Luby and Wigderson \cite{LW06}.
In computer science the use of limited independence probably originates with the papers of Carter and Wegman \cite{CW79,WC79} introducing universal, strongly universal (i.e., $k$-wise independent), and almost strongly universal (i.e., $k$-wise \eps-approximately independent) families of hash functions. (Some of these concepts have been used earlier in the probability literature and were rediscovered several times in the computer science literature.)
The use of pairwise independence for derandomizing algorithms was pioneered by Luby \cite{Luby86}. Standard tail bounds for $k$-wise independent random variables are from \cite{BR94,SSS95}. 


\subsubsection{Small pairwise independent families of hash functions}
\label{sec:limited-independence-small-families-pairwise}

\begin{ftheorem}\textbf{\emph{\cite{ABI86,CG89}}}
\label{thm:k-wise-independent-sample-space}
For every $n, \ell \in \NN$, one can construct a family of pairwise independent hash functions $\HH = \{h: \{1,\dots,n\} \rightarrow \{0,1\}^{\ell}\}$ such that choosing a uniformly random function $h$ from \HH takes 
at most $2(\ell + \log n) + O(1)$
random bits, and evaluating a function from \HH takes $\poly(\ell, \log n)$ time.
\end{ftheorem}


\subsubsection{Small $k$-wise \eps-approximately independent family of hash functions}
\label{sec:limited-independence-small-families-k-eps}

Naor and Naor \cite{NN93} presented an efficient construction of small $k$-wise \eps-approximately independent sample spaces for 0-1 \emph{uniform} random variables (which corresponds to a small $k$-wise \eps-approximately independent family of hash functions for $\ell=1$), which was later extended by Alon et al.\ \cite{AGHP92}. In our paper we are also interested in random variables with $\ell \gg 1$, for which we use the framework developed by Even \etal \cite{EGLNV98} (see also \cite{CRS00,AGM03,KJS01}).

\begin{ftheorem}\textbf{\emph{\cite[Theorem~1]{EGLNV98}}}
\label{thm:k-eps-0-1-arbitrary}
Let $0 \le \eps \le 1$ and let $k, n, \ell \in \NN$ with $k \le n$. One can construct a $k$-wise \eps-approximately independent family of hash functions $\HH = \{h: \{1,\dots,n\} \rightarrow \{0,1\}^{\ell}\}$ such that choosing a random function $h$ from \HH takes $O(k + \log(1/\eps) + \log\log n)$ random bits. Furthermore, each $h \in \HH$ can be specified and evaluated on any input in $\polylog(n, \ell, 2^k, 1/\eps)$ time and space.
\end{ftheorem}


\subsection{Tail bounds for $k$-wise \eps-approximate independence}
\label{sec:limited-independence-tail-bounds}

In our analysis, we will need some concentration bounds for $k$-wise \eps-approximately independent hash functions for distribution \DISTRnp, as defined in \Cref{thm:k-eps-0-1-arbitrary}. Let us first recall the result due to Bellare and Rompel \cite{BR94}, who applied the $k$-th moment inequality to obtain the following concentration bound for $k$-wise independent random variables (see also \cite{SSS95} for similar bounds).

\begin{theorem}{\rm\textbf{(Tail bound for $k$-wise independent random variables \cite[Lemma~2.3]{BR94})}}
\label{thm:concentration-BR94-II}
Let $k \ge 4$ be an even integer. Suppose that $X_1, \dots, X_n$ are $k$-wise independent random variables taking values in $[0,1]$. Let $X = X_1 + \dots + X_n$ and $\mu = \Ex{X}$, and let $t > 0$. Then
\begin{align*}
    \Pr{|X-\mu| \ge t} &\le
    8 \cdot \left(\frac{k \mu + k^2}{t^2}\right)^{k/2}
    \enspace.
\end{align*}
\end{theorem}

The main idea behind proving this bound is that if $X_1, \dots, X_n$ are $k$-wise independent random variables, then one can compute $\Ex{(X - \Ex{X})^k}$ by using only term involving sets of at most $k$ random variables, for which we can assume full independence.

While the extensions of the framework from \cite{BR94} (and also from a related framework due to Schmidt \etal \cite{SSS95}) to $k$-wise \eps-approximately independent functions (or sample spaces) have been occasionally mentioned in some other works (e.g., Naor and Naor proved it (Theorem~5 in \cite{NN93}) for uniform distribution over $\{-1,1\}^n$), we could not find a complete setting that would suit our analysis. Therefore we present here a suitable extension of the approach due to Bellare and Rompel \cite{BR94} from \Cref{thm:concentration-BR94-II} to $k$-wise \eps-approximately independent functions.

\begin{ftheorem}{\rm\textbf{(Tail bound for $k$-wise \eps-approximate independence)}}
\label{thm:tail-bound-k-eps-wise-approximation-II}
Let $k \ge 4$ be an even integer. Suppose that $X_1, \dots, X_n$ are a $k$-wise \eps-approximately independent 0-1 random variables (with identical marginal distribution). Let $X = X_1 + \dots + X_n$ and $\mu = \Ex{X}$, and let $t > 0$. Then
\begin{align*}
    \Pr{|X-\mu| \ge t} &\le
    8 \cdot \left(\frac{k \mu + k^2}{t^2}\right)^{k/2} +
        \eps \cdot \left(\frac{n}{t}\right)^k
    \enspace.
\end{align*}
\end{ftheorem}

\begin{remark}
Let us emphasize that \Cref{thm:tail-bound-k-eps-wise-approximation-II} \emph{does not} require the probability distribution to be uniform and it works for arbitrary \DISTRnp, for arbitrary $0 \le p \le 1$ (as required in our analysis).
\end{remark}


\subsubsection{Proof of the tail bound for $k$-wise \eps-approximate independence (\Cref{thm:tail-bound-k-eps-wise-approximation-II})}
\label{sec:limited-independence-tail-bounds-proof}

For the sake of completeness, we will prove \Cref{thm:tail-bound-k-eps-wise-approximation-II} here. Notice that since $k$ is even, $|X - \mu|^k = (X - \mu)^k$ and we can drop the absolute values in this expression. Our proof relies on the following lemma.

\junk{
We begin with a simple observation about $k$-wise \eps-approximation.

\begin{lemma}
\label{lem:ExpProd-0-1-rv}
For any sequence $X_1, \dots, X_k$ of $(k, \eps)$-wise approximately independent 0-1 random variables with arbitrary distribution $\mathcal{D}$, for every $\ell \le k$, the following inequality holds:
\begin{align*}
    \left|\EEEx{\prod_{i=1}^{\ell} X_i} - \prod_{i=1}^{\ell} \Ex{X_i}\right| &\le \eps
    \enspace.
\end{align*}
\junk{
\begin{align*}
    \Ex{\prod_{i=1}^k X_i} &\le \eps^{k-1} \cdot \prod_{i=1}^k \Ex{X_i}
    \enspace.
\end{align*}
}
\end{lemma}

\begin{proof}
From the definition of expectation,
\begin{align*}
	\EEx{\prod_{i=1}^{\ell} X_i} &=
    \sum_{(x_1,\dots,x_{\ell}) \in \{0,1\}^{\ell}}
        \left(\Pr{X_i=x_i, 1 \le i \le \ell} \cdot \prod_{i=1}^{\ell} x_i\right)
        =
    \Pr{\bigwedge_{i=1}^{\ell} X_i=1}
    \enspace.
\end{align*}
Now the claim follows from \Cref{def:k-eps-wise-independent-rvs} since $\ell \le k$ and $\Pr{\bigwedge_{i=1}^{\ell} X_i=1} = \prod_{i=1}^{\ell} \Ex{X_i}$.
\junk{
\bigskip\bigskip\bigskip
We prove the claim by induction on $k$. For $k=1$ the claim is trivially true. Assuming the claim is true for $k-1$,
\begin{align*}
	\EEx{\prod_{i=1}^k X_i} &=
    \sum_{(x_1,\dots,x_k) \in \{0,1\}^k}
        \left(\Pr{X_i=x_i, 1 \le i \le k} \cdot \prod_{i=1}^k x_i\right)
        \\
        &=
    \sum_{(x_1,\dots,x_{k-1}) \in \{0,1\}^{k-1}} \sum_{x_k \in \{0,1\}}
        \left(\Pr{X_i=x_i, 1 \le i \le k} \cdot \prod_{i=1}^k x_i\right)
        \\
        &=
    \sum_{(x_1,\dots,x_{k-1}) \in \{0,1\}^{k-1}} \left(\prod_{i=1}^m x_i \cdot
        \sum_{x_k \in \{0,1\}} \left(\Pr{X_i=x_i, 1 \le i \le k}
            \cdot x_k\right)\right)
    \enspace.
\end{align*}

By Bayes' rule, for any $(x_1,\dots,x_{m+1}) \in \Omega^{m+1}$ the following holds,
\begin{align*}
	\Pr{X_i=x_i, 1 \le i \le m+1} &=
    \Pr{X_i=x_i, 1 \le i \le m} \cdot
        \Pr{X_{m+1}=x_{m+1}|X_i=x_i, 1 \le i \le m}
    \enspace.
\end{align*}
	
By $(k,\eps)$-weakly-almost-independence, we have:
\begin{align*}
    \Pr{X_{m+1}=x_{m+1}|X_i=x_i, 1 \le i \le m} &\le
    \eps \cdot \Pr{X_{m+1}=x_{m+1}}
    \enspace,
\end{align*}

and therefore we can combine this bound with the bound above to get,
\begin{align*}
	\Pr{X_i=x_i, 1 \le i \le m+1} &=
    \Pr{X_i=x_i, 1 \le i \le m} \cdot
        \Pr{X_{m+1}=x_{m+1}|X_i=x_i, 1 \le i \le m}
            \\
            &\le
    \eps \cdot \Pr{X_i=x_i, 1 \le i \le m} \cdot \Pr{X_{m+1}=x_{m+1}}
    \enspace.
\end{align*}
	
Hence we obtain the following,
\begin{align*}
	\Ex{\prod_{i=1}^{m+1} X_i}
        &=
    \sum_{(x_1,\dots,x_m) \in \Omega^m} \left(\prod_{i=1}^m x_i \cdot
        \sum_{x_{m+1} \in \Omega} \left(\Pr{X_i=x_i, 1 \le i \le m+1}
            \cdot x_{m+1}\right)\right)
        \\
        &\le
    \sum_{(x_1,\dots,x_m) \in \Omega^m} \left(\prod_{i=1}^m x_i \cdot
        \sum_{x_{m+1} \in \Omega} \left(\eps \cdot \Pr{X_i=x_i, 1 \le i \le m}
            \cdot \Pr{X_{m+1}=x_{m+1}}\right) \cdot x_{m+1}\right)
        \\
        &=
    \eps \cdot \sum_{(x_1,\dots,x_m) \in \Omega^m} \left(\prod_{i=1}^m x_i \cdot
        \sum_{x_{m+1} \in \Omega}
            \left(\Pr{X_i=x_i, 1 \le i \le m} \cdot \Ex{X_{m+1}}\right)\right)
        \\
        &=
    \eps \cdot \Ex{X_{m+1}} \cdot \Ex{\prod_{i=1}^m X_i}
        \\
        &\le
    \eps \cdot \Ex{X_{m+1}} \cdot \eps^{m-1} \cdot \prod_{i=1}^m \Ex{X_i}
        \\
        &=
    \eps^m \cdot \prod_{i=1}^{m+1} \Ex{X_i}
    \enspace,
\end{align*}
by the inductive assumption.
}
\end{proof}
}

\begin{lemma}
\label{lemma:high-moment-indep-approx}
Let \DISTRnp be a joint distribution of $n$ identically distributed 0-1 independent random variables $\hat{X}_1, \dots, \hat{X}_k$ and let $\hat{X} = \sum_{i=1}^n \hat{X}_i$.
Let $0 \le \eps \le 1$ and $k \ge 2$ be an even integer. Let $X_1, \dots, X_n$ be a $k$-wise \eps-approximately independent random variables for \DISTRnp. Let $X = \sum_{i=1}^n X_i$ and $\mu = \Ex{X}$. Then,
\begin{align*}
    \left|
        \EEx{(X-\mu)^k} - \EEx{(\hat{X}-\mu)^k}
    \right|
        &\le
    \eps \cdot n^k
        \enspace.
\end{align*}
\end{lemma}

\begin{proof}
Observe that if $\bar{\mu} = \Ex{X_i}$ for every $1 \le i \le n$, then by linearity of expectation, we obtain,
\begin{align}
    \nonumber
    \EEx{(X-\mu)^k}
        &=
    \EEx{(\sum_{i=1}^n (X_i - \bar{\mu}))^k}
        =
    \EEx{\sum_{i_1, \dots, i_k \in \{1,2,\dots,n\}} \prod_{j=1}^k (X_{i_j} - \bar{\mu})}
        \\
    \label{bound1:in-lemma:high-moment-indep-approx}
        &=
    \sum_{i_1, \dots, i_k \in \{1,2,\dots,n\}}
        \EEx{\prod_{j=1}^k (X_{i_j} - \bar{\mu})}
        \enspace.
\end{align}
Similarly, for independent random variables $\hat{X}_1, \dots, \hat{X}_k$ we also have,
\begin{align}
    \label{bound2:in-lemma:high-moment-indep-approx}
    \EEx{(\hat{X}-\mu)^k}
        &=
    \sum_{i_1, \dots, i_k \in \{1,2,\dots,n\}}
        \EEx{\prod_{j=1}^k (\hat{X}_{i_j} - \bar{\mu})}
        \enspace.
\end{align}
Now, for fixed $i_1, \dots, i_k \in \{1,2,\dots,n\}$, consider $\prod_{j=1}^k (X_{i_j} - \bar{\mu})$. Notice that
\begin{align*}
    \EEx{\prod_{j=1}^k (X_{i_j} - \bar{\mu})} &=
    \sum_{b_1, \dots, b_k \in \{0,1\}} \PPr{X_{i_j} = b_j, 1\le j \le k}
            \cdot \prod_{j=1}^k (b_j - \bar{\mu})
        \enspace, \text{ and}
        \\
    \EEx{\prod_{j=1}^k (\hat{X}_{i_j} - \bar{\mu})} &=
    \sum_{b_1, \dots, b_k \in \{0,1\}}
        \PPr{\hat{X}_{i_j} = b_j, 1 \le j \le k}
            \cdot \prod_{j=1}^k (b_j - \bar{\mu})
        \enspace.
\end{align*}
Therefore, since by \Cref{def:k-eps-wise-independent-rvs} we have
\begin{align*}
    \left|
        \PPr{X_{i_j} = b_j, 1 \le j \le k} -
        \PPr{\hat{X}_{i_j} = b_j, 1 \le j \le k}
    \right|
        &\le
    \eps
    \enspace,
\end{align*}
we obtain the following
\begin{align*}
    \lefteqn{
    \left|
        \EEx{\prod_{j=1}^k (X_{i_j} - \bar{\mu})} -
        \EEx{\prod_{j=1}^k (\hat{X}_{i_j} - \bar{\mu})}
    \right|
    }
        \\
        &=
    \left|
        \sum_{b_1, \dots, b_k \in \{0,1\}}
            \left(\PPr{X_{i_j} = b_j, 1\le j \le k} -
            \PPr{\hat{X}_{i_j} = b_j, 1 \le j \le k} \right)
            \cdot \prod_{j=1}^k (b_j - \bar{\mu})
    \right|
        \\
        &\le
    \sum_{b_1, \dots, b_k \in \{0,1\}}
        \left|
            \left(\PPr{X_{i_j} = b_j, 1\le j \le k} -
            \PPr{\hat{X}_{i_j} = b_j, 1 \le j \le k} \right)
            \cdot \prod_{j=1}^k (b_j - \bar{\mu})
        \right|
        \\
        &\le
    \sum_{b_1, \dots, b_k \in \{0,1\}}
        \left|\eps \cdot \prod_{j=1}^k (b_j - \bar{\mu})\right|
        \enspace.
\end{align*}

Finally, we notice that (as can be easily proven)
\junk{
\footnote{The following is a simple proof of this fact:
\begin{align*}
    \sum_{b_1, \dots, b_N \in \{0,1\}} \left|\prod_{j=1}^N (b_j - \bar{\mu})\right| &=
    \sum_{\mathbf{b} \in \{0,1\}^N}
    (1-\bar{\mu})^{\|\mathbf{b}\|_1} \cdot \bar{\mu}^{N-\|\mathbf{b}\|_1}
        =
    \sum_{\ell=0}^N \binom{N}{\ell} (1-\bar{\mu})^{\ell} \cdot \bar{\mu}^{N-\ell}
        =
    \left((1-\bar{\mu}) + \bar{\mu}\right)^N
        =
    1
    \enspace.
\end{align*}
}
}
\begin{align*}
    \sum_{b_1, \dots, b_k \in \{0,1\}} \left|\prod_{j=1}^k (b_j - \bar{\mu})\right| &=
    1
        \enspace,
\end{align*}
and therefore we can conclude the following,
\begin{align*}
    \left|
        \EEx{\prod_{j=1}^k (X_{i_j} - \bar{\mu})} -
        \EEx{\prod_{j=1}^k (\hat{X}_{i_j} - \bar{\mu})}
    \right|
        &\le
        \sum_{b_1, \dots, b_k \in \{0,1\}}
            \left|\eps \cdot \prod_{j=1}^k (b_j - \bar{\mu})\right|
        =
        \eps
        \enspace.
\end{align*}

If we plug this in (\ref{bound1:in-lemma:high-moment-indep-approx}) and (\ref{bound2:in-lemma:high-moment-indep-approx}), then we obtain
\begin{align*}
    \left|
        \EEx{(X-\mu)^k} - \EEx{(\hat{X}-\mu)^k}
    \right|
        &=
    \left|
        \sum_{i_1, \dots, i_k \in \{1,\dots,n\}}
        \!\!\!\!\!
            \EEx{\prod_{j=1}^k (X_{i_j} - \bar{\mu})}
            -
        \!\!\!\!\!\!\!\!
        \sum_{i_1, \dots, i_k \in \{1,\dots,n\}}
            \EEx{\prod_{j=1}^k (\hat{X}_{i_j} - \bar{\mu})}
    \right|
        \\
        &\le
    \sum_{i_1, \dots, i_k \in \{1,\dots,n\}}
        \left|
            \EEx{\prod_{j=1}^k (X_{i_j} - \bar{\mu})}
            -
            \EEx{\prod_{j=1}^k (\hat{X}_{i_j} - \bar{\mu})}
        \right|
        \\
        &\le
    \eps \cdot n^k
        \enspace.
\end{align*}
\end{proof}

We can now combine Lemma \ref{lemma:high-moment-indep-approx}  with \Cref{thm:concentration-BR94-II} to prove \Cref{thm:tail-bound-k-eps-wise-approximation-II}.

\medskip

\begin{proof}
In order to prove \Cref{thm:tail-bound-k-eps-wise-approximation-II}, we will relate the distribution of random variables $X_1, \dots, X_n$ that are a $k$-wise \eps-approximately independent with respect to $n$ identically distributed \emph{independent} 0-1 random variables $\hat{X}_1, \dots, \hat{X}_k$. Let $\hat{X} = \sum_{i=1}^n \hat{X}_i$ and notice that $\Ex{\hat{X}} = \mu$.

Since $\hat{X}_1, \dots, \hat{X}_k$ are independent 0-1 random variables, we can apply \cite[Lemma~A.4]{BR94} to obtain
\begin{align*}
    \Ex{(\hat{X}-\mu)^k} &\le
    8 \cdot (k \mu + k^2)^{k/2}
        \enspace.
\end{align*}

At the same time, Lemma \ref{lemma:high-moment-indep-approx} yields
\begin{align*}
    \left|
        \EEx{(X-\mu)^k} - \EEx{(\hat{X}-\mu)^k}
    \right|
        &\le
    \eps \cdot n^k
        \enspace.
\end{align*}

Next, we apply the $k$-th moment inequality that for any $t > 0$,
\begin{align*}
    \PPr{|X - \Ex{X}| > t} &\le
    \frac{\Ex{(X - \Ex{X})^k}}{t^k}
\end{align*}
and combine it with the other two bounds to obtain the following,
\begin{align*}
    \PPr{|X - \Ex{X}| > t} &\le
    \frac{\Ex{(X - \Ex{X})^k}}{t^k}
    \\&\le
    \frac{\EEx{(\hat{X}-\mu)^k} + \eps \cdot n^k}{t^k}
    \\&\le
    \frac{C_k \cdot (k \mu + k^2)^{k/2} + \eps \cdot n^k}{t^k}
    \enspace.
\end{align*}
\end{proof}


\section{Proof of Lemma \ref{lemma:aux1-dominating} from \Cref{subsec:alg:deterministic-dominating-set-in-dense-graphs}}
\label{sec:proofs:alg:deterministic-dominating-set-in-dense-graphs}

In this section we prove Lemma \ref{lemma:aux1-dominating}. The proof follows immediately from the following two auxiliary claims, Claims \ref{claim:aux1-dominating} and \ref{claim:aux2-dominating}.


\begin{claim}
\label{claim:aux1-dominating}
Let $c$ be a positive constant and let $\lambda > c$. Let \HH be a $k$-wise \eps-approximately independent family of hash functions for distribution \DISTRnp, with any $p$, even $k \ge 4$, \eps satisfying $0 < p < 1$, $\eps = n^{-\lambda}$, 
$2k \le \sqrt{bp}$, and $4c\log_{bp}n \le k \le (\lambda - c) \log_{1/p}n$. Let random variables $X_1, \dots, X_n$ be generated at random from\footnote{That is, we select $h \in \HH$ uniformly at random, and then set $X_i = h(i)$ for every $1 \le i \le n$.} \HH. Let $S \subseteq \{1,\dots,n\}$ with $|S| = b$. Then the following is true:
\begin{align*}
    \PPPr{\sum_{j \in S} X_i = 0} &\le 9n^{-c}
    \enspace.
\end{align*}
\end{claim}

\begin{proof}
Let $X = \sum_{j \in S} X_i$ and $\mu = \Ex{X} = bp$. Since $k \ge 4$, we have by \Cref{thm:tail-bound-k-eps-wise-approximation-II},
\begin{align*}
    \Pr{|X-\mu| \ge t} &\le
    8 \cdot \left(\frac{k \mu + k^2}{t^2}\right)^{k/2} +
        \eps \cdot \left(\frac{b}{t}\right)^k
    \enspace,
\end{align*}
which (since $k \le \mu$) can be simplified to
\begin{align*}
    \Pr{|X-\mu| \ge t} &\le
    8 \cdot \left(\frac{k \mu + k^2}{t^2}\right)^{k/2} +
        \eps \cdot \left(\frac{b}{t}\right)^k
            \le
    8 \cdot \left(\frac{2 k \mu}{t^2}\right)^{k/2} +
        \eps \cdot \left(\frac{b}{t}\right)^k
    \enspace.
\end{align*}

We will use this bound with $t = \mu = bp$ to obtain,
\begin{align}
\label{ineq1:aux1-dominating}
    \Pr{X = 0} &\le
    \Pr{|X-\mu| \ge \mu} \le
    8 \cdot \left(\frac{2 k \mu}{\mu^2}\right)^{k/2} +
        \eps \cdot \left(\frac{b}{\mu}\right)^k
        =
    8 \cdot \left(\frac{2 k}{bp}\right)^{k/2} +
        \eps \cdot p^{-k}
    \enspace.
\end{align}

Next, since $2k \le \sqrt{bp}$ and $k \ge 4c\log_{bp}n$, we obtain the following,
\begin{align}
\label{ineq2:aux1-dominating}
    \left(\frac{2k}{bp}\right)^{k/2}&\le
    (bp)^{-k/4}\le
    (bp)^{-c \log_{bp}n}=
    n^{-c}
    \enspace.
\end{align}

Similarly, since $0 < p < 1$, $\eps = n^{-\lambda} <  n^{-c}$ and $k \ge (\lambda - c) \log_{1/p}n$, we have the following,
\begin{align}
\label{ineq3:aux1-dominating}
    \eps \cdot (1/p)^k &\le
    n^{-\lambda} \cdot (1/p)^{(\lambda - c) \log_{1/p}n} =
    n^{-\lambda} \cdot n^{\lambda - c} =
    n^{-c}
    \enspace.
\end{align}

Therefore, we can apply (\ref{ineq2:aux1-dominating}) and (\ref{ineq3:aux1-dominating}) to (\ref{ineq1:aux1-dominating}) to conclude Claim \ref{claim:aux1-dominating}:
\begin{align*}
    \Pr{X = 0} &\le
    8 \cdot \left(\frac{2 k}{bp}\right)^{k/2} +
        \eps \cdot p^{-k}
            \le
    9 n^{-c}
    \enspace.
\end{align*}
\end{proof}


\begin{claim}
\label{claim:aux2-dominating}
Let $c$ be a positive constant and let $\lambda > c$. Let \HH be a $k$-wise \eps-approximately independent family of hash functions for distribution \DISTRnp, with any $p$, even $k \ge 4$, \eps satisfying $0 < p < 1$, $\eps = n^{-\lambda}$, 
$2k \le \sqrt{np}$, and $4c\log_{np}n \le k \le (\lambda - c) \log_{1/p}n$. Let random variables $X_1, \dots, X_n$ be generated at random from \HH. Then the following is true:
\begin{align*}
    \PPPr{\sum_{i=1}^n X_i \ge 2np} &\le
    9 n^{-c}
    \enspace.
\end{align*}
\end{claim}

\begin{proof}
Let $X = \sum_{i=1}^n X_i$ and $\mu = \Ex{X} = np$. Since $k \ge 4$, \Cref{thm:tail-bound-k-eps-wise-approximation-II} gives
\begin{align*}
    \Pr{|X-\mu| \ge t} &\le
    8 \cdot \left(\frac{k \mu + k^2}{t^2}\right)^{k/2} +
        \eps \cdot \left(\frac{n}{t}\right)^k
    \enspace,
\end{align*}
which can be simplified (since $k \le \mu$) to
\begin{align*}
    \Pr{|X-\mu| \ge t} &\le
    8 \cdot \left(\frac{k \mu + k^2}{t^2}\right)^{k/2} +
        \eps \cdot \left(\frac{n}{t}\right)^k
            \le
    8 \cdot \left(\frac{2 k \mu}{t^2}\right)^{k/2} +
        \eps \cdot \left(\frac{n}{t}\right)^k
    \enspace.
\end{align*}

We will use this bound with $t = \mu = bp$ to obtain,
\begin{align}
\label{ineq1:aux2-dominating}
    \PPPr{X \ge 2np} &\le
    \PPPr{|X-\mu| \ge np}
        \le
    8 \cdot \left(\frac{2 k n p}{(np)^2}\right)^{k/2} +
        \eps \cdot \left(\frac{n}{np}\right)^k
        =
    8 \cdot \left(\frac{2 k}{np}\right)^{k/2} +
        \eps \cdot p^{-k}
        \enspace.
\end{align}

Next, similarly to the bound (\ref{ineq2:aux1-dominating}), since $2k \le \sqrt{np}$ and $k \ge 4c\log_{np}n$, we have the following,
\begin{align}
\label{ineq2:aux2-dominating}
    \left(\frac{2k}{np}\right)^{k/2}&\le
    (np)^{-k/4}\le
    (np)^{-c \log_{np}n}=
    n^{-c}
    \enspace.
\end{align}

Further, as in (\ref{ineq3:aux1-dominating}), since $0 < p < 1$, $\eps = n^{-\lambda} <  n^{-c}$ and $k \ge (\lambda - c) \log_{1/p}n$, we get
\begin{align}
\label{ineq3:aux2-dominating}
    \eps \cdot p^{-k} &\le
    n^{-\lambda} \cdot (1/p)^{(\lambda - c) \log_{1/p}n} =
    n^{-\lambda} \cdot n^{\lambda - c} =
    n^{-c}
    \enspace.
\end{align}

Therefore we can apply (\ref{ineq2:aux2-dominating}) and (\ref{ineq3:aux2-dominating}) to (\ref{ineq1:aux2-dominating}) to conclude Claim \ref{claim:aux2-dominating}:
\begin{align*}
    \PPPr{X \ge 2np} &\le
    8 \cdot \left(\frac{2 k}{np}\right)^{k/2} +
        \eps \cdot p^{-k}
            \le
    9 n^{-c}
    \enspace.
\end{align*}
\end{proof}


\end{document}